\documentclass[%
 reprint,
nofootinbib,
 amsmath,amssymb,
 aps,
]{revtex4-1}

\newcommand{\grafe}[1]{\left\{ #1 \right\}}
\newcommand{\tonde}[1]{\left( #1 \right)}
\newcommand{\quadre}[1]{\left[ #1 \right]}

\usepackage{xcolor}
\usepackage{graphicx}
\usepackage{dcolumn}
\usepackage{amsthm, amsfonts}
\usepackage{bbm}
\usepackage{comment}

\usepackage{empheq}

\newcommand*\widefbox[1]{\fbox{\hspace{2em}#1\hspace{2em}}}

\DeclareMathOperator{\Tr}{Tr}
\newtheorem{theorem}{Lemma}
\usepackage{bm}
\usepackage{dsfont}

\begin{document}

\preprint{APS/123-QED}

\title{Overlaps between eigenvectors of spiked, correlated random matrices:\\from matrix PCA to random Gaussian landscapes }

\author{Alessandro Pacco}
\email{alessandro.pacco@universite-paris-saclay.fr}
\author{Valentina Ros}%
\email{valentina.ros@universite-paris-saclay.fr}
\affiliation{Université Paris-Saclay, CNRS, LPTMS, 91405, Orsay, France}

\date{\today}

\begin{abstract}

We consider pairs of GOE (Gaussian Orthogonal Ensemble) matrices which are correlated with each others, and subject to additive and multiplicative rank-one perturbations. We focus on the regime of parameters in which the finite-rank perturbations generate outliers in the spectrum of the matrices. We investigate the statistical correlation (i.e., the typical overlap) between the eigenvectors associated to the outlier eigenvalues of each matrix in the pair, as well as the typical overlap between the outlier eigenvector of one matrix with the eigenvectors in the bulk of the spectrum of the other matrix. We discuss implications of these results for the signal recovery problem for spiked matrices, as well as for problems of high-dimensional random landscapes. 

\end{abstract}

\maketitle

\section{\label{sec:intro}Introduction}
Spiked random matrices, \emph{i.e.} random matrices deformed by additive or multiplicative perturbations,  have been the object of investigation since the foundational works of random matrix theory~\cite{jones1978eigenvalue, edwards1976eigenvalue,furedi1981eigenvalues}.  A huge amount of work has been devoted to characterizing the effect of low-rank perturbations on the spectral density of matrices extracted from invariant ensembles, \emph{i.e.} in determining the statistics of the isolated eigenvalues, or outliers, generated by the perturbations (referred to as \emph{spikes}, following a terminology introduced in~\cite{johnstone2001distribution}). For perturbed Wishart matrices, it has been shown by Baik, Ben Arous and P\'{e}ch\'e that in the limit of large matrix size, the outliers pop out from the bulk of the eigenvalues density in a sharp phase transition~\cite{baik2005phase}; with reference to this seminal work, the spectral transitions associated to the emergence of outliers are generically referred to as \emph{BBP transitions}. They  have been discussed extensively in the mathematical literature, see for instance~\cite{peche2006largest,benaych2011eigenvalues, benaych2012singular, capitaine2011free, knowles2014outliers, tao2013outliers,bordenave2016outlier,rochet2016isolated, capitaine2016spectrum}. Deformed random matrices and their outliers have been shown to play a relevant role in a variety of contexts: examples can be found in finance~\cite{bun2017cleaning}, inference and detection problems~\cite{montanari2015limitation, lu2020phase}, constraint satisfaction problems~\cite{hwang2020force}, quantum chaos~\cite{fyodorov2022extreme}, localization of polymers by defects~\cite{ikeda2022bose}, theoretical ecology~\cite{fraboul2021artificial, baron2022non}. The eigenvectors associated to the outliers play a relevant role in these applications: their projection on the subspace spanned by the low-rank perturbations remains large in the limit of large matrix size~\cite{nadler2008finite,benaych2011eigenvalues}, a phenomenon akin to condensation~\cite{kosterlitz1976spherical}.

When interpreted as \emph{signal vs noise} problems (the low-rank perturbations representing the signal), the spiked matrices are a prototypical example of a transition between a phase in which signal recovery is impossible  (because the spectral properties of the deformed matrix are completely determined by the random contribution and are therefore identical to the ones of the unperturbed matrix) and a phase in which information on the signal can be recovered, at least partially. This is done by determining the extremal eigenvalues of the matrix and the associated eigenvectors -- that is, by means of Principal Component Analysis (PCA). In this context, it can also be relevant to determine the projection of bulk eigenvectors on the subspace spanned by the perturbation, and results in this direction are given for instance in~\cite{ledoit2011eigenvectors, noiry2021spectral}.\\
\indent
In this work, we are interested in characterizing the squared overlaps between the eigenvectors of \emph{pairs} of correlated random matrices extracted from a Gaussian Orthogonal Ensemble (GOE), which are deformed by rank-one additive \emph{and} multiplicative perturbations. Our analysis builds on the works \cite{bun2016rotational,bun2018overlaps} (see also the comprehensive discussion in~\cite{potters_bouchaud_2020}), which (among other results) present the explicit expression of the overlaps of eigenvectors of matrices of the form ${\bf H}+ {\bf W}^{(a)}$, where ${\bf H}$ is a (possibly random) matrix in common to both elements of the pairs, while the ${\bf W}^{(a)}$ with $a=0,1$ are independent GOE matrices.  Our generalization consists in deforming the statistics of the matrices along one single direction in the basis space, by means of a combination of rank-1 additive and multiplicative perturbations. In certain parameter regimes, these perturbations generate outliers in the spectra of the pair of matrices: we determine the overlap between the eigenvectors of the outliers, as well as between the eigenvector of the outlier of one matrix and any other eigenvector associated to eigenvalues in the bulk of the other matrix. \\
\indent
Our analysis is motivated by the study of high-dimensional random landscapes: indeed, it can be shown~\cite{subag2017complexity,ros2019complex, ros2020distribution, ros2019complexity} that the local curvature of simple Gaussian landscapes in the vicinity of their stationary points (local minima, maxima or saddles) is described by matrices having exactly the statistics considered in this work. Determining the overlap between eigenvectors of the Hessians (the matrix encoding the information on the  local landscape curvature) is relevant to understand the geometry of the landscape, in particular how the curvature is correlated in different regions of the landscape. This is an important piece of information whenever one is interested in optimizing high-dimensional landscapes, to characterize how geometrical features affect its exploration by means of local optimization algorithms~\cite{ros2021dynamical, rizzo2021path}. As a byproduct, our analysis allows us to address a question that may be of its own interest in the context of spiked matrices problems: namely, it gives us access to the correlations between different estimators of the signal vector,
 obtained from different sets of noisy measurements in which the noise is correlated.  
 
 The work is structured as follows: in Section~\ref{sec:theoretical_res} we introduce our matrix ensemble of interest, and summarize its spectral properties. In Section~\ref{sec:theoretical_res2} we present the results of our calculation, namely the explicit expressions of the eigenvectors overlaps. In Section~\ref{sec:overview_comp} we give an overview of the calculation,  discussing how to obtain the relevant overlaps from the calculation of products of resolvent operators and their finite-size corrections. Section \ref{sec:applications} contains a discussion of applications of our results, while the conclusions are given in Section~\ref{sec:conclusions}. Details of the calculation are given in the Appendices.

\section{Perturbed, coupled  GOE matrices}\label{sec:theoretical_res}
\subsection{The matrix ensembles}
\label{subsec:matrix}
We consider pairs of correlated random matrices with a perturbed GOE (Gaussian Orthogonal Ensemble) statistics. We recall that a GOE matrix of size $N$ is a symmetric real random matrix with off-diagonal entries distributed as $\mathcal{N}(0,\sigma^2/N)$ and diagonal entries as $\mathcal{N}(0,2\sigma^2/N)$, where $\mathcal{N}$ denotes the Gaussian distribution.  In our model of interest, the perturbation is given by a special row and column in each matrix of the pair, whose entries are correlated to each others in a different way. More precisely, let $\mathbf{M}^{(a)}$ with $a=0,1$ be a pair of $N \times N$ matrices with the following block structure:  
\begin{align}\label{eq:MatrixForm}
    \mathbf{M}^{(a)}=
    \begin{pmatrix}
     & && & m^{a}_{1\,N}\\
     &  &{\bf B}^{(a)}  & & \vdots\\
      & & & & m^{a}_{N-1\,N}\\
     m^{a}_{1\,N} && \ldots & m^{a}_{N-1\,N} &m_{N\,N}^{a}
    \end{pmatrix}
\end{align}
\noindent where the ${\bf B}^{(a)}$ are two $N-1 \times N-1$ correlated GOE matrices with components $b_{ij}^a$ having zero mean, and correlations given by:
\begin{equation}
\label{eq:b_correlations}
    \mathbb{E}[b_{ij}^a \, b_{kl}^b]= \tonde{\delta_{a b} \frac{\sigma^2}{N}+ (1-\delta_{ab})\frac{\sigma^2_H}{N}}(\delta_{ik} \delta_{jl}+ \delta_{il} \delta_{jk})
\end{equation}
for $a, b \in \grafe{0,1}$.  
The two GOE matrices ${\bf B}^{(a)} $ have equal variance $ N^{-1}\sigma^2$, and for all $i \leq j$ the component $b_{ij}^0$ is correlated only with $b_{ij}^1$. Similarly, the entries $m^a_{i N}$ for $i<N$ have zero mean and correlations given by:
\begin{equation}
\label{eq:m_correlations}
    \mathbb{E}[m_{iN}^a \, m_{kN}^b]= \tonde{\delta_{a b} \frac{\Delta^2_a}{N}+ (1-\delta_{ab})\frac{\Delta^2_h}{N}}\delta_{ik}
\end{equation}
for $a,b\in\{0,1\}$. Finally, the diagonal entries $m^a_{NN}$ have a non-zero average:
\begin{equation}
\mathbb{E}[m_{NN}^a]= \mu_a,\,\,\quad a\in\{0,1\},
\end{equation}
 and covariances given by:
\begin{equation}
\label{eq:m_NN}
    \mathbb{E}[m_{NN}^a \, m_{NN}^b]-\mu_a \mu_b= \tonde{\delta_{a b} \frac{v^2_a}{N}+ (1-\delta_{ab})\frac{v^2_h}{N}}
\end{equation}
for $a,b\in\{0,1\}$.
The choice of correlations in \eqref{eq:b_correlations} implies that the matrices ${\bf B}^{(0)}, {\bf B}^{(1)}$ can be written as the sum of two GOE matrices:
\begin{equation}
    {\bf B}^{(a)}={\bf H}+ {\bf W}^{(a)},\quad a\in\{0,1\}
\end{equation}
where  ${\bf H}$ is an $N-1\times N-1$ GOE matrix with 
\begin{equation}
\begin{split}
    &\mathbb{E}[H_{ij} H_{kl}]=\frac{\sigma_H^2}{N} (\delta_{ik} \delta_{jl}+ \delta_{il} \delta_{jk}),
\end{split}
\end{equation}
that is in common to both elements of the pair, while  ${\bf W}^{(0)}, {\bf W}^{(1)}$ are $N-1\times N-1$ independent and identically distributed GOE matrices satisfying
\begin{equation}
\begin{split}
\mathbb{E}[W^{a}_{ij} W^{a}_{kl}]=\frac{\sigma_W^2}{N} (\delta_{ik} \delta_{jl}+ \delta_{il} \delta_{jk}),\quad a\in\{0,1\}
\end{split}
\end{equation}
and clearly $\sigma^2= \sigma^2_H + \sigma^2_W$. Thanks to \eqref{eq:m_correlations} and \eqref{eq:m_NN}, the entries belonging to the last row and column admit a similar decomposition in terms of independent random variables, 
\begin{equation}
m^{a}_{iN}= h_{iN} + w_{iN}^a,\quad a\in\{0,1\}
\end{equation}
with $h_{iN}\sim\mathcal{N}(0,N^{-1}\Delta^2_h)$ and $w^a_{iN}\sim\mathcal{N}(0,N^{-1}\Delta^2_{w,a})$ for $i<N$, while  $h_{NN}\sim\mathcal{N}(0,N^{-1} v^2_h)$ and $w^a_{NN}\sim\mathcal{N}(\mu_a,N^{-1} v^2_{w,a})$. Of course, $\Delta^2_a= \Delta^2_h + \Delta^2_{w,a}$ and $v^2_a= v^2_h+ v^2_{w,a}$, for $a=0$ and $a=1$.  

Each matrix of the form \eqref{eq:MatrixForm} can be re-written as a GOE matrix perturbed with both additive and multiplicative rank-one perturbations along one fixed direction identified by the basis vector ${\bf e}_N$ (corresponding to the last row and column). We can indeed write:
\begin{equation}\label{eq:MatRiscritte}
    {\bf M}^{(a)}= {\bf F}^{(a)} \cdot {\bf X}^{(a)} \cdot  {\bf F}^{(a)} +  \tonde{\mu_a + \zeta_a\frac{\xi^a}{\sqrt{N}} } \, {\bf e}_N {\bf e}_N^T
\end{equation}
where ${\bf X}^{(a)}$ is now an $N \times N$ GOE with variance $N^{-1}\sigma^2$, $ \xi^a \sim \mathcal{N}(0,1)$ is an independent standard Gaussian variable, and the terms ${\bf F}^{(a)}$ and $\zeta_a$ are introduced to reproduce the correct variance of the entries belonging to the special row and column of the matrices \eqref{eq:MatrixForm}; more precisely, 
\begin{equation}
    {\bf F}^{(a)}= \mathbbm{1} - \tonde{1-\frac{\Delta_a}{\sigma}}{\bf e}_N {\bf e}_N^T,
\end{equation}
while $\zeta_a= \tonde{v^2_a - \Delta^4_a/\sigma^{2}}^{\frac{1}{2}}$ is chosen in such a way that $m^a_{NN} \sim \mathcal{N}(\mu_a, N^{-1} v^2_a)$ is recovered. The matrix ${\bf F}^{(a)}$ represents a deterministic, multiplicative perturbation to the GOE, while the second term in \eqref{eq:MatRiscritte} is the additive one. 

\noindent We introduce the notation $\mathbf{u}_1^{a},\ldots,\mathbf{u}_N^{a}$ for the eigenvectors of the matrix $\mathbf{M}^{(a)}$, and $\lambda_{1}^{a}, \ldots,\lambda_N^{a}$ for the associated real eigenvalues.
In the rest of the paper, we give for granted that the index $a$ can be either 0 or 1, and every time it appears it is understood that that property holds for both $a=0$ and $a=1$.\\

\noindent Let us comment on the connection between the matrices we consider and those discussed in Refs.~\cite{bun2018overlaps, potters_bouchaud_2020}. There are different sources of correlations of ${\bf M}^{(0)}$ and ${\bf M}^{(1)}$: first, the noisy part is correlated, since the two matrices share the components $h_{ij}$ which are in common to both elements of the pair. Moreover, when $\Delta_0=\Delta_1$ or $\mu_0=\mu_1$, additional correlations are due to the fact that the matrices are subject to the same multiplicative or additive perturbation. In the latter case, the matrices can be cast in the form ${\bf M}^{(a)}=\sqrt{{\bf C}} {\bf X}^{(a)} \sqrt{{\bf C}}$ when $\mu_a=0$, and ${\bf M}^{(a)}={\bf X}^{(a)}+ {\bf C}$ when $\Delta_a=v_a=\sigma$. These matrices are of the same form as those considered in Refs.~\cite{bun2018overlaps, potters_bouchaud_2020}, with the caveat that the population matrix ${\bf C}$ is of rank-1 and that the noise ${\bf X}^{(a)}$ has always GOE statistics (correlated whenever $\sigma_H \neq 0$). \footnote{In \cite{bun2018overlaps}, the authors consider matrices of the form $\sqrt{{\bf C}} {\bf X}^{(a)} \sqrt{{\bf C}}$ with  ${\bf X}^{(a)}$ uncorrelated Wishart matrices.}\\
\noindent We also remark that each of the two matrices ${\bf M}^{(a)}$  has a statistics that \emph{is not} rotational invariant, since there is a basis vector ${\bf e}_N$ that identifies a special direction along which the statistics of the entries is special. Nonetheless, rotational invariance is preserved in the subspace orthogonal to ${\bf e}_N$, given that the corresponding blocks ${\bf B}^{(a)}={\bf H}+{\bf W}^{(a)}$ have a statistics which is invariant with respect to changes of basis. \\\\
\noindent As mentioned in the introduction, our motivation for looking at matrices with this structure is due to the fact that they describe the local curvature of random Gaussian functions defined on high-dimensional manifolds (for instance, on high-dimensional spheres).
These random fields are studied extensively as toy models of energy landscapes in the theory of glassy and complex systems, of fitness landscapes in evolutionary biology, of loss landscapes in problems of learning (see~\cite{ros2022high} for a recent review). The landscape at two different configurations is correlated, and so is its curvature, described by the Hessians matrices of the random function. It can be shown that such Hessian matrices at two different configurations have correlations described by the formulas above. In particular, due to the isotropy of the random field, the statistics of the Hessians is \emph{almost} rotational invariant (the matrices are of the GOE type), except for one single direction which can be identified with ${\bf e}_N$ in the formulas above, and which corresponds to the direction connecting the two configurations in the manifold. We discuss this mapping in more detail in Sec.~\ref{sec:applications}.

\begin{figure}[h!]
\centering
\includegraphics[width=0.48\textwidth]{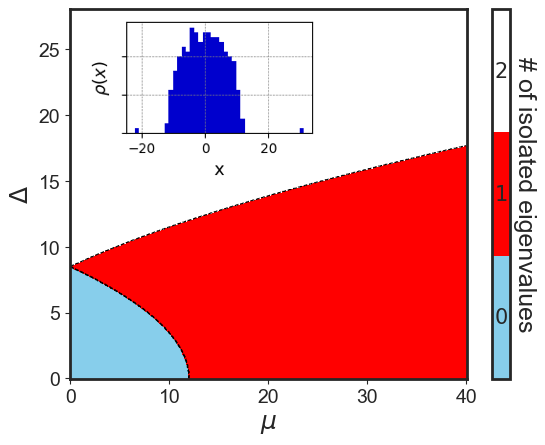}
\caption{The figure represents the regions in the plane $(\Delta,\mu)$ where either 0, 1 or 2 isolated eigenvalues emerge out of the bulk of the spectral density of $\mathbf{M}$. The plot is given for $\sigma=6$. As discussed in the main text, the existence conditions and the typical value of the isolated eigenvalue(s) are independent on $v$. \emph{Inset. } A particular realization of the spectral density for a random matrix $\mathbf{M}$ of size $N=300$, with $\sigma=6, \Delta=25$, $\mu=10$ belonging to the "white" zone, thus presenting two outliers. For simplicity, we took $v=0$  (this parameter does not affect the eigenvalue density in the large-$N$ limit).}
\label{fig:double_evals}
\end{figure}


\subsection{Spectral properties and outliers}
\label{main:isoalted_eigenvalue}
\noindent 
We summarize here the main features of the eigenvalue distribution of matrices of the form \eqref{eq:MatrixForm} (equivalently \eqref{eq:MatRiscritte}), and refer to  Appendix~\ref{app:Density} for more details. Notice that the matrices $\mathbf{M}^{(0)}$ and $\mathbf{M}^{(1)}$ have the same structure; each one has a statistics fully described by the parameters $\sigma, \Delta_a, v_a$ and $\mu_a$ for $a=0,1$. Since the spectral properties discussed in this section involve only eigenvalues and eigenvectors of one single element of the pair of matrices, they are independent of the parameters $\Delta_h, \sigma_H$ and $v_h$ describing the correlations between the entries of the two matrices in the pair. We therefore drop the superscript $a$ and denote the single-matrix parameters simply with $\sigma, \Delta, v$ and $\mu$ in this section. 

\noindent In Refs.~\cite{ros2019complexity, ros2020distribution} it is shown that the perturbation given by the special row and column of ${\bf M}$ can generate a transition in the eigenvalue density in the large-$N$ limit, occurring at a critical value of the parameters $\Delta, \mu, \sigma$; this transition separates a regime in which the eigenvalue density is independent of $\Delta, \mu$ and simply coincides with the eigenvalue density of the GOE matrix $ {\bf X}$ in \eqref{eq:MatRiscritte}, from a regime in which one or two isolated eigenvalues are present, see Fig. \ref{fig:double_evals}. These isolated eigenvalues $\lambda_{\rm iso}$ are detached from the bulk of eigenvalues forming a continuum density in the limit $N \to \infty$. 
These types of spectral transitions belong to the BBP-like transition family \cite{baik2005phase}. For GOE matrices, the BBP transition has been widely investigated in the case of an additive finite-rank perturbations  \cite{edwards1976eigenvalue, benaych2011eigenvalues, potters_bouchaud_2020}, corresponding in our setting to $\Delta=\sigma$. We now discuss the results holding in the general case $\Delta \neq \sigma$, and refer to Appendix \ref{app:Density} for the technical details.\\\\
\noindent In an expansion in $N^{-1}$ the average spectral measure of the matrices ${\bf M}$ reads:
\begin{equation}\label{eq:Dos}
d\nu_N(\lambda)= \rho_N(\lambda) d\lambda + \frac{1}{N}\sum_{*=\pm} \alpha_*\delta(\lambda-\lambda_{\rm iso,*}) + \mathcal{O}\tonde{\frac{1}{N^2}},
\end{equation}
where $\rho_N(\lambda)$ is defined for $|\lambda| \leq 2 \sigma$ and it 
admits the expansion 
\begin{equation}\label{eq:ED}
\rho_N(\lambda)= \rho_\sigma(\lambda)+ \frac{1}{N} \rho^{(1)}_\sigma(\lambda) + \mathcal{O}\tonde{\frac{1}{N^2}}.
\end{equation}
The leading order term in this expansion reduces to the eigenvalue density of the unperturbed GOE matrix ${\bf X}$,
\begin{equation}\label{eq:DoSgoe}
\rho_\sigma(\lambda)=\frac{1}{2 \pi \sigma^2} \sqrt{4 \sigma^2 - \lambda^2},
\end{equation}
while the subleading correction equals to:
\begin{equation}\label{eq:Rho1}
\rho^{(1)}_{\sigma}(\lambda)= \frac{\sqrt{4 \sigma^2-\lambda^2}}{2 \pi (\lambda^2-4 \sigma^2)}- \frac{\text{sign}(\lambda)}{\sqrt{4 \sigma^2-\lambda^2}}+ \sum_{x=\pm 1}\frac{1}{4} \, \delta(\lambda +2 x \sigma),
\end{equation}
see Appendix \ref{app:Density} for a derivation.
 The delta peaks in the measure \eqref{eq:Dos} correspond to the isolated eigenvalues $\lambda_{\rm iso, \pm}$. It can be shown that the eigenvalues are real solutions of the equation
\begin{equation}
\label{egval_eqn}
    \lambda - \mu - \Delta^2 \mathfrak{g}_\sigma(\lambda)=0,
\end{equation}
where for $\lambda$ real such that $ |\lambda|> 2 \sigma$ one has
\begin{equation}\label{eq.Stjlt}\mathfrak{g}_\sigma(\lambda)=\frac{1}{2 \sigma^2} \tonde{\lambda-  \text{sign}(\lambda)\sqrt{\lambda^2- 4 \sigma^2}}, 
\end{equation}
and $\mathfrak{g}_\sigma$
is the Stieltjes transform of the unperturbed GOE matrix ${\bf X}$, obtained from:
\begin{equation}
 \mathfrak{g}_\sigma(z)= \lim_{N \to \infty} \frac{1}{N} \mathbb{E}\quadre{\Tr \tonde{\frac{1}{z- {\bf X}}}}.
\end{equation}
Depending on the strength of the perturbations $\mu, \Delta$, the matrices can exhibit either none, one or two isolated eigenvalues, as we report in Appendix \ref{app:phenom_isolated_eigenvals} and summarize in Fig. \ref{fig:double_evals}.  In particular, for any choice of $\Delta \geq 0$ one isolated eigenvalue exists whenever
\begin{equation}\label{eq:ConditionEx}
|\mu|>  \sigma \tonde{1+ \frac{\sigma^2-\Delta^2}{\sigma^2}},
\end{equation}
and reads 
 \begin{equation}\label{eq:IsoExplicit}
\lambda_{\rm iso, -}=\frac{{2 \mu \sigma^2- \Delta^2 \mu- \text{sign}(\mu) \Delta^2 \sqrt{\mu^2-4 (\sigma^2- \Delta^2)}}}{2 (\sigma^2-\Delta^2)}.
\end{equation}
This expression was first obtained in Refs.~\cite{ros2019complexity, ros2020distribution}, and we re-derive it in Appendix \ref{app:phenom_isolated_eigenvals}.
It is simple to check that for $\mu<0$ it holds $\lambda_{\rm iso,-}< -2 \sigma$, meaning that the isolated eigenvalue is the smallest eigenvalue of the random matrix; similarly, for $\mu>0$ the isolated eigenvalue is the largest one. To connect with known results, it is convenient to express these quantities in terms of the inverse of \eqref{eq.Stjlt}, which for real $y$ is defined in $|y| \leq \sigma^{-1}$ and reads: 
\begin{equation}
\mathfrak{g}^{-1}_\sigma(y )=\frac{1}{y}+ \sigma^2 y.
\end{equation}
It has been shown in \cite{ros2020distribution} that the isolated eigenvalue \eqref{eq:IsoExplicit} can be equivalently written as 
\begin{equation}
 \lambda_{\rm iso, -}=    \mathfrak{g}^{-1}_\sigma \tonde{ \mathfrak{g}_{\overline \sigma}(\mu)}= \frac{1}{ \mathfrak{g}_{\overline \sigma}(\mu)} + \sigma^2  \mathfrak{g}_{\overline \sigma}(\mu)
\end{equation}
with ${\overline \sigma}=(\sigma^2- \Delta^2)^{\frac{1}{2}}$; this expression is well-defined for $|\mathfrak{g}_{\overline \sigma}(\mu)|< \sigma^{-1} $, which corresponds to the existing condition~\eqref{eq:ConditionEx}. From these equations one easily obtains the well-know expression of the isolated eigenvalue in presence of a purely additive rank-one perturbation \cite{edwards1976eigenvalue,benaych2011eigenvalues, potters_bouchaud_2020}: it suffices to set $\Delta \to \sigma, {\overline \sigma} \to 0$ and use the fact that $\lim_{{\overline \sigma} \to 0} \mathfrak{g}_{\overline \sigma}(\mu)= \mu^{-1}$ to get:
\begin{equation}\label{eq:IsoPCA}
 \lambda_{\rm iso, -} \stackrel{\Delta \to \sigma}{\longrightarrow}  \mu + \frac{\sigma^2}{\mu}= \mathfrak{g}^{-1}_\sigma \tonde{ \frac{1}{\mu}}.
\end{equation}

\noindent As derived in Appendix \ref{app:phenom_isolated_eigenvals}, when $\Delta > \sqrt{2} \sigma$ a second isolated eigenvalue exists in the regime 
$|\mu| < \Delta^2 \sigma^{-1}- 2 \sigma$, and equals to:
 \begin{equation}\label{eq:IsoExplicit2}
\lambda_{\rm iso, +}=\frac{{-2 \mu \sigma^2+ \Delta^2 \mu- \text{sign}(\mu) \Delta^2 \sqrt{\mu^2+4 (\Delta^2-\sigma^2)}}}{2 (\Delta^2-\sigma^2)}.
\end{equation}
This case was not discussed in Refs.~\cite{ros2019complexity, ros2020distribution} and, to the best of our knowledge, has not be considered in previous literature.
The eigenvalue $\lambda_{\rm iso, +}$ has a sign opposite to that of $\mu$; thus, for $\mu>0$ one has $\lambda_{\rm iso, +}\leq 0 \leq  \lambda_{\rm iso, -}$, while for $\mu<0$ it holds  $\lambda_{\rm iso, -}\leq 0 \leq \lambda_{\rm iso, +}$. These existence conditions are encoded in \eqref{eq:Dos} by choosing:
\begin{equation}
\begin{split}
&\alpha_-=\Theta\tonde{|\mu|-2\sigma+\frac{\Delta^2}{\sigma}}\\
&\alpha_+=\Theta(\Delta-\sqrt{2}\sigma)\Theta\tonde{-|\mu|-2\sigma+\frac{\Delta^2}{\sigma}}.
\end{split}
\end{equation}

\subsection{The outlier eigenvectors}\label{sec:outlier}
\noindent
The eigenvector ${\bf u}_{\rm iso,-}$ associated to the isolated eigenvalue~\eqref{eq:IsoExplicit} has a 
 projection on the basis vector ${\bf e}_N$ corresponding to the special line and column of the matrix, which remains of $O(1)$ when $N$ is large; the typical value of this projection has been computed in~\cite{ros2020distribution} and reads:
\begin{equation}\label{eq:ProjVectors}
\begin{split}
  &({\bf u}_{\rm iso,-} \cdot  {\bf e}_N)^2= \mathfrak{q}_{\sigma,\Delta}(\lambda_{\rm iso,-},\mu)
   \end{split}
\end{equation}
where we introduced the function:
\begin{equation}\label{eq:Defq}
\begin{split}
   & \mathfrak{q}_{\sigma,\Delta}(\lambda,\mu):=\text{sign}(\mu)\times \\
   &\frac{\text{sign}(\lambda)\Delta^2\sqrt{\lambda^2-4\sigma^2}-\lambda(2\sigma^2-\Delta^2)+2\mu\sigma^2}{2\Delta^2\sqrt{\mu^2-4(\sigma^2-\Delta^2)}}.
    \end{split}
\end{equation}
It can be shown rather easily that whenever Eq.\eqref{eq:ConditionEx} is satisfied, then \eqref{eq:ProjVectors} is positive, as it should be. This can be done by considering separately the cases $\mu>0$ and $\mu<0$ and by using the expression of $\lambda_{\rm iso,-}$. In particular, since $\lambda_{\text{iso},-}$ and $\mu$ have the same sign, one can show that the condition \eqref{eq:ConditionEx} is equivalent to 
\begin{align*}
    -|\lambda_{\text{iso},-}|(2\sigma^2-\Delta^2)+2|\mu|\sigma^2 > 0
\end{align*}
which immediately implies the positivity of \eqref{eq:ProjVectors}. In particular, Eq.\eqref{eq:ProjVectors} is zero if and only if $|\lambda_{\text{iso}, -}|=2\sigma$, which is equivalent to $|\mu|=\sigma^{-1}(2\sigma^2-\Delta^2)$, i.e. the isolated eigenvalue is at the edge of the bulk. \\
\noindent The explicit dependence of \eqref{eq:ProjVectors} on the parameters $\sigma, \Delta,\mu$ reads:
 \begin{equation}
     ({\bf u}_{\rm iso, -} \cdot  {\bf e}_N)^2= \frac{\quadre{\Delta^2 (|\mu|+ \Xi)-2 \sigma^2 \Xi+ \text{sign}(\Delta^2-\sigma^2)\sqrt{\kappa}}}{4 \Xi (\Delta^2-\sigma^2)}
 \end{equation}
 where $\Xi=(\mu^2- 4 \sigma^2+ 4 \Delta^2)^{\frac{1}{2}}$
and $\kappa=(\Delta^2 |\mu|-2 \sigma^2 |\mu|+ \Delta^2 \Xi)^2-16 \sigma^2(\sigma^2-\Delta^2)^2$.  In the case of a purely additive perturbation ($\Delta=\sigma$), using \eqref{eq:IsoPCA} one sees that this expression reduces to:
\begin{equation}\label{eq:OvPCA}
({\bf u}_{\rm iso,-} \cdot  {\bf e}_N)^2 \stackrel{\Delta \to \sigma}{\longrightarrow}  1-\frac{\sigma^2}{\mu^2}
\end{equation}
consistently with previous results~\cite{benaych2011eigenvalues, potters_bouchaud_2020}.  
We remark that for the matrices \eqref{eq:MatrixForm} the joint isolated eigenvalue-eigenvector large deviation function has been determined as well~\cite{ros2020distribution}, generalizing the case of a purely additive perturbation \cite{BiroliGuionnet}.

\section{Eigenvectors overlaps}\label{sec:theoretical_res2}
\noindent In this work we aim at characterizing the correlations between eigenvectors of pairs of correlated matrices with the distribution  \eqref{eq:MatrixForm}, similarly to what is discussed in \cite{bun2018overlaps} for unperturbed GOE matrices. In particular, our objects of interest are the {averaged squared overlaps} between eigenvectors associated to different eigenvalues of the two matrices:
\begin{equation}\label{phi}
\Phi(\lambda^{0},\lambda^{1}):=N\mathbb{E}[\langle \mathbf{u}_{\lambda^{0}},\mathbf{u}_{\lambda^{1}}\rangle^2],
\end{equation}
where $\lambda^a$ are eigenvalues of ${\bf M}^{(a)}$, ${\bf u}_{\lambda^a}$ the corresponding eigenvectors, and the expectation value $\mathbb{E}$ represents the average over the distribution of all the entries of the two matrices. In the limit of large $N$  this quantity remains of $\mathcal{O}(1)$ for values of $\lambda^a$ belonging to the continuous part (henceforth, the \emph{bulk}) of the eigenvalue density of the two matrices. 
We are interested in computing both the overlap between eigenstates associated to eigenvalues in the bulk, as well as the overlaps involving the eigenvectors associated to the isolated eigenvalues of the matrices, whenever they exist. In the first case, the average $\mathbb{E}$ over different realizations of the random matrices can be replaced by an average, for fixed
randomness, over eigenvectors associated to eigenvalues within windows of width $d\lambda \gg N^{-1}$ centred around $\lambda^{0}, \lambda^{1}$: as a matter of fact, the quantity \eqref{phi} is self-averaging in the large $N$ limit \cite{bun2018overlaps}. 

\noindent Consider now the overlaps involving the eigenvectors associated to the isolated eigenvalues. As we have discussed in the previous section, any element of the pair of matrices in \eqref{eq:MatrixForm} can present zero, one or two isolated eigenvalues. Such eigenvalues pop out of the bulk of the spectral density, which at leading order in $N$ is given by the Wigner's semicircle law. Two isolated eigenvalues, denoted by $\lambda^a_{\text{iso},\pm}$, exist for each matrix $a\in\{0,1\}$ only when the noise from the special row and column is considerably bigger than the variance of the main GOE blocks, i.e. $\Delta_a >\sqrt{2} \sigma$. 
In the following, we restrict to the case in which only one isolated eigenvalue exists, equal to $\lambda^a_{\text{iso},-}$. To simplify the notation, henceforth we set
\begin{equation}
\label{eval_iso}
\lambda_{\rm iso}^a \equiv \lambda_{\rm iso, -}^a,
\end{equation}
meaning that 
 $\lambda^0_{\text{iso}}$ is the isolated eigenvalue \eqref{eq:IsoExplicit}  of $\mathbf{M}^{(0)}$, and analogously for $\lambda^1_{\text{iso}}$. All of the results presented in the following can be easily generalized to the other isolated eigenvalues of the random matrices, whenever they exist. 
 
 We also remark that in the case in which both the eigenvalues in \eqref{phi} are isolated, the relevant quantity to determine is the rescaled function:
\begin{align}
\label{phi_rescaled}
\tilde \Phi(\lambda_{\rm iso}^0,\lambda_{\rm iso}^1)   :=\frac{\Phi(\lambda_{\rm iso}^0,\lambda_{\rm iso}^1)}{N}=\mathbb{E}\left[\langle \mathbf{u}_{\lambda_{\text{iso}}^0},\mathbf{u}_{\lambda_{\text{iso}}^1} \rangle^2\right].
\end{align}
This is because both eigenvectors have an $\mathcal{O}(1)$ projection on the special direction ${\bf e}_N$ given by \eqref{eq:ProjVectors}, so that their overlap is at least of the order of \eqref{eq:ProjVectors}. This clearly indicates that the quantity that remains of $\mathcal{O}(1)$ in the limit of large $N$ is the 
rescaled quantity \eqref{phi_rescaled}.


\noindent The overlap \eqref{phi} takes a different form depending on whether the considered eigenvalues (either both, or one of them) belong to the bulk of the eigenvalues density of their respective matrices, or are isolated. Our main results are the explicit formulas for the overlaps in all the different cases, as a function of the parameters defining the statistics of the matrices. An overview of the calculations leading to these results is given in Sec.~\ref{sec:overview_comp}, and details are reported in the Appendices. In the following subsections we report the final explicit expressions.

\subsection{Eigenvector overlaps  of bulk-bulk eigenvalues}

\noindent Each element of the pair of random matrices defined in Eq. \eqref{eq:MatrixForm} has a GOE block ${\bf B}^{(a)}$ having the same statistics (only the matrix elements in the special row and column have a statistics that  depends on $a$). The bulk spectral densities  $\rho_\sigma(\lambda)$ of both matrices in the large $N$ limit are determined by these blocks, and thus are exactly the same for both matrices, given by \eqref{eq:DoSgoe}. The spectral densities are supported on the interval $[-2\sigma,2\sigma]$; when the respective eigenvalues $\lambda^0,\lambda^1\in [-2\sigma,2\sigma]$, the overlap between the two correspondent eigenvectors reads:
\begin{align}
\label{eq:phi_bulk_bulk}
\Phi(\lambda^0, \lambda^1)=  \frac{2 \sigma_W^2 \, \tonde{1-\frac{\sigma_W^2}{2 \sigma^2}} (\lambda^0-\lambda^1)^2}{\displaystyle \prod_{k=\pm }A_k}  + \mathcal{O}\tonde{\frac{1}{N}}
\end{align}
with 
\begin{align}
\begin{split}
    A_k=&\frac{\sigma_W^4}{4 \sigma^4}\tonde{\sqrt{4 \sigma^2-(\lambda^0)^2}+k\sqrt{4 \sigma^2-(\lambda^1)^2}}^2\\
    &+\tonde{1-\frac{\sigma_W^2}{2 \sigma^2}}^2 (\lambda^0-\lambda^1)^2.
\end{split}
\end{align}
This expression depends only on the parameters $\sigma, \sigma_W$ defining the statistics of the GOE blocks ${\bf D}^{(a)}$, and it is consistent with the results of Ref.~\cite{bun2018overlaps}. Indeed, Ref.~\cite{bun2018overlaps} presents the calculation of the overlap between bulk eigenvectors of matrices of the form ${\bf C}+ {\bf A}+ {\bf D}^{(a)}$, where ${\bf C}$ is a deterministic (population) matrix, while ${\bf A}$ and ${\bf D}^{(a)}$ are $N\times N$ GOE matrices with variances $\rho_{12}$ and $\sigma^2_a-\rho_{12}$, respectively. The overlap is shown to be independent of the matrix ${\bf C}$, and to coincide with \eqref{eq:phi_bulk_bulk} with $ \sigma_H^2 \to {\rho_{12}}$ and $ \sigma_W^2 \to \sigma_a^2-\rho_{12}$, as expected.  
Notice that the case considered in Ref.~\cite{bun2018overlaps} corresponds to vanishing finite-rank perturbations ($\Delta_a=v_a=\sigma$, $\mu_a=0$); therefore, no isolated eigenvalue(s) are present in that case. Eq.~\eqref{eq:phi_bulk_bulk}  shows that the finite rank perturbations do not affect the overlap between bulk eigenvectors, to leading order in $N$.  
We remark that the $1/N$ contribution to \eqref{eq:phi_bulk_bulk} can also be determined explicitly: we discuss this in Sec.~\ref{sec:FFiniteSize}. \\

\begin{figure}[h!]
\centering
\includegraphics[width=0.48\textwidth]{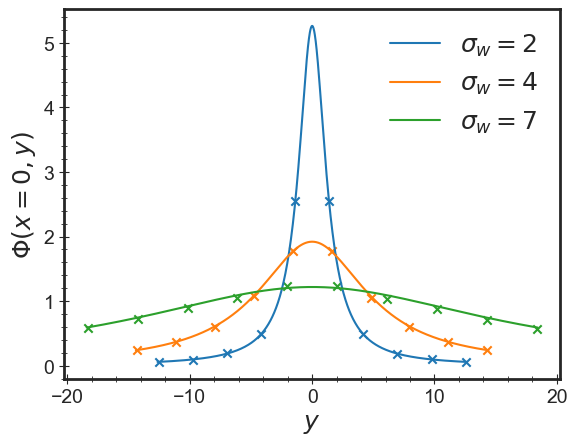}
\caption{Plot representing the theoretical curves of the bulk-bulk overlap \eqref{eq:phi_bulk_bulk} with their respective numerical simulations (colored points). We used $\sigma_H=6, \Delta_h=2.5, \Delta_{w,0}=2, \Delta_{w,1}=1.5, \mu_0=\mu_1=0$, and we plot the overlap for $x=\lambda^0=0$ and $y=\lambda^1\in[-2\sigma,2\sigma]$, for several choices of $\sigma_W$. The numerical simulations were carried out by generating 500 times pairs of random matrices of size $N=200$. As for Fig.\ref{fig:double_evals} we set $v_0=v_1=0$ given that the final results do not depend on them, to leading orders.}
\label{fig:phi_bulk_bulk}
\end{figure}

\noindent A numerical check of \eqref{eq:phi_bulk_bulk} is given in Fig.~\ref{fig:phi_bulk_bulk}. We briefly comment on how the numerical simulations are performed. In order to obtain the eigenvectors overlaps numerically, we generate the three GOE random matrices $\mathbf{H}$, $\mathbf{W}^{(0)}$ and $\mathbf{W}^{(1)}$; similarly, we generate the Gaussian variables $h_{iN}$, $w_{iN}^0, w_{iN}^1$. The elements $m_{NN}^0$ and $m_{NN}^1$ are simply set equal to $\mu_0$ and $\mu_1$ respectively, i.e. we set $v_0=v_1=0$; this is motivated by the fact that, as we show below, to the $1/N$ order we are interested in, all of our analytical results are independent on the variances $v_0,v_1$. After having generated such entries, we sum them up to get the two matrices $\mathbf{M}^{(0)}$ and $\mathbf{M}^{(1)}$, according to Sec.\ref{subsec:matrix}. We then diagonalize them and consider eigenvectors associated to eigenvalues in the intervals $[x-d\lambda,x+d\lambda]$ and $[y-d\lambda,y+d\lambda]$ respectively, with $d\lambda\gg N^{-1}$. Then for each pair of such eigenvectors of the two matrices, we compute their squared dot product, and average them together. We repeat this procedure over many realizations: the numerical points in the Figures correspond to averages over the realizations. All the Figures reported in the following are generated in this way, with the slight difference that when isolated eigenvalues are considered, there is no window $d\lambda$ on which to perform the first average, and thus the number of realizations has to be increased significantly. 

\subsection{Eigenvector overlaps  of isolated-isolated eigenvalues}
\noindent We now consider the case in which both $\lambda_{\text{iso}}^0$ and $\lambda_{\text{iso}}^1$ exist, and we give the expression for the rescaled overlap \eqref{phi_rescaled} of the corresponding eigenvectors. 
Given the function:
\begin{equation}
\label{BigPsi}
\Psi(z,\xi):=\frac{\mathfrak{g}_\sigma(z)-\mathfrak{g}_\sigma(\xi)}{\xi -z-\sigma_W^2(\mathfrak{g}_\sigma(\xi)-\mathfrak{g}_\sigma(z))},
\end{equation} 
we find that the overlap can be compactly written as:
\begin{equation}
\label{eq:phi_iso_iso}
\begin{split}
  & \tilde \Phi(\lambda_{\rm iso}^0,\lambda_{\rm iso}^1)=\mathfrak{q}_{\sigma,\Delta_0}(\lambda_{\rm iso}^0,\mu_0)\mathfrak{q}_{\sigma,\Delta_1}(\lambda_{\rm iso}^1,\mu_1)\; \\&\times\quadre{\Delta_h^2\Psi(\lambda_{\rm iso}^0,\lambda_{\rm iso}^1)
+1}^2 +\mathcal{O}\tonde{\frac{1}{N}},
    \end{split}
\end{equation}
where  $\mathfrak{q}_{\sigma,\Delta}$ is defined in Eq. \eqref{eq:Defq}. 

Let us comment on some limiting values of this expression. In the case in which the two matrices ${\bf M}^{(a)}$ have uncorrelated entries in the special line and column (meaning that $\Delta_h=0$) then the overlap reduces to $\mathfrak{q}_{\sigma,\Delta_0}(\lambda_{\rm iso}^0,\mu_0)\mathfrak{q}_{\sigma,\Delta_1}(\lambda_{\rm iso}^1,\mu_1)$ and thus it coincides with the product of two terms like \eqref{eq:ProjVectors}, one for each matrix. In fact, it is natural to expect that when the entries are uncorrelated, the eigenvectors corresponding to the isolated eigenvalues are orthogonal in the subspace that is complementary to the special direction ${\bf e}_N$, implying that their overlap is fully determined by their projection on the special direction ${\bf e}_N$. More precisely, given the decomposition $\mathbf{u}_{\lambda_{\text{iso}}^a}= (\mathbf{u}_{\lambda_{\text{iso}}^a} \cdot {\bf e}_N) {\bf e}_N + {\bf v}^a$ with ${\bf v}^a$ being the projection of $\mathbf{u}_{\lambda_{\text{iso}}^a}$ on the space orthogonal to ${\bf e}_N$, the above assumption corresponds to ${\bf v}^0 \cdot {\bf v}^1=0$, which implies $(\mathbf{u}_{\lambda_{\text{iso}}^0} \cdot \mathbf{u}_{\lambda_{\text{iso}}^1})^2= (\mathbf{u}_{\lambda_{\text{iso}}^0} \cdot {\bf e}_N)^2 (\mathbf{u}_{\lambda_{\text{iso}}^1} \cdot {\bf e}_N)^2$, which using \eqref{eq:ProjVectors} is precisely \eqref{eq:phi_iso_iso} for $\Delta_h=0$. This is the minimal value one expects for the overlap. 
On the other hand, when the two matrices are fully correlated ($\sigma_W=0=\Delta_{w,0}=\Delta_{w,1}$) the overlap is maximal and equal to one, as it can be checked from the above formulas. The dependence of $\tilde{\Phi}(\lambda_{\rm iso}^0,\lambda_{\rm iso}^1)$ on the variances $\sigma_W, \sigma_H$ is shown in Fig. \ref{fig:phi_iso_iso_pca_color_plot} for the particular case in which the perturbation is fully additive, and identical in strength for both matrices. 
\noindent Further comparisons of the formula \eqref{eq:phi_iso_iso} with numerical simulations are given in Fig.~\ref{fig:phi_iso_iso_pca_vs_sigmah_sigmaw} in Sec.~\ref{sec:applications}, where we discuss the special case of matrix PCA. 

\begin{figure}[h!]
\centering
\includegraphics[width=0.48\textwidth]{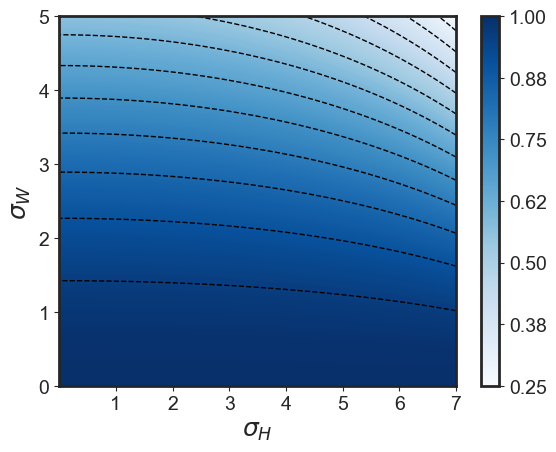}
\caption{Density plot for $ \tilde \Phi(\lambda_{\rm iso}^0,\lambda_{\rm iso}^1)$ for matrices subject to purely additive perturbations ($\Delta_h=\sigma_H$, $\Delta_{w,0}=\Delta_{w,1}=\sigma_W$) with $\mu=\mu_0=\mu_1=10$. The dashed lines are level curves.}
\label{fig:phi_iso_iso_pca_color_plot}
\end{figure}

\subsection{Eigenvector overlaps  of bulk-isolated eigenvalues}
\noindent Consider the case in which at least one of the two matrices $\mathbf{M}^{(0)}, \mathbf{M}^{(1)}$ has the isolated eigenvalue. Without loss of generality, we take such matrix to be $\mathbf{M}^{(0)}$, meaning that \eqref{eq:ConditionEx} is satisfied and $\lambda^{0}_{\text{iso}}$ exists (it is clear that all results will hold if we exchange the two matrices).  We impose no condition on $\mathbf{M}^{(1)}$, and we pick a bulk eigenvalue $y:=\lambda^1\in[-2\sigma,2\sigma]$. 
In this case, the formula for the overlap is rather cumbersome, as it is given by the following expression:
\begin{align}
\begin{split}
\label{eq:phi_iso_bulk}
&\Phi(\lambda^0_{\rm iso},y)=\frac{\mathfrak{q}_{\sigma, \Delta_0}(\lambda_{\rm iso}^0,\mu_0)}{2\pi\rho_\sigma(y)}\Bigg[\frac{4\Delta_0^2\sigma^2}{\sqrt{[\lambda_{\rm iso}^0]^2-4\sigma^2}}\frac{bc-ad}{c^2+d^2}\\
&-4\sigma^2\Delta_h^4\frac{b_1c_1e_1-a_1d_1e_1-a_1c_1f_1-b_1d_1f_1}{(c_1^2+d_1^2)(e_1^2+f_1^2)}\\
&-8\sigma^2\Delta_h^2\frac{b_2c_2e_2-a_2d_2e_2-a_2c_2f_2-b_2d_2f_2}{(c_2^2+d_2^2)(e_2^2+f_2^2)}\\
&+\frac{\Delta_0^2\Delta_1^2\mathfrak{g}_\sigma(\lambda_{\rm iso}^0)}{\sigma^2(\lambda_{\rm iso}^0-\mu_0)(y-\mu_1)}\frac{b_3c_3-a_3d_3}{c_3^2+d_3^2}
\Bigg] + \mathcal{O}\tonde{\frac{1}{N}}.
\end{split}
\end{align}
The quantities $a,b,c,a_1,b_1,c_1,a_2,b_2,c_2,a_3,b_3,c_3$ are functions of $\lambda^0_{\rm iso}$ and $y$, and depend explicitly on the parameters $\sigma, \sigma_W$ and $\Delta_1$.
For compactness, we list their explicit expressions in Appendix \ref{app:computation_of_phi_y_iso}. We also recall that the expressions for $\mathfrak{g}_\sigma$, $\rho_\sigma$ and $\mathfrak{q}_{\sigma,\Delta}$ are given in Eq.\eqref{eq:DoSgoe}, Eq.\eqref{eq.Stjlt} and  Eq. \eqref{eq:Defq}. 

\noindent In Fig.\ref{fig:phi_iso_y} we show that the complicated parameter dependencies of \eqref{eq:phi_iso_bulk} are exact, and numerical simulations perfectly agree with our theoretical results. \\\\

\begin{figure}[h!]
\centering
\includegraphics[width=0.48\textwidth]{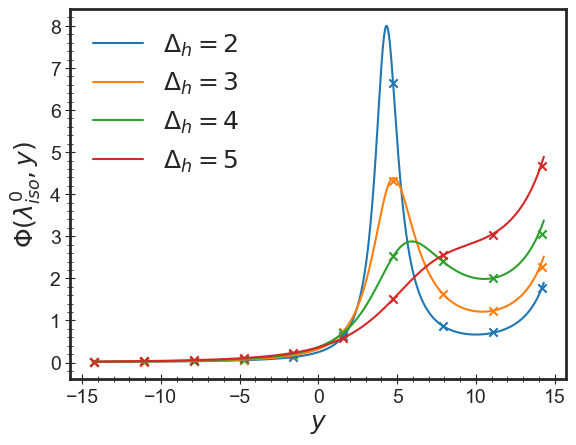}
\caption{Plot representing the theoretical curves of the bulk-isolated overlap \eqref{eq:phi_iso_bulk} with their respective numerical simulations (colored points). We used $\sigma_H=6.5, \sigma_W=3, \Delta_{w,0}=2, \Delta_{w,1}=1.5, \mu_0=15, \mu_1=4$, and we plot the overlap for $x=\lambda^0_{\text{iso}}$ and $y=\lambda^1\in[-2\sigma,2\sigma]$ (were clearly $\sigma=\sqrt{\sigma_H^2+\sigma_W^2}$) for several choices of $\Delta_h$. The numerical simulations were carried out by generating 1000 times pairs of random matrices of size $N=500$. As for Fig.\ref{fig:double_evals} we set $v_0=v_1=0$ given that the final results do not depend on them, to leading orders.}
\label{fig:phi_iso_y}
\end{figure}

\noindent As it is evident from Eqs. \eqref{eq:phi_bulk_bulk}, \eqref{eq:phi_iso_iso} and \eqref{eq:phi_iso_bulk}, the expressions for the overlaps do not depend on the parameters $v_h, v_{w,a}$ which control the strength of the fluctuations of the matrix elements $m_{NN}^a$ at the scale $N^{-1/2}$; on the other hand, they depend explicitly on the average values $\mu_a$ of those matrix elements, which are of $\mathcal{O}(1)$. More generally, the bulk properties of the pair of matrices ${\bf M}^{(a)}$ depend only on the parameters $\sigma, \sigma_H, \sigma_W$ defining the statistics of the blocks ${\bf B}^{(a)}$ (see for instance Eqs. \eqref{eq:DoSgoe} and \eqref{eq:phi_bulk_bulk}): the $\mathcal{O}(N^{-1/2})$ fluctuations of $\mathcal{O}(N^2)$ matrix elements contribute  to these quantities to leading order, while the fluctuations and correlations of a subleading number of matrix elements (such as those in the special line and column) matter only at the subsequent order in the $1/N$ expansion. On the other hand, the isolated eigenvalues and eigenvectors (which give a subleading contribution to the eigenvalue density with respect to the bulk, see Eq. \eqref{eq:Dos})  are sensitive to the $\mathcal{O}(N^{-1/2})$ fluctuations of the $\mathcal{O}(N)$ entries $m_{iN}^a$ for $i<N$, as well as to changes in the averages $\mu_a$ that are of $\mathcal{O}(1)$; this appears evident from the Eqs. \eqref{eq:phi_iso_iso} and \eqref{eq:phi_iso_bulk}. However, the fluctuations of a single matrix element $m_{NN}^a$ at the scale $N^{-1/2}$ are not strong enough to shift the typical value of these quantities (even though they affect the large deviation functions describing the probability that these quantities take  atypical values, as shown in Ref.~\cite{ros2020distribution}). We naturally expect that the dependence on the parameters $v_h, v_{w,a}$ will appear in typical values only at the next orders in the $1/N$ expansion.

\section{Overview of the computations}
\label{sec:overview_comp}
\subsection{A formula to extract the overlaps}
In this section, we aim at giving an overview on how the computation of the overlaps \eqref{phi} is carried out in the three cases presented in Section \ref{sec:theoretical_res2}. The derivation is similar to that discussed in Ref.~\cite{bun2018overlaps}. We begin by introducing the auxiliary function
\begin{align}\label{eq:ExpProRe}
\psi(z,\xi):&=\frac{1}{N}\mathbb{E}\left[\Tr[(z-\mathbf{M}^{(0)})^{-1}(\xi-\mathbf{M}^{(1)})^{-1}]\right],
\end{align}
which will be useful as a computation tool. For finite $N$, the spectral decomposition of the matrices yields:
\begin{align*}\label{eq:DoubleSum1}
    &\psi(x-i\eta,y\pm i\eta)=\frac{1}{N}\sum_{\alpha,\beta}\mathbb{E}\bigg[\frac{1}{x-i\eta-\lambda^{0}_{\alpha}}\\&\times\frac{1}{y\pm i\eta-\lambda^{1}_{\beta}}|\langle \mathbf{u}^{0}_{\alpha}|\mathbf{u}^{1}_{\beta}\rangle|^2
    \bigg]\\
    &=\frac{1}{N^2}\sum_{\alpha,\beta}\mathbb{E}\left[
  R_{x,y,\eta}^{\pm}(\lambda^{0}_{\alpha},\lambda^{1}_{\beta}) \; 
    N|\langle\mathbf{u}^{(0)}_{\alpha}|\mathbf{u}^{1}_{\beta}\rangle|^2
    \right].
\end{align*}
where we defined:
\begin{equation}
R_{x,y,\eta}^{\pm}(\lambda, \chi)=\frac{1}{x-\lambda-i\eta}\frac{1}{y-\chi\pm i\eta}.
\end{equation}

\noindent In the large $N$ limit, the sums over the eigenvalues can be turned into integrals over the spectral measure of the matrices, taking care of the presence of the subleading terms due to the isolated eigenvalues. 

\noindent The above expression hence becomes equivalent to 
{\medmuskip=0mu
\thinmuskip=0mu
\thickmuskip=0mu
\begin{align*}
\psi(x-i\eta,y\pm i\eta)=&  \int d\lambda\,d\chi\,\rho_\sigma(\lambda)\rho_\sigma(\chi)R^{\pm}_{x,y,\eta}(\lambda, \chi)  \Phi(\lambda,\chi)\\
    &+\frac{1}{N}\int  d\lambda\rho_\sigma(\lambda)R^{\pm}_{x,y,\eta}(\lambda, \lambda^{1}_{\rm iso}) \Phi(\lambda,\lambda^{1}_{\rm iso})\\
    &+\frac{1}{N}\int d\chi\rho_\sigma(\chi)R^{\pm}_{x,y,\eta}(\lambda^{0}_{\rm iso}, \chi) \Phi(\lambda^{0}_{\rm iso},\chi)\\
    &+\frac{1}{N}R^{\pm}_{x,y,\eta}(\lambda^{0}_{\rm iso}, \lambda^{1}_{\rm iso}) \frac{\Phi(\lambda^{0}_{\rm iso}, \lambda^{1}_{\rm iso})}{N}\\
\end{align*}}
where $\rho_\sigma$ denotes the continuous part of the eigenvalue densities of the matrices ${\bf M}^{(a)}$, for $a\in\{0,1\}$, defined in \eqref{eq:DoSgoe}. We set  $\psi_0=\lim_{\eta\to0^+}\psi$. The Sokhotski-Plemelj identity implies 
\begin{align}
\label{re_psi}
    \begin{split}
    &\text{Re}\left[\psi_0(x-i\eta,y+i\eta)-\psi_0(x-i\eta,y-i\eta)\right]\\
&=2\pi^2\Phi(x,y)\rho_\sigma(x)\rho_\sigma(y)\\
    &+\frac{2\pi^2}{N}\Phi(\lambda^{0}_{\rm iso},y)\rho_\sigma(y)\delta(x-\lambda^{0}_{\rm iso})\\
    &+\frac{2\pi^2}{N}\Phi(x,\lambda^{1}_{\rm iso})\rho_\sigma(x)\delta(y-\lambda^{1}_{\rm iso})\\
    &+\frac{2\pi^2}{N} \tilde{\Phi}(\lambda^{0}_{\rm iso},\lambda^{1}_{\rm iso})\delta(x-\lambda^{0}_{\rm iso})\delta(y-\lambda^{1}_{\rm iso}).
    \end{split}
\end{align}
We therefore see that in order to get the expression for $\Phi(\lambda^{0}_{\rm iso},y)$ we have to compute $\psi(z, \xi)$ and isolate the $1/N$ correction proportional to $\delta(x-\lambda^{0}_{\rm iso})$  appearing in the formula above. Instead, the term proportional to two delta peaks will give information on the overlap $\tilde \Phi(\lambda^{0}_{\rm iso},\lambda^{1}_{\rm iso})$. Notice that even though we are focusing on the case in which one single isolated eigenvalue $\lambda_{\rm iso} \equiv \lambda_{\rm iso, -}$ exists, all calculations can be extended straightforwardly to the second isolated eigenvalue, whenever it exists. \\
\noindent The above computations show that the expressions for the various overlaps can be obtained provided that the auxiliary function $\psi$ is computed up to order $1/N$. In the following sections we give an overview of such computation.

\subsection{Accounting for the finite-rank perturbations}\label{sec:BlockExpansion}
\noindent The matrices $\mathbf{M}^{(0)}, \mathbf{M}^{(1)}$ have a block structure, implying that $(z-\mathbf{M}^{(a)})$ can be inverted using Schur matrix inversion lemma, recalled in  Appendix \ref{app:stieltjes}. We set $M=N-1$, and introduce, for $a\in\{0,1\}$, the $M \times M$ matrices
\begin{equation}
\label{eq:DefA}
   {\bf A}^{(a)}(z)=\frac{{\bf m}^a \; [{\bf m}^a]^T }{z- m^a_{NN}}, {\bf m}^a=(m^a_{1N}, m^a_{2N} \cdots , m^a_{M \,  N})^T.
\end{equation}
Let us exploit Schur's matrix inversion lemma. For $i,j\leq M$ one has
\begin{equation}\label{eq:Block1}
    (z  -\mathbf{M}^{(a)})^{-1}_{ij}=
    (z -{\bf H}-{\bf W}^{(a)}-{\bf A}^{(a)}(z))^{-1}_{ij}
\end{equation}
and 
\begin{equation}
    (z -\mathbf{M}^{(a)})^{-1}_{iN}=
       -\sum_{k=1}^{N-1} \frac{m^a_{kN}}{z-m^a_{NN}}   (z  -\mathbf{M}^{(a)})^{-1}_{ik},
\end{equation}
while the remaining component reads
\begin{equation}\label{eq:Block2}
\begin{split}
    &(z-\mathbf{M}^{(a)})^{-1}_{N\, N}=(z-m^a_{NN})^{-1}\times\\&\times\left\{1+\sum_{k,l=1}^{N-1}\frac{m^a_{kN}m^a_{lN}}{z-m^a_{NN}}(z  -\mathbf{M}^{(a)})^{-1}_{kl}\right\}^{-1}.
    \end{split}
\end{equation}
It is thus convenient to decompose $\psi(z, \xi)$ as $\psi=\psi_{00}+\psi_{0N}+\psi_{NN}$ with:
\begin{align}
\label{psi_decomposition}
\begin{split}
    &\psi_{00}(z, \xi)=\mathbb{E}\left[\frac{1}{N}\sum_{i,j=1}^{N-1}(z-\mathbf{M}^{(0)})^{-1}_{ij}(\xi-\mathbf{M}^{(1)})^{-1}_{ij}\right]\\
    &\psi_{0N}(z, \xi)=\mathbb{E}\left[\frac{2}{N}\sum_{i=1}^{N-1}(z-\mathbf{M}^{(0)})^{-1}_{iN}(\xi-\mathbf{M}^{(1)})^{-1}_{iN}\right]\\
    &\psi_{NN}(z, \xi)=\mathbb{E}\left[\frac{1}{N}(z-\mathbf{M}^{(0)})^{-1}_{NN}(\xi-\mathbf{M}^{(1)})^{-1}_{NN}\right].
\end{split}
\end{align}
We make use of the expansion:
\begin{align}
\label{eq:dyson}
  \frac{1}{z-{\bf H}-{\bf W}^{(a)}-{\bf A}^{(a)}(z)}=\mathbf{G}_a(z) \sum_{u=0}^\infty [{\bf A}^{(a)}(z) {\bf G}_a(z)]^u
\end{align}
where the resolvent operator, defined as 
\begin{equation}
\label{eq:resolvent}
    {\bf G}_a(z):=   (z-{\bf H}-{\bf W}^{(a)})^{-1}
    \quad\quad a\in\{0,1\},
\end{equation}
does not depend on the components $m^a_{i N}$. For simplicity, we first perform the average over the entries  $m^a_{i \, N}$ for $i <N$, with $a=0,1$, since they don't appear in the resolvents. 
As shown in Appendix \ref{app:comp_psi_pieces}, the average of this operator expansion can be computed term by term~\footnote{We remark that these expansions can be recovered making use of the multi-resolvent local law proved in Ref. \cite{cipolloni2022thermalisation}.}, and the re-summation of the  contributions up to order $1/N$ can be performed explicitly. In particular, calling:
\begin{equation}
\label{eq:f}
    f(z; \Delta_a, \mu_a)= \frac{\Delta_a^2}{z-\mu_a- \frac{\Delta_a^2}{N} \Tr \mathbb{E}\mathbf{G}_a(z)},
    \quad a\in\{0,1\}
\end{equation}
we find
\begin{equation}
\label{eq:psi00Implicit}
\begin{split}
  &   \psi_{00}(z,\xi)= \frac{1}{N}\Tr\mathbb{E}[\mathbf{G}_0(z)\mathbf{G}_1(\xi)]+\\
  & \frac{1}{N}\left(\frac{1}{N}\Tr\mathbb{E}[\mathbf{G}^2_0(z)\mathbf{G}_1(\xi)]\right)  f(z; \Delta_0, \mu_0)+\\
  &\frac{1}{N}\left(\frac{1}{N}\Tr\mathbb{E}[\mathbf{G}_0(z)\mathbf{G}_1^2(\xi)]\right)    f(\xi; \Delta_1, \mu_1)+\\
 & \frac{1}{N}\tonde{\frac{\Delta_h^4}{\Delta^2_0 \,  \Delta^2_1}} \left(\frac{1}{N}\Tr\mathbb{E}[\mathbf{G}_0(z)\mathbf{G}_1(\xi)]\right)^2   f(\xi; \Delta_1, \mu_1)\times\\
 &  f(z; \Delta_0, \mu_0)+ \mathcal{O}\tonde{\frac{1}{N^2}}.
     \end{split}
\end{equation}
where we recall that $\Delta^2_a= \Delta^2_h + \Delta^2_{w,a}$.
Similarly, 
\begin{equation}
\label{eq:psi0NImplicit}
\begin{split}
  &  \psi_{0N}(z,\xi)=\frac{2}{N} \frac{\Delta^2_h}{\Delta^2_0 \,  \Delta^2_1} \; f(\xi; \Delta_1, \mu_1)  f(z; \Delta_0, \mu_0)\times\\
    & \times \frac{1}{N}\Tr\mathbb{E}[\mathbf{G}_0(z)\mathbf{G}_1(\xi)]+\mathcal{O}\left(\frac{1}{N^2}\right)
    \end{split}
\end{equation}
and 
\begin{align}
\label{eq:psiNNImplicit}
  \psi_{NN}(z,\xi)=\frac{1}{N}\frac{L_{NN}(z, \xi)}{(z-\mu_0)(\xi-\mu_1)} +\mathcal{O}\left(\frac{1}{N^2}\right)
\end{align}
with 
\begin{equation}
\begin{split}
\label{eq:L}
    &L_{NN}(z,\xi)= 1 + 
    f(z; \Delta_0, \mu_0) \frac{1}{N}\Tr\mathbb{E}[\mathbf{G}_0(z)]\\
    &+f(\xi; \Delta_1, \mu_1) \frac{1}{N}\Tr\mathbb{E}[\mathbf{G}_1(\xi)]\\
    &+ f(z; \Delta_0, \mu_0) \, f(\xi; \Delta_1, \mu_1) \frac{1}{N^2} \Tr\mathbb{E}[\mathbf{G}_0(z)] \,\Tr\mathbb{E}[\mathbf{G}_1(\xi)].
    \end{split}
\end{equation}
We introduce the deterministic matrices:
\begin{equation}
\label{main:multi_resolv}
  {\bf \Pi}_{k,m}:= \mathbb{E} \left[  \mathbf{G}_0(z)^{k+1} \, \mathbf{G}_1(\xi)^{m+1} \right]
\end{equation}
for $m,k$ non-negative integers.
It appears from the above expressions that in order to obtain explicit formulas for $\psi(z, \xi)$, one needs to compute the leading order contributions to the quantities ${\bf \Pi}_{1,1}$, ${\bf \Pi}_{1,2}$ and ${\bf \Pi}_{2,1}$. 

In Appendix \ref{app:ExpectationProductResolvents} we compute \eqref{main:multi_resolv} for general values of $k,m$, to order $1/N$. 
In the following subsection we report the resulting expression, which is of its own interest and, to the best of our knowledge, not given in previous literature. Given this general result, we can obtain the explicit form of ${\bf \Pi}_{1,1}$, ${\bf \Pi}_{1,2}$ and ${\bf \Pi}_{2,1}$, and thus get explicit formulas for equations \eqref{eq:psi00Implicit},\eqref{eq:psi0NImplicit}, \eqref{eq:psiNNImplicit} and therefore for $\psi$, see Appendix \ref{app:comp_psi}. The expression for $\psi$ can then be plugged inside \eqref{re_psi}, and from there one could extract the formulas for the various overlaps.  
These final steps are exposed in detail in Appendix \ref{app:computation_overlaps}. We remark that the fact that we can compute  \eqref{main:multi_resolv} to order $1/N$ also allows us to determine the $1/N$ corrections to the overlap  \eqref{eq:phi_bulk_bulk}, that we also present below.


\subsection{Multiresolvents products and finite size corrections to $\Phi\left(\lambda^0,\lambda^1\right)$}\label{sec:FFiniteSize}

\noindent In this subsection we present the formula for the expected matrix \eqref{main:multi_resolv} to order $1/N$, and also summarize our results for the $1/N$ corrections to the bulk-bulk overlap, whose leading order expression is equation \eqref{eq:phi_bulk_bulk}. 
As we prove in Appendix \ref{app:From higher-order products to products of pairs.} we have that
\begin{align}
\label{eq:multi_resolv_formula0}
    {\bf \Pi}_{k,m}=\frac{(-1)^{k+m}}{k!m!}\frac{\partial^k}{\partial z^k} \frac{\partial^m}{\partial \xi^m } 
\mathbb{E} \left[  \mathbf{G}_0(z) \, \mathbf{G}_1(\xi)\right].
\end{align}
To leading order in $N$, the matrix $\mathbb{E} \left[  \mathbf{G}_0(z) \, \mathbf{G}_1(\xi)\right]$ converges to a diagonal one with components given by \eqref{BigPsi} \cite{cipolloni2022thermalisation, potters_bouchaud_2020}. There are two types of $1/N$ corrections that contribute to the next order: the first ones come from the fact that our GOE blocks ${\bf H}+{\bf W}^{(a)}$ have size $N-1\times N-1$ but have variances rescaled with $N$; the other contributions are those normally arising even for GOE matrices of size $N\times N$. To distinguish such contributions we introduced a parameter $u$, in such a way that plugging $u=0$ gives only the second type of contributions, while using $u=1$ takes both of them into account. The analysis of the second type of contributions is already found in \cite{VERBAARSCHOT1984367}, in the standard case of $N\times N$ GOE matrices, where however $\mathbf{H}$ is fixed and not random as in our case. The additional computations are carried out in Appendix \ref{app:ExpectationProductResolvents}, where we took the results in \cite{VERBAARSCHOT1984367}, averaged over $\mathbf{H}$ and added the first type of contributions, multiplied by $u$. As a result we find:
\begin{align}
    \label{eq:corrections_prod}
\begin{split}
   \mathbb{E}[\mathbf{G}_0(z)\mathbf{G}_1(\xi)]=\Psi(z,\xi)+\frac{\bar \Psi^{(1)}(z,\xi)}{N}+\mathcal{O}\tonde{\frac{1}{N^{2}}}
\end{split}
\end{align}
with 
\begin{align}
\label{eq:finite_size_corr}
\begin{split}
&\bar \Psi^{(1)}(z, \xi)=   \bar \Lambda(z,\xi)  + \alpha (z) \partial_z \Psi(z,\xi)\\&+ \alpha(\xi) \partial_\xi \Psi(z,\xi)
+\beta(z) \partial^2_z \Psi(z,\xi) + \beta(\xi) \partial^2_\xi \Psi(z,\xi)
\end{split}
\end{align}
where  
\begin{align}\label{eq:CoefficientiLabda}
\begin{split}
    &\alpha (z)= \frac{\sigma_W^4}{2}\frac{\mathfrak{g}_\sigma''(z)}{[1-\sigma_W^2\mathfrak{g}_\sigma'(z)]^2}  - 
    \frac{\sigma_W^2}{1-\sigma_W^2\mathfrak{g}_\sigma'(z)} \times\\&\times\frac{ \mathfrak{g}_\sigma(z)}{1- \sigma^2 \mathfrak{g}_\sigma(z)} \tonde{\frac{\sigma^2 \mathfrak{g}^2_\sigma(z) }{1-\sigma^2 \mathfrak{g}^2_\sigma(z) }-u}\\
    &\beta(z)= \frac{\sigma_W^2}{2(1-\sigma_W^2\mathfrak{g}_\sigma'(z))}
    \end{split}
\end{align}
and 
\begin{align}
\label{eq:Lambda}
\begin{split}
&\bar \Lambda(z, \xi)=\frac{1}{{\xi-z -\sigma^2_W \mathfrak{g}_\sigma(\xi)+ \sigma_W^2\mathfrak{g}_\sigma(z)}}\Bigg(\frac{\sigma_H^2\mathfrak{g}^3_\sigma(z)}{[1-\sigma_H^2\mathfrak{g}^2_\sigma(z)]^2}\\&- \frac{u \sigma_H^2\mathfrak{g}^3_\sigma(z)}{1-\sigma_H^2\mathfrak{g}^2_\sigma(z)}
-\frac{\sigma_H^2\mathfrak{g}^3_\sigma(\xi)}{[1-\sigma_H^2\mathfrak{g}^2_\sigma(\xi)]^2}+ 
\frac{u \sigma_H^2\mathfrak{g}^3_\sigma(\xi)}{1-\sigma_H^2\mathfrak{g}^2_\sigma(\xi)}\Bigg).
\end{split}
\end{align}
Eq.~\eqref{eq:multi_resolv_formula0} then implies that to order $1/N$:
\begin{align}
\label{eq:multi_resolv_formula}
    {\bf \Pi}_{k,m}=\frac{(-1)^{k+m}}{k!m!}\frac{\partial^k}{\partial z^k} \frac{\partial^m}{\partial \xi^m } 
\quadre{\Psi(z,\xi)+\frac{\bar \Psi^{(1)}(z,\xi)}{N}}.
\end{align}

\noindent We now come to the $1/N$ corrections to 
 equation \eqref{eq:phi_bulk_bulk}. As we see from Eq.~\eqref{re_psi}, and from the terms that make up $\psi$, i.e. equations \eqref{eq:psi00Implicit}, \eqref{eq:psi0NImplicit} and \eqref{eq:psiNNImplicit},the only term that will give us contributions to the $1/N$ corrections of $\Phi\left(\lambda^0,\lambda^1\right)$ is ${N}^{-1}\Tr\mathbb{E}[\mathbf{G}_0(z)\mathbf{G}_1(\xi)]$. 
Then, the finite size contributions at order $1/N$ of $\Phi\left(\lambda^0,\lambda^1\right)$, that we denote as $\Phi^{(1)}\left(\lambda^0,\lambda^1\right)$, can be found as
\begin{equation}\label{eq:CorrBulk}
\begin{split}
&\Phi^{(1)}\left(\lambda^0,\lambda^1\right)=\frac{1}{2\pi^2\rho_\sigma(\lambda^0)\rho_\sigma(\lambda^1)}\times\\
    &\lim_{\eta\to 0}\text{Re}\left[\Psi^{(1)}(\lambda^0-i\eta,\lambda^1+i\eta)-\Psi^{(1)}(\lambda^0-i\eta,\lambda^1-i\eta)
    \right],
    \end{split}
\end{equation}
where $\Psi^{(1)}(z,\xi):=\bar \Psi^{(1)}(z,\xi)- u \Psi (z,\xi)$. The resulting expressions are rather long and cumbersome, and we do not report them for brevity. We nevertheless verified their exactitude by comparing with numerical simulations, see Fig.~\ref{fig:FiniteSize}.

\begin{figure}[h!]
\centering
\includegraphics[width=0.48\textwidth]
{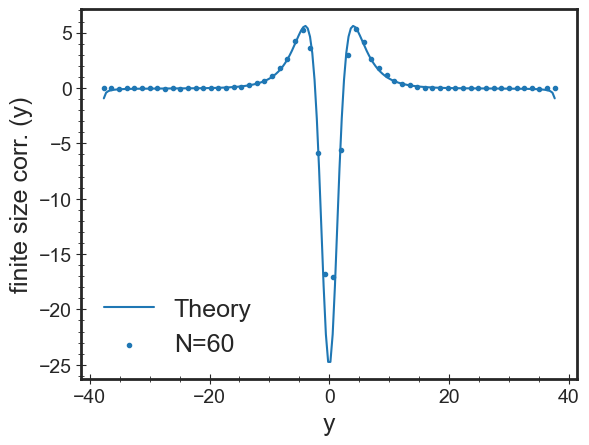}
\includegraphics[width=0.45\textwidth]
{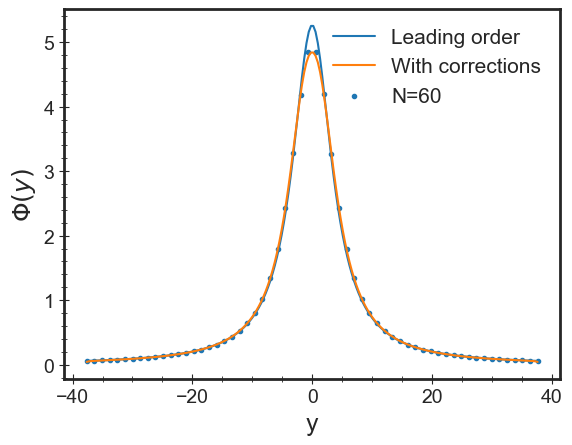}
\caption{{\it Top.} $1/N$ corrections to the bulk-bulk overlap $\Phi(x,y)$, see Eq. \eqref{eq:phi_bulk_bulk}. The plot corresponds to $x=0$ (center of the spectrum of the first matrix) and $y$ ranging through the bulk of the spectrum of the second matrix. The numerical points are obtained by diagonalizing pairs of random matrices of size $N=60$. The parameters are $\sigma_H=18$, $\sigma_W=6$, $\Delta_h=\sigma_H$, $\Delta_{w,0}=\Delta_{w,1}=\sigma_W$, $\mu_0=\mu_1=0$ and $v_0=v_1=0$.   {\it Down.} Respective comparison between the leading order term of the overlap $\phi(y):=\Phi(0,y)$, and the same quantity including the subleading $1/N$ corrections. }
\label{fig:FiniteSize}
\end{figure}

\section{Two special cases}\label{sec:applications}

\subsection*{Repeated signal {\it vs} noise measurements with coupled noise: correlations of the estimators}

Wee consider in this section the case of a purely additive rank-1 perturbations to the GOE matrices. In our setting, this corresponds to choosing  $\Delta_h=v_h=\sigma_H$, and $\Delta_{w,a}=v_{w,a}=\sigma_W$ for both $a=0,1$. This setting has a clear interpretation as a denoising problem: the perturbed matrices \eqref{eq:MatRiscritte} can in fact be written in this case as:
\begin{equation}\label{eq:addPert}
{\bf M}^{(a)}= {\bf X}^{(a)} + \mu_a\, {\bf e}_N {\bf e}_N^T,
\end{equation}
where $  {\bf X}^{(a)}$ are $N \times N$  GOE matrices with variance $\sigma^2$ identified with noise, while the rank-1 projector onto the unit vector ${\bf e}_N$ is identified with the signal ($\mu_a/\sigma$ being referred to as the \emph{signal-to-noise ratio}). Let us consider only one element of the pair, and drop the superscript $a$. In the context of denoising, the relevant question is whether (for which values of $\mu$) having access only to the matrix ${\bf M}$ and assuming that the unit vector ${\bf e}_N$ is unknown, one is able to recover some information on the signal  ${\bf e}_N$, \emph{i.e.} on the direction it identifies on the $N$-dimensional unit sphere. In the limit of large $N$, this problem is known to exhibit a sharp transition at a critical value $\mu_c$:  detecting the presence of the signal is possible only for $|\mu| \geq  \mu_c$ \cite{perry2018optimality}. Moreover, in the case of dense Gaussian matrices perturbed as \eqref{eq:addPert} $\mu_c$ coincides exactly with the critical $\mu$ at which the matrices exhibit the BBP spectral transition, i.e. $\mu_c=\sigma$: for $|\mu| < \mu_c$ the eigenvalue of the matrices are distributed with a continuous density (given by the semicircle law) supported in the finite interval $[-2 \sigma, 2 \sigma]$, while for $|\mu|>\mu_c$ the isolated eigenvalue exists. This spectral criterion is often referred to as matrix PCA. 
For $|\mu|>\mu_c$, the eigenvector ${\bf u}_{\rm iso}$ associated to the isolated eigenvalue is a statistical estimator of the signal ${\bf e}_N$: its overlap  $({\bf u}_{\rm iso} \cdot {\bf e}_N)^2$ with the signal remains of $\mathcal{O}(1)$ when $N \to \infty$ (its typical value is given in \eqref{eq:OvPCA}), and thus  ${\bf u}_{\rm iso}$ provides some information on the position of the signal on the $N$-dimensional sphere. This information becomes exact in the limit $\mu \to \infty$, when the overlap converges to one and the signal can be exactly recovered.

\begin{figure}[h!]
\centering
\includegraphics[width=0.48\textwidth]
{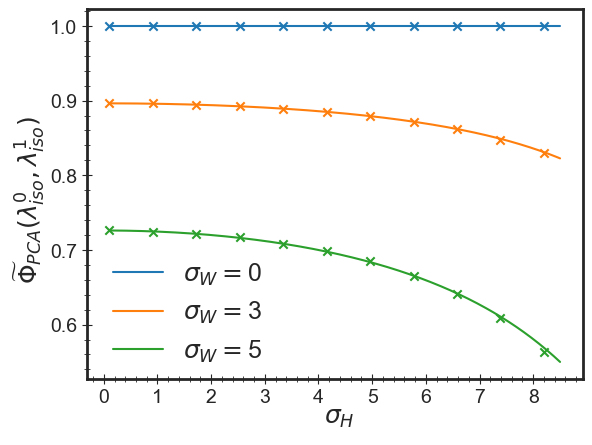}
\includegraphics[width=0.48\textwidth]{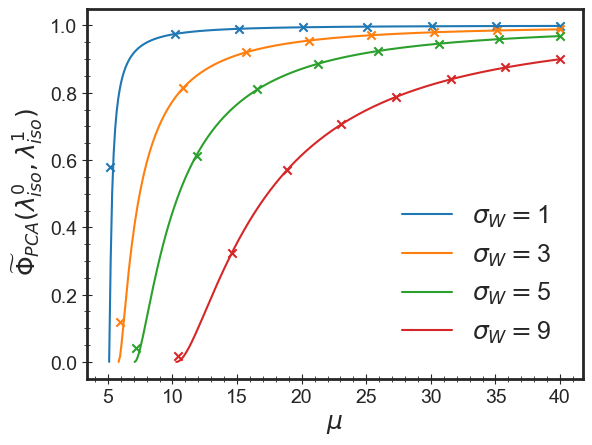}
\caption{
Overlap of the eigenvectors associated to the isolated eigenvalues of matrices subject to purely additive perturbations ($\Delta_h=\sigma_H$, $\Delta_{w,0}=\Delta_{w,1}=\sigma_W$) with $\mu=\mu_0=\mu_1$. The points are obtained from direct diagonalization of matrices of size $N=600$ averaged over $1500$ realization, while the continuous curves correspond to Eq.~\eqref{eq:PhiPCA}. As for Fig.\ref{fig:double_evals} we set $v_0=v_1=0$ given that the final results do not depend on them, to leading orders. {\it Top.}  Overlap as a function of $\sigma_H$,  for various $\sigma_W$ and $\mu_0=\mu_1=13$. {\it Bottom.} Overlap as a function of $\mu$, for  various $\sigma_W$ and $\sigma_H=5$.}
\label{fig:phi_iso_iso_pca_vs_sigmah_sigmaw}
\end{figure}

Consider now the case in which pairs of  matrices ${\bf M}^{(a)}$ of the form \eqref{eq:addPert} are given, which differ from each others only by the  fluctuations in the noisy component ${\bf X}^{(a)}$ (thus $\mu_0=\mu=\mu_1$), the noise being correlated as described in Sec.~\ref{subsec:matrix}. Such pairs may correspond to  measurements performed at different times between which the noise has changed partially, without decorrelating completely with the previous configuration. At both times the estimator of the signal is given by the eigenvector ${\bf u}^{a}_{\rm iso}$ associated to the isolated eigenvalue of the spiked matrix. The correlation in the noisy components of the matrices implies that estimators ${\bf u}^a_{\rm iso}$ will have a non-trivial overlap with each other, which corresponds to $\tilde \Phi(\lambda^0_{\rm iso},\lambda^1_{\rm iso})$. This function then quantifies the typical similarity between the estimators  ${\bf u}^{a}_{\rm iso}$ of the signal, obtained from different measurements of the signal corrupted by correlated noise. 

For a purely additive rank-1 perturbation the isolated eigenvalue reads $\lambda_{\rm iso}= \mu + \sigma^2 \mu^{-1}$, and $\mathfrak{q}_{\sigma,\Delta}(\lambda_{\rm iso},\mu)= 1- \sigma^2 \, \mu^{-2}$. 
The overlap \eqref{eq:phi_iso_iso} in this particular limit reduces to:
\begin{equation}\label{eq:PhiPCA}
\begin{split}
&\tilde \Phi_{\rm PCA}(\lambda_{\rm iso},\lambda_{\rm iso})=\tonde{1-\frac{\sigma^2}{\mu^2}}^2 \, \bigg[\sigma_H^2\omega(\lambda_{\rm iso})
    +1\bigg]^2,
 \end{split}
\end{equation}
where
\begin{align}\label{eq:omega}
\omega(z):=\lim_{z\to\xi}\Psi(z,\xi)=\frac{\mathfrak{g}_\sigma(z)}{z+(\sigma_W^2-2\sigma^2)\mathfrak{g}_\sigma(z)}.
\end{align}

In Fig.~\ref{fig:phi_iso_iso_pca_vs_sigmah_sigmaw} we compare this expression with the overlaps obtained from the direct diagonalization of the random matrices, for different values of $\sigma_W$. As expected, at fixed $\sigma_H$ the overlap is equal to one in the case of fully correlated noise ($\sigma_W=0$), and decreases monotonically with the strength $\sigma_W$ of the uncorrelated part of the noise. At fixed $\sigma_W$, the overlap also decreases with increasing $\sigma_H$, as the relative contribution of the noise $\sigma=(\sigma_W^2 + \sigma_H^2)^{1/2}$
with respect to the signal $\mu$ increases. For $\sigma_H=0$, the noise in the two sets of measurements is uncorrelated and the overlap converges to the square of \eqref{eq:OvPCA}. As discussed in Sec. \ref{sec:theoretical_res2}, this corresponds to the fact that the estimators ${\bf u}^{a}_{\rm iso}$ are orthogonal in the subspace orthogonal to the signal direction ${\bf e}_N$.

\subsection*{Hessians of random landscapes:\\ correlations in the landscape curvature}
The analysis presented in this work is motivated by the study of the geometrical properties of high-dimensional random landscapes with Gaussian statistics. Random functions defined in high-dimensional configuration space emerge in a variety of contexts. A prototypical example is given by functions $\mathcal{E}[{\bf s}]$ parametrized as: 
\begin{equation}\label{eq:Enp}
\mathcal{E}[{\bf s}]= \sqrt{\frac{D p!}{2}}\sum_{i_1 <  i_2 \cdots <i_p} a_{i_1 \, i_2 \cdots i_p} s_{i_1} s_{i_2} \cdots s_{i_p},
\end{equation}
where ${\bf s}=(s_1, \cdots, s_D)$ belongs to a manifold with a simple topology, such as the unit sphere ($\sum_{i=1}^D s_i^2=1$).
In the simplest case, the coefficients $ a_{i_1 \, i_2 \cdots i_p}$ are chosen to be independent, centred Gaussian variables with unit variance. The value of the landscape at different configurations ${\bf s}_0$, ${\bf s}_1$ is correlated as:
\begin{equation}\label{eq:EnpC}
\langle \mathcal{E}[{\bf s}_0]  \mathcal{E}[{\bf s}_1]\rangle = \frac{D}{2} \,  \tonde{{\bf s}_0 \cdot {\bf s}_1}^p.
\end{equation}

For $p \geq 3$, typical realizations of this random landscape exhibit an exponentially large (in $D$) number of minima, maxima and saddles, which are stationary points where the landscape is locally flat (where the gradient of~\eqref{eq:Enp} vanishes); the landscape is therefore highly non-convex, or { glassy}.  Characterizing the distribution of stationary points in high-dimensional random landscapes is relevant to understand how these landscapes are explored by local optimization algorithms. For models of the form \eqref{eq:Enp}, the large-$D$ scaling of the typical number of stationary points at fixed energy density $\epsilon= \lim_{D \to \infty} D^{-1}\mathcal{E}$ has been determined in the early works \cite{crisanti1992sphericalp, cavagna1998stationary, cavagna1999quenched}, and the resulting expressions are by now known with a mathematical level of rigor~\cite{auffinger2013random}. Subsequent works \cite{cavagna1997structure, subag2017complexity,subag2017complexity} have investigated the distribution of pairs of stationary configurations ${\bf s}_0$, ${\bf s}_1$ as a function of their similarity or overlap 
\begin{equation}\label{eq:Overlap}
q=\lim_{D\to \infty} {\bf s}_0 \cdot {\bf s}_1.
\end{equation}
Interesting questions concerning the correlation between such stationary points are however still open. In particular, one may be interested in understanding how the {\it curvature} of the random landscape in the surroundings of its stationary points (which encodes for their linear stability) is correlated, as a function of the energy density of the points and of their proximity $q$ in configuration space. This piece of information turns out to be crucial to characterize profiles of the random landscape along paths interpolating between different local minima, or more generally different configurations of the system. We discuss in detail this application in a forthcoming work, and here limit ourselves to commenting on the connections to the random matrix problem discussed here.

The local curvature of $\mathcal{E}[{\bf s}]$ around a configuration ${\bf s}$ is described by the Hessian matrix $\mathcal{H}[{\bf s}]$ of the landscape, which is a random matrix whose statistics depends on the constraints imposed on the configuration  ${\bf s}$ -- for example, the constraint of being a stationary point having a given energy density. Due to the spherical constraint defining the space of configurations, the Hessians matrices at different points ${\bf s}$ are defined on different ${\bf s}$-dependent subspaces of dimension $(D-1)$,  which are the $(D-1)$-dimensional tangent planes to the sphere at the configurations ${\bf s}$. It can be shown that pairs of Hessian matrices at two 
stationary points ${\bf s}_0, {\bf s}_1 $ have statistical properties strongly related to those of the matrices considered in this work. More precisely, consider two stationary points ${\bf s}_0, {\bf s}_1 $ at overlap $q$ and having energy densities $\epsilon_0, \epsilon_1$.
Let $\tau[{\bf s}_a]$ denote the tangent plane associated to each stationary point. One can choose a suitable orthonormal basis in each tangent plane, with respect to which  the rescaled Hessians can be written as 
\begin{equation}\label{eq:Hesshift}
 \frac{1}{\sqrt{D-1}}\mathcal{H}[{\bf s}_a]= {\bf M}^{(a)}- \sqrt{  \frac{2 D}{D-1}} p \epsilon_a \mathbbm{1}  
\end{equation}
where the ${\bf M}^{(a)}$ are $(D-1) \times (D-1)$ matrices of the form \eqref{eq:MatrixForm}, while $\mathbbm{1}$ is the identity matrix. We set $N=D-1$.
To have such a representation, the basis of the tangent plane $\tau[{\bf s}_0]$ has to be chosen in such a way that one vector, say ${\bf e}_{D-1}[{\bf s}_0]$, is aligned along the direction connecting the two configurations:
\begin{equation}
{\bf e}_{D-1}[{\bf s}_0]= \frac{1}{\sqrt{1-q^2}}({\bf s}_1- q {\bf s}_0),
\end{equation}
while all other vectors ${\bf e}_{i=1, \cdots, D-2}$ span the subspace that is orthogonal to both ${\bf s}_0, {\bf s}_1$. Similarly, the basis of $\tau[{\bf s}_1]$ is chosen in such a way that the first $D-2$ vectors coincide with the  ${\bf e}_{i=1, \cdots, D-2}$ chosen above, while 
\begin{equation}
{\bf e}_{D-1}[{\bf s}_1]= \frac{1}{\sqrt{1-q^2}}({\bf s}_0- q {\bf s}_1).
\end{equation}
When expressed in the corresponding basis, each matrix ${\bf M}^{(a)}$ is made of a $(D-2) \times (D-2)$ dimensional block that has a GOE-like statistics (invariant under rotation of the basis in the corresponding subspace), and of a row and column that are special. The special lines are associated to the basis vectors  ${\bf e}_{D-1}[{\bf s}_a]$, which are aligned along the direction connecting the two stationary points in configuration space. 
This block structure is a consequence of the fact that the statistics of the landscape, encoded in the correlation function \eqref{eq:Enp}, is isotropic; if no constraint was imposed on the ${\bf s}_a$, the   statistics of the Hessians would be fully rotational invariant. The constraint of the overlap breaks such an invariance as it singles out one special direction (the one connecting the two configurations), along which the statistics is perturbed.

We now discuss how the correlations of the entries of the ${\bf M}^{(a)}$ depend on the parameter $p$ characterizing the structure of the random landscape, as well as on the parameters $q, \epsilon_a$ with $a=0,1$ defining the properties of the stationary points.
This has been determined explicitly  in  \cite{ros2019complexity, ros2020distribution} (see also Lemma 13 in \cite{subag2017complexity}).  In the notation of Sec.~\ref{sec:theoretical_res}, one finds: 
\begin{equation}\label{eq:VarHEssGOE}
\sigma^2_H= p(p-1) q^{p-2}, \quad \quad  \sigma^2_W= p(p-1)[1-q^{p-2}],
\end{equation}
which fully specify the statistics of the GOE blocks. 
The statistics of the special row and column is described by the parameters $\Delta_h, \Delta_\omega, v_h, v_\omega$ and $\mu_a$. Since nothing in the above calculation depends explicitly on $v_h, v_w$, we can neglect the corresponding expressions. One finds:
\begin{equation}
\begin{split}
&\Delta_h^2=p(p-1)q^{p-3}\left[ (p-2)-(p-1) q^2 \frac{1-q^{2p-4}}{1-q^{2p-2}} \right]\\
&\Delta_w^2=p(p-1)\left[\frac{q^{2 p}+(p-2) (1-q^2) q^{p+1}-q^4}{q^3 \left(q^p-q\right)} \right]
\end{split}
\end{equation}
which implies:
\begin{equation}
\Delta^2=p(p-1)\left[1-\frac{(p-1)(1-q^2)q^{2p-4}}{1-q^{2p-2}}\right].
\end{equation}
The fluctuations of the elements $m^a_{i N}$ for $i<N=D-1$ are thus determined uniquely by $p$, and by the overlap $q$. On the other hand, the dependence on the energies $\epsilon_a$ enters in the averages $\mu_a$. We have:
{\medmuskip=1mu
\thinmuskip=1mu
\thickmuskip=1mu
\begin{equation}
 \mu_0=\frac{ (p-1) p \left(1-q^2\right) \left(a_0(q) \epsilon_0-a_1(q) \epsilon_1 \right)}{q^{6-p}+q^{3 p+2}- q^{p+2}\left((p-1)^2 (q^4+1)-2 (p-2) p q^2\right)}
\end{equation}
}
with
\begin{equation}
\begin{split}
 a_1(q)&=q^{3 p}+ q^{p+2}\left(p-2-(p-1) q^2\right)\\
 a_0(q)&= q^4+ q^{2 p}\left(1-p +(p-2)q^2\right),
   \end{split}
\end{equation}
and 
{\medmuskip=1mu
\thinmuskip=1mu
\thickmuskip=1mu
\begin{equation}
 \mu_1=\frac{ (p-1) p \left(1-q^2\right) \left(a_0(q) \epsilon_1-a_1(q) \epsilon_0 \right)}{q^{6-p}+q^{3 p+2}- q^{p+2}\left((p-1)^2 (q^4+1)-2 (p-2) p q^2\right)}.
\end{equation}
}

We observe that these formulas describe the fluctuation of the entries of the two Hessians expressed in different bases, differing by the last vector ${\bf e}_{D-1}[{\bf s}^a]$. The derivation of the overlap formula given above assumes however that both matrices are expressed in the same basis. It follows that when applied to this Hessian problem, $\Phi(x,y)$ gives the square of the overlap between eigenvectors shifted by a quantity (related to the components of the eigenvectors along the directions ${\bf e}_{D-1}[{\bf s}^a]$), as we discuss in more detail in Appendix \ref{app:RotationLAndscape}. 
In Fig.~\ref{fig:p_spin} 
 we show the quantity $\Phi(\lambda_{\rm iso}^0,y)$, which is related to the overlap between the isolated eigenvector of ${\bf M}^{(0)}$ and the eigenvectors associated to eigenvalues in the bulk of ${\bf M}^{(1)}$. The plots are given for $p=3$ and for fixed overlap $q$ and energy density $\epsilon_0$ of the first stationary point. For the chosen values of $\epsilon_1$, the Hessian \eqref{eq:Hesshift} at ${\bf s}_0$ has a single negative mode given by the isolated eigenvalue, while the Hessian at ${\bf s}_1$ has either an extensive number of negative modes (main panel) or no negative modes (inset). As the energy density $\epsilon_1$ of the second stationary point decreases (getting closer to $\epsilon_0$), the peak in the overlap shifts towards the lower edge of the support of the eigenvalue density of ${\bf M}^{(1)}$, indicating that the direction of the isolated mode of one Hessian becomes progressively more correlated with the smallest modes of the other Hessian.
 
\begin{figure}[h!]
\centering
\includegraphics[width=0.48\textwidth]
{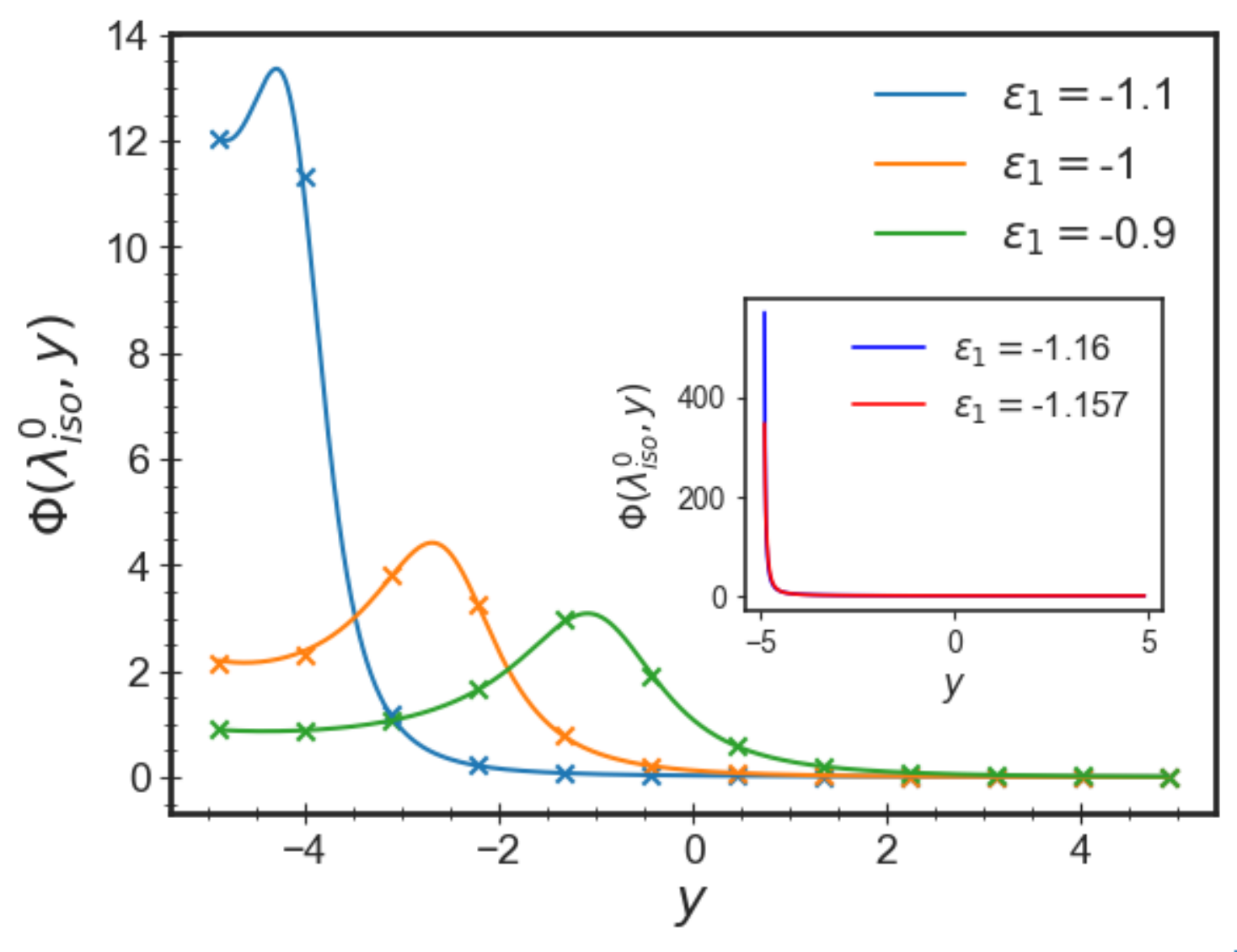}
\caption{Profile of $\Phi(\lambda_{\rm iso}^0,y)$ for Hessian matrices at two stationary points of the landscape \eqref{eq:Enp}, with  $\epsilon_0=-1.167$, $q=0.72$. {\it Inset.} $\Phi(\lambda_{\rm iso}^0,y)$ for $q=0.67$ and smaller values of $\epsilon_1$. }
\label{fig:p_spin}
\end{figure}

\section{Conclusion}\label{sec:conclusions}
\noindent We have considered pairs of correlated GOE matrices deformed with additive and multiplicative rank-1 perturbations, giving rise to outliers eigenvectors in their spectrum. We have determined the explicit expression of the overlap between the eigenvectors of the two matrices: in particular, we have derived expressions for the overlaps between the outlier eigenvector of one matrix and arbitrary eigenvectors (bulk or outliers) of the other matrix, see Eqs. \eqref{eq:phi_iso_iso} and \eqref{eq:phi_iso_bulk}. Moreover, we have generalized the results of Ref.~\cite{bun2018overlaps} by computing the subleading corrections to the overlap between eigenvectors belonging to the bulk of the two matrices.  Our analysis includes the special case of correlated GOE matrices perturbed by an additive signal term, a case widely studied in the literature and often referred to as matrix PCA. We have shown how in this case the overlap between the estimators of the signal takes a particularly simple form, see Eq.~\ref{eq:PhiPCA}, and quantifies the correlation between estimators obtained from different sets of measurements with correlated noise. As an intermediate result, we have determined the finite-size corrections to the expectation of the product of resolvents of correlated GOE matrices, see Eqs. 	\eqref{eq:multi_resolv_formula0} and \eqref{eq:corrections_prod}, generalising known results for the leading-order term~\cite{cipolloni2022thermalisation}. 

\noindent We remark that similar questions concerning overlaps between outliers have been considered in previous literature for a rather broad class of covariance matrices. In particular,  the overlaps between the eigenvectors of population (averaged) covariance matrices  and those of sample covariance matrices have been discussed in Refs. \cite{
bun2016rotational, bun2017cleaning}, and cases involving outliers eigenvectors are discussed in Refs. \cite{bun2018optimal, bloemendal2016principal}. For pairs of sample covariance matrices, the problem is discussed in \cite{bun2018overlaps}. However, we are not aware of results involving the overlap between outliers of pairs of correlated, spiked covariance matrices, so this remains a direction for future work. The extension of our results to the case of complex matrices with GUE (Gaussian Unitary Ensemble) statistics can be performed rather straightforwardly following the same steps presented in this work, and we also leave it to future work.

\noindent The matrix ensembles considered in this work describe the statistical properties of the curvature of high-dimensional, Gaussian random landscapes. The overlap between eigenvectors computed in this work give the correlations between the eigenvectors of the  Hessian matrices describing the landscape curvature; therefore, our result allows us to determine, for instance, how the softest modes at two nearby stationary points of the landscapes are oriented with respect to each others. This piece of information is relevant to determine, for example, energy profiles along paths interpolating between minima and saddles in the high-dimensional configuration space. This is the subject of ongoing work.

\subsection*{Acknowledgements}
\noindent We thank two anonymous reviewers for their constructive input. VR acknowledges funding by the ``Investissements d’Avenir” LabEx PALM (ANR-10-LABX-0039-PALM).

\bibliography{apssamp.bib}

\newpage

\onecolumngrid
\appendix

\vspace{1 cm}
\section{spectral properties of the perturbed random matrix ensemble}\label{app:Density}
In this Appendix we discuss the spectral properties of random  matrices with the statistics given in \eqref{eq:MatrixForm}. In particular, we derive their averaged spectral measure and determine the conditions on our parameters (i.e. $\sigma,\Delta, \mu$) that guarantee the existence of {isolated} eigenvalue(s). {We remark that in the present section we are interested in properties of a single matrix of the form  \eqref{eq:MatrixForm}, averaged over the ensemble. More precisely, this means that we consider an $N\times N$ random matrix of the form
\begin{align*}
    \mathbf{M}=
    \begin{pmatrix}
     & && & m_{1\,N}\\
     &  &{\bf B}  & & \vdots\\
      & & & & m_{N-1\,N}\\
     m_{1\,N} && \ldots & m_{N-1\,N} &m_{N\,N}
    \end{pmatrix}
\end{align*}
with ${\bf B}$ an $N-1\times N-1$ GOE matrix of variance $\sigma^2$, and $m_{i,N}\sim \mathcal{N}(0,\Delta^2/N)$, for $i=1,\ldots,N-1$, with $m_{NN}\sim\mathcal{N}(\mu,v^2/N)$. Given that the quantities we discuss in the following are independent of the particular choice of $v$ (to the order in $N$ that we are interested in), we set $v=0$ from the start.}

\subsection{Perturbed matrix ensemble: Stieltjes transform}
\label{app:stieltjes}
Given a matrix with spectral measure $\nu(u)$ (which may contain a continuous and a discrete part), we denote with
 $\mathfrak{g}(z)$ the Stieltjes transform: 
 \begin{equation}\label{eq:Sti}
\mathfrak{g}(z)=\int\frac{d \nu(u)}{z-u}.
\end{equation}
This function has singularities on the real line at points belonging to the spectrum of the matrix: a branch cut in correspondence to the continuous part of the eigenvalue density, and poles wherever isolated eigenvalues exist. For GOE matrices ${\bf M}_{\rm GOE}$ with variance $\sigma$, the transform \eqref{eq:Sti} depends only on $\sigma$ and reads \cite{Meh2004}:
 \begin{equation}\label{eq:StjltiesGOE}
     \mathfrak{g}_\sigma(z)= \frac{z- \text{sign}(\Re z)\sqrt{z^2 - 4 \sigma^2}}{2 \sigma}, \quad \quad z \notin [-2 \sigma, 2 \sigma]
 \end{equation}
 where the sign in front of the square root is chosen to guarantee that  $ \mathfrak{g}_\sigma(z) \to 0$ when $|z| \to \infty$. This formula coincides with \eqref{eq.Stjlt} when $z$ is taken to be real. In order to find the average spectral measure $\nu_N(x)$ of the $N \times N$ matrix $\mathbf{M}$, we exploit the inversion of \eqref{eq:Sti}:
\begin{align}\label{eq:InvStj}
    \nu_N(x)=\frac{1}{\pi}\lim_{\eta\to0^+}\text{Im}\left[\mathfrak{g}(x-i\eta)\right],
\end{align}

where we can write \eqref{eq:Sti} also as 
\begin{align}
\label{app:sokholski}
    \mathfrak{g}(z)=\frac{1}{N}\mathbb{E}\left[\Tr(z-\mathbf{M})^{-1}\right]=\frac{1}{N}\mathbb{E}\left[\sum_{i=1}^{N-1}(z-\mathbf{M})^{-1}_{ii}+(z-\mathbf{M})^{-1}_{NN}\right].
\end{align}
Given the block structure of the matrix ${\bf M}$, the components of the resolvent matrix can be obtained making use of the matrix inversion lemma:
\begin{align}
\begin{bmatrix}
    \mathbf{A} & \mathbf{B}\\
    \mathbf{C} & \mathbf{D}
    \end{bmatrix}^{-1}=
    \begin{bmatrix}
    (\mathbf{A}-\mathbf{B}\mathbf{D}^{-1}\mathbf{C}) ^{-1} & -(\mathbf{A}-\mathbf{B}\mathbf{D}^{-1}\mathbf{C})^{-1}\mathbf{B}\mathbf{D}^{-1}\\
    - \mathbf{D}^{-1}\mathbf{C}(\mathbf{A}-\mathbf{B}\mathbf{D}^{-1}\mathbf{C})^{-1}\quad\quad & \mathbf{D}^{-1} + \mathbf{D}^{-1}\mathbf{C}(\mathbf{A}-\mathbf{B}\mathbf{D}^{-1}\mathbf{C})^{-1}\mathbf{B}\mathbf{D}^{-1}
    \end{bmatrix}.
\end{align} 
By using the formulas \eqref{eq:Block1}, \eqref{eq:Block2} in the main text and by averaging first over the components $m_{iN}$ for $i <N$ (see also the subsequent appendices), we obtain 
\begin{equation}\label{eq:WeC}
\mathbb{E} \quadre{\sum_{i=1}^{N-1}\tonde{z-\mathbf{M}}^{-1}_{ii}}=\mathbb{E} \quadre{\sum_{i=1}^{N-1}\tonde{z-\mathbf{H}-\mathbf{W}}^{-1}_{ii}
+ \frac{1}{N} \sum_{i=1}^{N-1}
(z-\mathbf{H}-\mathbf{W})^{-2}_{ii} \frac{\Delta^2}{z-\mu -\Delta^2 \frac{1}{N} \text{Tr} \tonde{\frac{1}{z-\mathbf{H}-\mathbf{W}}} }}+ \mathcal{O}\tonde{\frac{1}{N^2}}.
\end{equation}
At first, notice that to leading order in $N$, we have:
\begin{equation}
\frac{1}{N}\mathbb{E} \quadre{ \frac{1}{N} \sum_{i=1}^{N-1}
(z-\mathbf{H}-\mathbf{W})^{-2}_{ii} \frac{\Delta^2}{z-\mu -\Delta^2 \frac{1}{N} \text{Tr} \tonde{\frac{1}{z-\mathbf{H}-\mathbf{W}}} }}=
- \frac{1}{N} \frac{\Delta^2}{z-\mu -\Delta^2 \mathfrak{g}_\sigma(z)} \partial_z \mathfrak{g}_\sigma(z)+ \mathcal{O}\tonde{\frac{1}{N^2}},
\end{equation}
where $\mathfrak{g}_\sigma(z)$ is the transform of the GOE matrix $\mathbf{H}+\mathbf{W}$, and we made use of the identity:
\begin{equation}
\frac{1}{N}\mathbb{E} \quadre{\sum_{i=1}^{N-1}
(z-\mathbf{H}-\mathbf{W})^{-2}_{ii} }=-\partial_z \frac{1}{N}\mathbb{E} \quadre{\sum_{i=1}^{N-1}
(z-\mathbf{H}-\mathbf{W})^{-1}_{ii} }= - \partial_z \mathfrak{g}_\sigma(z)+  \mathcal{O}\tonde{\frac{1}{N}}.
\end{equation}

We now focus on the first term in the right-hand side of \eqref{eq:WeC}. To leading order, this term is $\mathfrak{g}_\sigma(z)$. Two types of $1/N$ corrections contribute to this term: one coming from the fact that the sum is only over $N-1$ (and not $N$) matrix elements, and one coming from the $1/N$ corrections of GOE resolvents already determined in \cite{VERBAARSCHOT1984367}. In order to distinguish between these terms, we multiply the first set of corrections by a factor $u$ and eventually take $u \to 1$ at the end of the calculation. Adapting the derivation of \cite{VERBAARSCHOT1984367} to the perturbed case we obtain:
\begin{equation}\label{eq:SingleRes}
\frac{1}{N}\mathbb{E} \quadre{\sum_{i=1}^{N-1}
(z-\mathbf{H}-\mathbf{W})^{-1}_{ii} }= \mathfrak{g}_\sigma(z)+ \frac{1}{N} \quadre{\frac{z- \sqrt{z^2-4 \sigma^2}}{2[z^2- 4 \sigma^2]} - \frac{u \sigma^2 \mathfrak{g}_\sigma^3(z) }{1- \sigma^2 \mathfrak{g}^2_\sigma(z)}   - u \mathfrak{g}_\sigma(z)}+\mathcal{O}\tonde{\frac{1}{N^2}}.
\end{equation}
In here, the first contribution to the $1/N$ corrections is the one determined in \cite{VERBAARSCHOT1984367}; the first term proportional to $u$ arises from the fact that we are considering matrices of size $N-1$ with variances normalized by a factor $N$, while the second term proportional to $u$ is due to the fact that we are normalizing by $N$ the sum over $N-1$ components. Proceeding as above we also find:
\begin{equation}
\frac{1}{N}\mathbb{E}\left[(z-\mathbf{M})^{-1}_{NN}\right]=\frac{1}{N}\frac{\Delta^2}{z-\mu-\Delta^2\mathfrak{g}_{\sigma}(z)}\frac{\mathfrak{g}_\sigma(z)}{z-\mu}+  \mathcal{O}\tonde{\frac{1}{N^2}}, 
\end{equation}
and combining everything we finally obtain:
\begin{align}\label{eq:FullG}
    \mathfrak{g}(z)=\mathfrak{g}_\sigma(z)+\frac{1}{N} \quadre{\frac{z- \sqrt{z^2-4 \sigma^2}}{2[z^2- 4 \sigma^2]}- \frac{u}{\sqrt{z^2- 4 \sigma^2}}}+\frac{1}{N}\frac{\Delta^2}{z-\mu-\Delta^2\mathfrak{g}_{\sigma}(z)}\left[\frac{\mathfrak{g}_\sigma(z)}{z-\mu}-\partial_z\mathfrak{g}_{\sigma}(z)\right]+ \mathcal{O} \tonde{\frac{1}{N^2}}.
\end{align}
The spectral measure can then be obtained using \eqref{eq:InvStj}. To leading order, one recovers the GOE density \eqref{eq:DoSgoe}. The first contribution to the $1/N$ correction, denoted with \eqref{eq:Rho1} in the main text, is obtained from:
\begin{equation}
\rho^{(1)}_{\sigma}(x)=\frac{1}{\pi}\lim_{\eta \to 0^+}\text{Im } \left[\frac{z- \sqrt{z^2-4 \sigma^2}}{2[z^2- 4 \sigma^2]}- \frac{u}{\sqrt{z^2- 4 \sigma^2}}\right]_{z=x-i \eta}
\end{equation}
with $u=1$ (instead, setting $u=0$ would amount to consider matrices of size $N\times N$) . The term within brackets exhibits a branch cut in the region $z \in [-2 \sigma, 2 \sigma]$, and two poles at the boundaries of the interval.  Thus, this term gives rise to $1/N$ corrections to the continuous eigenvalue density plus two delta peaks at the boundaries.  The second set of $1/N$ corrections arises from the second term in \eqref{eq:FullG}, which exhibits poles at the solutions of the equation \eqref{egval_eqn}, which we recall reads $z-\mu-\Delta^2\mathfrak{g}_{\sigma}(z)=0$. 
The real solutions $\lambda_{\rm iso, \pm}$ of this equation, whenever they exist, are the isolated eigenvalues of the matrix. We discuss extensively the conditions for their existence in the following subsection.

\subsection{Phenomenology of the isolated eigenvalue(s)}
\label{app:phenom_isolated_eigenvals}
Isolated eigenvalues are real solutions $z \to \lambda$ of \eqref{egval_eqn}. By using \eqref{eq:StjltiesGOE} we can easily rewrite the equation as
\begin{align}
\label{app:sign_condition}
\lambda\left(1-\frac{\Delta^2}{2\sigma^2}\right)-\mu=-\text{sign}(\lambda)\frac{\Delta^2}{2\sigma^2}\sqrt{\lambda^2-4\sigma^2}
\end{align}
from which we can take the square on both sides, keeping in mind that the equation to be satisfied by our final solution is \eqref{app:sign_condition}, with the proper sign on the right-hand side. Taking the square we get
\begin{align*}
\lambda^2\left(1-\frac{\Delta^2}{\sigma^2}\right)-2\mu\left(1-\frac{\Delta^2}{2\sigma^2}\right)\lambda+\mu^2+\frac{\Delta^4}{\sigma^2}=0
\end{align*}
which gives us the two solutions
\begin{align}
\label{app:solutions}
\lambda_{\rm iso,\pm}(\mu, \Delta, \sigma)= \frac{{2 \mu \sigma^2- \Delta^2 \mu\pm \text{sign}(\mu) \Delta^2 \sqrt{\mu^2-4 (\sigma^2- \Delta^2)}}}{2 (\sigma^2-\Delta^2)}.
\end{align}
From this expression we get that the condition 
\begin{align}
\label{app:determinant_condition}
\mu^2-4(\sigma^2-\Delta^2)\geq 0
\end{align}
must be satisfied for the solutions to exist on the real line. We now discuss several cases for the parameters.\\

\underline{ \bf Case $\Delta= \sigma$. }

This case has been investigated already in the early work \cite{edwards1976eigenvalue}. It corresponds to a GOE matrix perturbed by a rank-1 additive term. The solution $\lambda_{\rm iso, +}$ diverges in this limit, while the solution $\lambda_{\rm iso, -}$ converges to Eq. \eqref{eq:IsoPCA}, that is
\begin{align*}
    \lambda_{\rm iso}= \lambda_{\rm iso,-}=\mu+\frac{\sigma^2}{\mu}.
\end{align*}
 Both conditions \eqref{app:sign_condition}, \eqref{app:determinant_condition} are satisfied provided that $|\mu|\geq \sigma$. The (only) isolated eigenvalue thus exists (i.e. it is bigger than $2\sigma$ in absolute value) for any $|\mu|>\sigma$.\\

\underline{ \bf Case $\Delta< \sigma$. }

This case has been already discussed in \cite{ros2019complexity,ros2020distribution}, and here we re-derive those results.
The isolated eigenvalues exists whenever at least one among $\lambda_{\rm iso, \pm}$ is bigger than $2\sigma$ in absolute value and the conditions \eqref{app:sign_condition}, \eqref{app:determinant_condition} are both satisfied. We notice that in this setting, if they exist, the eigenvalues satisfy $\text{sign}(\lambda_{\rm iso, \pm})=\text{sign}(\mu)$. In order to study these existence conditions, we plug the expressions for \eqref{app:solutions} inside equation \eqref{app:sign_condition} and find:
\begin{align*}
    \text{sign}(\mu)\text{sign}\left[\left(2\sigma^2-\Delta^2\pm\Delta^2\sqrt{1-\frac{4(\sigma^2-\Delta^2)}{\mu^2}}\right)\frac{2\sigma^2-\Delta^2}{4\sigma^2(\sigma^2-\Delta^2)}-1
    \right]=-\text{sign}(\mu)
\end{align*}
where we are assuming that the eigenvalues are indeed isolated (to be verified a posteriori). This expression is equivalent to 
\begin{align*}
    &\left(2\sigma^2-\Delta^2\pm\Delta^2\sqrt{1-\frac{4(\sigma^2-\Delta^2)}{\mu^2}}\right)\frac{2\sigma^2-\Delta^2}{4\sigma^2(\sigma^2-\Delta^2)}-1\leq 0 \quad \Leftrightarrow \quad\Delta^2\pm(2\sigma^2-\Delta^2)\sqrt{1-\frac{4(\sigma^2-\Delta^2)}{\mu^2}}\leq 0
\end{align*}
from which it is clear that the only acceptable isolated eigenvalue in this setting is $\lambda_{\rm iso, -}$, since otherwise we would have that the sum of two positive quantities is smaller or equal than 0. By studying the second degree equation $\lambda_{\rm iso, -}(\mu, \Delta,\sigma )>2\sigma$ we find exactly that it is verified provided that Eq.~\eqref{eq:ConditionEx} in the main text holds, which is the condition of existence of the isolated eigenvalue.
Under this condition it is straightforward to see that \eqref{app:determinant_condition} is automatically verified.  Hence in this setting there exists only one isolated eigenvalue, whose explicit expression is precisely Eq.~\eqref{eq:IsoExplicit} in the main text.
This eigenvalue appears as soon as Eq.~\eqref{eq:ConditionEx} is satisfied. For $\mu>0$, this eigenvalue is the maximal eigenvalue of the random matrix, while for $\mu<0$ it is the minimal.\\

\underline{ \bf Case  $ \sigma<\Delta$. }

This case is richer, and to the best of our knowledge was not discussed in previous literature. We notice that in this case whenever they exist, then the isolated eigenvalues satisfy $\text{sign}(\lambda_{\rm iso, \pm})=\mp\text{sign}(\mu)$. Moreover the condition \eqref{app:determinant_condition} is always verified in this setting. By plugging \eqref{app:solutions} into \eqref{app:sign_condition} we obtain 
\begin{align*}
    \text{sign}(\mu)\text{sign}\left(\frac{\pm\Delta^2\sqrt{1+\frac{4(\Delta^2-\sigma^2)}{\mu^2}}-\Delta^2+2\sigma^2}{2(\Delta^2-\sigma^2)}\frac{2\sigma^2-\Delta^2}{2\sigma^2}+1\right)=-\text{sign}(\mu)\text{sign}\left(\frac{\pm\Delta^2\sqrt{1+\frac{4(\Delta^2-\sigma^2)}{\mu^2}}-\Delta^2+2\sigma^2}{2(\Delta^2-\sigma^2)}\right)
\end{align*}
which gives us the condition
\begin{align*}
    \pm\left(\frac{\pm\Delta^2\sqrt{1+\frac{4(\Delta^2-\sigma^2)}{\mu^2}}-\Delta^2+2\sigma^2}{2(\Delta^2-\sigma^2)}\frac{2\sigma^2-\Delta^2}{2\sigma^2}+1\right)\leq 0
\end{align*}
In the case in which we choose the sign $-$, this inequality becomes 
\begin{align*}
    \Delta^2\geq (2\sigma^2-\Delta^2)\sqrt{1+\frac{4(\Delta^2-\sigma^2)}{\mu^2}}
\end{align*}
from which we deduce that it is always verified when $\Delta\geq\sqrt{2}\sigma$ and it is verified only for $|\mu|\geq 2\sigma-\Delta^2/\sigma$ when $\Delta<\sqrt{2}\sigma$.The combination of these conditions leads to $|\mu|\geq 2\sigma-\Delta^2/\sigma$ (in the first case $2\sigma-\Delta^2/\sigma$ becomes negative and therefore any $\mu$ will satisfy the condition).
In the case in which we choose the sign $+$, the inequality becomes 
\begin{align*}
    0\geq \Delta^2+ (2\sigma^2-\Delta^2)\sqrt{1+\frac{4(\Delta^2-\sigma^2)}{\mu^2}}
\end{align*}
which is never true for $\Delta\leq\sqrt{2}\sigma$ and becomes true for $\Delta>\sqrt{2}\sigma$ as long as $|\mu|\leq \Delta^2/\sigma-2\sigma$. It remains to verify that both of these isolated eigenvalues are bigger than $2\sigma$ in their domain of existence. By plugging the expression for $\lambda_{\rm iso, \pm}$ as in \eqref{app:solutions}, it is straightforward to verify that both of these isolated eigenvalues are outside of the bulk if we take strict inequalities in the existence conditions that we have just found. \\\\
\noindent Henceforth, we can resume our results as follows: $\lambda_{\rm iso, -}$ is an isolated eigenvalue provided that $|\mu|>2\sigma-\Delta^2/\sigma$; $\lambda_{\rm iso, +}$ is also an isolated eigenvalue provided that $\Delta>\sqrt{2}\sigma$ and $|\mu|<\Delta^2/\sigma-2\sigma$. In particular, notice that whenever $\lambda_{\text{iso},+}$ exists, then also $\lambda_{\text{iso},-}$ exists.

\subsection{The averaged spectral measure}
Combining the results of the previous sub-section with \eqref{eq:FullG} and \eqref{eq:InvStj}, we finally obtain that 
\begin{align}
    &\nu_N(x)=\rho_{\sigma}(x)+\frac{1}{N}\rho^{(1)}_{\sigma}(x)\\
   &+\frac{1}{N}\Theta \tonde{|\mu|-2\sigma+\frac{\Delta^2}{\sigma}}\delta(x-\lambda_{\rm iso,-})\mathfrak{q}_{\sigma, \Delta}(\lambda_{\rm iso,-}, \mu)\left[\frac{\Delta^2 \,\mathfrak{g}_{\sigma}(\lambda_{\rm iso, -})}{\lambda_{\rm iso,-}-\mu}-\frac{1-\lambda_{\rm iso,-}-\frac{\text{sign}(\lambda_{\rm iso,-})}{\sqrt{\lambda_{\rm iso,-}^2-4\sigma^2}}}{2\sigma^2 \Delta^{-2}}\right]\\
   &-
    \frac{1}{N}\Theta(\Delta-\sqrt{2}\sigma)\Theta\tonde{-|\mu|-2\sigma+\frac{\Delta^2}{\sigma}}\delta(x-\lambda_{\rm iso,+})\mathfrak{q}_{\sigma, \Delta}(\lambda_{\rm iso,+}, \mu)\left[\frac{\Delta^2\,\mathfrak{g}_{\sigma}(\lambda_{\rm iso,+})}{\lambda_{\rm iso,+}-\mu}-\frac{1-\lambda_{\rm iso,+}\frac{\text{sign}(\lambda_{\rm iso,+})}{\sqrt{\lambda_+^2-4\sigma^2}}}{2\sigma^2 \Delta^{-2}}\right]\\
    &+\mathcal{O} \tonde{\frac{1}{N^2}},
\end{align}
where $\mathfrak{q}_{\sigma, \Delta}(z, \mu)$ is given explicitly in \eqref{eq:Defq}, and it is obtained as
\begin{equation}
\lim_{z \to \lambda_{\rm iso,\pm}}\frac{1}{\pi} \text{Im} \frac{1}{z-\mu-\Delta^2\mathfrak{g}_{\sigma}(z)}= \delta(z- \lambda_{\rm iso, \pm}) \mathfrak{q}_{\sigma, \Delta}(z, \mu).
\end{equation}
The identities 
\begin{equation}
\frac{\Delta^2\,\mathfrak{g}_{\sigma}(\lambda_{\rm iso, \pm})}{\lambda_{\rm iso,\pm}-\mu}=1= \mp \mathfrak{q}_{\sigma, \Delta}(\lambda_{\rm iso,\pm})\tonde{\frac{\Delta^2 \,\mathfrak{g}_{\sigma}(\lambda_{\rm iso, \pm})}{\lambda_{\rm iso,\pm}-\mu}-\frac{1-\lambda_{\rm iso, \pm}\frac{\text{sign}(\lambda_{\rm iso, \pm})}{\sqrt{\lambda_{\rm iso,\pm}^2-4\sigma^2}}}{2\sigma^2 \Delta^{-2}}}
\end{equation}
allow us to simplify the average spectral measure, and to recover
 Eq.~\eqref{eq:Dos} in the main text.

\vspace{1 cm}

\section{Expectation of products of resolvents: deterministic limit and finite-size corrections}
\label{app:ExpectationProductResolvents}

In this Appendix we discuss the behaviour of the averaged  matrices  ${\bf \Pi}_{k,m}=\mathbb{E}\left[ \mathbf{G}_0(z)^{k+1}\mathbf{G}_1(\xi)^{m+1}\right]$ introduced in \eqref{main:multi_resolv}. Here the average is taken over the three $(N-1) \times (N-1)$ GOE matrices $\mathbf{H}, \mathbf{W}^{(0)}, \mathbf{W}^{(1)}$. We set $M=N-1$. We determine both the $N \to \infty$ limit of ${\bf \Pi}_{k,m}$, as well as the finite-size corrections to order $1/N$. These finite size corrections can be used to compute the corrections to the bulk-bulk overlap $\Phi\left(\lambda^0,\lambda^1\right)$, see \eqref{eq:CorrBulk}. As recalled in the main text, the $1/N$ corrections to this product have two different types of contributions. One type is generated by the fact that 
the matrices appearing in this formulas have size $M=N-1$, but their variance is normalized to $N$. Another type corresponds to the $1/N$ corrections that would be present also in the case of $N$-dimensional matrices. To distinguish between the different terms, we multiply the first ones by a constant $u$, and set $u \to 1$ at the end of the calculation. By taking $u \to 0$, we can recover the corrections to the standard case of $N$-dimensional GOE matrices, with no perturbations.  We begin by reducing the problem to the case $k=0=m$.\\

\subsection{\bf From higher-order products to products of pairs. } 
\label{app:From higher-order products to products of pairs.}
Let us prove here the formula \eqref{eq:multi_resolv_formula0}. We introduce two infinitesimal parameters $\epsilon, \gamma$ and write:
\begin{align*}
 {\bf \Pi}_{k,m}=   \mathbb{E}\left[ \mathbf{G}_0(z)^{k+1}\mathbf{G}_1(\xi)^{m+1}\right] = \lim_{\epsilon,\gamma\to 0}\mathbb{E}\left[ \mathbf{G}_0(z)\cdots \mathbf{G}_0(z+k\gamma)\mathbf{G}_1(\xi)\cdots \mathbf{G}_1(\xi+m\epsilon)\right].
\end{align*}
We aim at re-writing this product as a sum of single resolvent matrices. To do this, we make use of:
\begin{theorem}
\label{lemma_prod}
If ${\bf M}$ is a symmetric real matrix and we denote ${\bf A}_j:=(j\epsilon+{\bf M})^{-1}$ for $j\in\mathbb{Z}, \epsilon\in\mathbb{R}$ (such that $j\epsilon$ is not an eigenvalue of ${\bf M}$), then for any  $k\in\mathbb{N}_{\geq 1}$:
\begin{align*}
   {\bf  A}_0\cdots {\bf A}_k = \frac{1}{\epsilon^kk!}\sum_{j=0}^k(-1)^j\binom{k}{j}{\bf A}_j.
\end{align*}
\end{theorem}
\begin{proof}

We proceed by induction. Indeed notice that for $k=1$ we have ${\bf A}_0{\bf A}_1=(\mathbf{M})^{-1}(\epsilon+\mathbf{M})^{-1}=\frac{1}{\epsilon}((\mathbf{M})^{-1}-(\epsilon+\mathbf{M})^{-1})=\frac{1}{\epsilon}\mathbf{A}_0-\frac{1}{\epsilon}\mathbf{A}_1$. Now suppose that our Lemma is true for a certain $k$, we will prove that it works also for $k+1$. Let us write:
\begin{align*}
    &\mathbf{A}_0\cdots \mathbf{A}_k\mathbf{A}_{k+1}= \frac{1}{\epsilon^kk!}\sum_{j=0}^k(-1)^j\binom{k}{j}\mathbf{A}_j\mathbf{A}
_{k+1}=\frac{1}{\epsilon^{k+1}k!}\sum_{j=0}^k(-1)^j\binom{k}{j}\frac{1}{(k+1-j)}(\mathbf{A}_j-\mathbf{A}_{k+1})\\
    &=\frac{1}{\epsilon^{k+1}(k+1)!}\left[\sum_{j=0}^k(-1)^j\binom{k+1}{j}\mathbf{A}_j-\sum_{j=0}^k(-1)^{j}\binom{k+1}{j}\mathbf{A}_{k+1}\right]=\frac{1}{\epsilon^{k+1}(k+1)!}\sum_{j=0}^{k+1}(-1)^j\binom{k+1}{j}\mathbf{A}_j
\end{align*}
where in the last equality we used that the identity $0=(1-1)^{k+1}=\sum_{j=0}^{k+1}(-1)^j\binom{k+1}{j}$ implies that $\sum_{j=0}^k(-1)^j\binom{k+1}{j}=-(-1)^{k+1}\binom{k+1}{k+1}=-(-1)^{k+1}$. Hence the induction hypothesis is proved.
\end{proof}

Applying this Lemma, we see that the expectation ${\bf \Pi}_{k,m}=\mathbb{E}\left[ \mathbf{G}_0(z)^{k+1}\mathbf{G}_1(\xi)^{m+1}\right]$ can be written as a linear combination of terms of the form $\mathbb{E}\left[ \mathbf{G}_0(z + i \gamma)\mathbf{G}_1(\xi+ j \epsilon)\right]$ for integer $i,j$. For instance, for $k=0$ it holds:
\begin{align*}
    \lim_{\epsilon\to 0}\mathbb{E}\left[\mathbf{G}_0(z)\mathbf{G}_1(\xi)\cdots \mathbf{G}_1(\xi+m\epsilon)\right]&=\lim_{\epsilon\to 0}\frac{(-1)^m}{\epsilon^mm!}\sum_{j=0}^m(-1)^{m-j}\binom{m}{j}\mathbb{E}_\mathbf{H}\mathbb{E}_
    {\mathbf{W}^{(0)},\mathbf{W}^{(1)}}\left[\mathbf{G}_0(z)\mathbf{G}_1(\xi+j\epsilon)\right]\\
    &=\frac{(-1)^m}{m!}\frac{\partial^m}{\partial\xi^m}\mathbb{E}[\mathbf{G}_0(z)\mathbf{G}_1(\xi)]
\end{align*}
From this expressions, for $\epsilon \to 0$ one recovers an $m-$th order derivative. The same holds for $k>0$. Therefore, we finally get Eq.~\eqref{eq:multi_resolv_formula0}.
Hence we see that to determine ${\bf \Pi}_{k,m}$ one can focus on the behavior of the product ${\bf \Pi}_{0,0}=\mathbb{E}\left[ \mathbf{G}_0(z) \, \mathbf{G}_1(\xi)\right]$. A precise analysis of ${\bf \Pi}_{0,0}$ and its finite size corrections is carried out in the subsections below.

\subsection{\bf Partial expectation over ${\bf W}^{(a)}$. }

Let us study 
$\mathbb{E}\left[\mathbf{G}_0(z)\mathbf{G}_1(\xi)\right]$. Consider first its expectation over ${\bf W}^{(a)}$ for both $a=0$ and $a=1$, which are independent random matrices. For GOE matrices of size $N$, this expectation has been determined explicitly in  \cite{VERBAARSCHOT1984367} to order $1/N$. Adapting those results to the present case, we find that the expectation value of a single resolvent over ${\bf W}^{(a)}$ admits the expansion: 
\begin{equation}\label{eq:BE}
\begin{split}
   & \mathbb{E}_{{\bf W}^{(a)}}[ {\bf G}_a(z)]={\bf G}_a^{(0)}(z)+\frac{1}{N}{\bf G}_a^{(1)}(z)+ \mathcal{O}\tonde{\frac{1}{N^2}}
    \end{split}
\end{equation}
where ${\bf G}_a^{(0)}(z)$ satisfies the self-consistent equation
\begin{equation}\label{eq:SelfG0app}
\begin{split}
   {\bf G}_a^{(0)}(z)=\left(z-\frac{\sigma_W^2}{M} \text{Tr} {\bf G}_a^{(0)}(z)-{\bf H}\right)^{-1}
    \end{split}
\end{equation}
while 
\begin{equation}\label{eq:G1}
\begin{split}
    {\bf G}_a^{(1)}(z)&=
   \frac{\sigma_W^2 \mathbf{R}^3_{\mathbf{H}}
    \tonde{z- \frac{\sigma_W^2}{M} \text{Tr} {\bf G}_a^{(0)}} }{1-\frac{\sigma_W^2}{M}\Tr\mathbf{R}^2_{\mathbf{H}}\tonde{z-\frac{\sigma_W^2}{M} \text{Tr} {\bf G}_a^{(0)}} }+\frac{\frac{\sigma_W^4}{M}\Tr\mathbf{R}^3_{\mathbf{H}}(z-\frac{\sigma_W^2}{M} \text{Tr} {\bf G}_a^{(0)})}{[1-\frac{\sigma_W^2}{M}\Tr\mathbf{R}^2_{\mathbf{H}}\tonde{z-\frac{\sigma_W^2}{M} \text{Tr} {\bf G}_a^{(0)}} ]^2} \;\mathbf{R}^2_{\mathbf{H}}\tonde{z-\frac{\sigma_W^2}{M} \text{Tr} {\bf G}_a^{(0)}}\\
    &-u \frac{\frac{\sigma_W^2}{M}\Tr\mathbf{R}_{\mathbf{H}}(z-\frac{\sigma_W^2}{M} \text{Tr} {\bf G}_a^{(0)})}{1-\frac{\sigma_W^2}{M}\Tr\mathbf{R}^2_{\mathbf{H}}\tonde{z-\frac{\sigma_W^2}{M} \text{Tr} {\bf G}_a^{(0)}} } \;\mathbf{R}^2_{\mathbf{H}}\tonde{z-\frac{\sigma_W^2}{M} \text{Tr} {\bf G}_a^{(0)}},
    \end{split}
\end{equation}
where we have defined $\mathbf{R}_{\mathbf{H}}(z):=(z- {\bf H})^{-1}$. Here ${\bf G}_a^{(0)}$ collects terms that are of $\mathcal{O}(N^0)$ for fixed ${\bf H}$. For $u \to 0$, the result of \cite{VERBAARSCHOT1984367} is recovered.  The normalized trace $M^{-1} \text{Tr}{\bf G}_a^{(0)}$ (recall $M:=N-1$) is a random variable due to the randomness in ${\bf H}$. It  converges to a deterministic limit when $N \to \infty$, with fluctuations of order $1/N$ that we denote with $\eta_a $:
\begin{equation}\label{eq:expansi0}
\frac{1}{M} \text{Tr}{\bf G}_a^{(0)}= \mathfrak{g}_\sigma(z) + \frac{1}{N}  \eta_a(z; {\bf H}) + \mathcal{O}\tonde{\frac{1}{N^2}}.
\end{equation}
Plugging this into \eqref{eq:SelfG0app} we find:
\begin{equation}\label{eq:SelfG0app2}
\begin{split}
   {\bf G}_a^{(0)}(z)=\frac{1}{z-{\bf H}- \sigma_W^2 \mathfrak{g}_\sigma(z)}+ \frac{\sigma^2_W}{N} \frac{ \eta_a(z; {\bf H}) }{[z- {\bf H} - \sigma^2_W \mathfrak{g}_\sigma(z)]^2}+ \mathcal{O}\tonde{\frac{1}{N^2}},
    \end{split}
\end{equation}
which implies:
\begin{align}
\frac{1}{M}\mathbb{E}_\mathbf{H}\Tr\mathbf{G}_a^{(0)}&=\mathbb{E}_\mathbf{H} \quadre{\frac{1}{M} \text{Tr}\frac{1}{z-{\bf H}- \sigma_W^2 \mathfrak{g}_\sigma(z)}} +\frac{\sigma_W^2}{N}\mathbb{E}_\mathbf{H} [\eta_a(z; {\bf H})]\mathbb{E}_\mathbf{H}\quadre{\frac{1}{M}\text{Tr} \tonde{\frac{1}{z-\mathbf{H}-\sigma_W^2\mathfrak{g}_\sigma(z)}}^{2}} + \mathcal{O}\tonde{\frac{1}{N^2}}.
\end{align}
Here we are assuming that the expectation value on the right-hand side factorizes to leading order in $N$.  On the other hand,  \eqref{eq:expansi0} also implies
\begin{equation}
\frac{1}{M} \mathbb{E}_\mathbf{H} \text{Tr}{\bf G}_a^{(0)}= \mathfrak{g}_\sigma(z) + \frac{1}{N} \mathbb{E}_\mathbf{H} [ \eta_a(z; {\bf H})] + \mathcal{O}\tonde{\frac{1}{N^2}}.
\end{equation}
Equating these expressions allows us to to solve for $ \mathbb{E}_\mathbf{H} [ \eta_a(z; {\bf H})]$. We use the fact that:
\begin{equation}\label{eq:FirstDer}
\mathbb{E}_{\bf H} \quadre{\frac{1}{M} \text{Tr} \tonde{\frac{  1 }{z- {\bf H} - \sigma^2_W \mathfrak{g}_\sigma(z)}}^2}= - \partial_\zeta \mathbb{E}_{\bf H} \quadre{\frac{1}{M} \text{Tr} \tonde{\frac{  1 }{\zeta- {\bf H}}}}_{\zeta=z-\sigma^2_W \mathfrak{g}_\sigma(z)} =- \frac{\mathfrak{g}'_\sigma(z)}{1-\sigma^2_W \mathfrak{g}'_\sigma(z)}+ \mathcal{O}\tonde{\frac{1}{N}},
\end{equation}
where we exploited the identity \cite{potters_bouchaud_2020}:
\begin{equation}\label{eq:Bouch}
    \mathfrak{g}_{\sigma_H}(z-\sigma^2_W \mathfrak{g}_\sigma(z))=\mathfrak{g}_{\sigma}(z).
\end{equation}
 Moreover, the analogue of \eqref{eq:BE} gives:
\begin{align}\label{eq:Mio1}
\mathbb{E}_\mathbf{H} \quadre{\frac{1}{M} \text{Tr}\frac{1}{z-{\bf H}- \sigma_W^2 \mathfrak{g}_\sigma(z)}}=\quadre{\mathfrak{g}_{\sigma_H}(\zeta)+ \frac{1}{N} \tonde{\frac{\sigma_H^2 \mathfrak{g}_{\sigma_H}^3(\zeta) }{[1- \sigma_H^2 \mathfrak{g}_{\sigma_H}^2(\zeta)]^2}- u \frac{\sigma_H^2 \mathfrak{g}_{\sigma_H}^3(\zeta) }{1- \sigma_H^2 \mathfrak{g}_{\sigma_H}^2(\zeta)}}+ \mathcal{O}\tonde{\frac{1}{N^2}}}_{\zeta= z- \sigma^2_W \mathfrak{g}_\sigma(z)}.
\end{align}
It can be easily checked that the following identity holds:
\begin{equation}
\frac{\sigma_H^2 \mathfrak{g}_{\sigma_H}^3(z-\sigma_W^2\mathfrak{g}_\sigma(z)) }{[1- \sigma_H^2 \mathfrak{g}_{\sigma_H}^2(z-\sigma_W^2\mathfrak{g}_\sigma(z))]^2}=\frac{\sigma_H^2 \mathfrak{g}_\sigma^3(z) }{[1- \sigma_H^2 \mathfrak{g}_\sigma^2(z)]^2},
\end{equation}
leading to:
\begin{equation}
 \mathbb{E}_\mathbf{H} [ \eta_a(z; {\bf H})]=
 \tonde{\frac{\sigma_H^2 \mathfrak{g}^3_\sigma(z) }{[1- \sigma_H^2 \mathfrak{g}^2_\sigma(z)]^2}- u \frac{\sigma_H^2 \mathfrak{g}^3_\sigma(z) }{1- \sigma_H^2 \mathfrak{g}^2_\sigma(z)}}
  \quadre{1-\sigma^2_W \mathfrak{g}'_\sigma(z)}.
\end{equation}
Assuming that the $1/N$ contribution to the trace, $\eta_a$, is self-averaging in the large-$N$ limit, we can replace it with its average in \eqref{eq:expansi0} and get the final expansion:
\begin{align}
\label{app:one_resolvent2}
\mathbb{E}_{\mathbf{W}^{(a)}}\left[\mathbf{G}_a(z)\right]=\mathbf{R}_{\mathbf{H}}(z-\sigma_{W}^2\mathfrak{g}_{\sigma}(z))+\frac{1}{N}\overline{c}(z)\mathbf{R}^2_{\mathbf{H}}(z-\sigma_{W}^2\mathfrak{g}_{\sigma}(z))+\frac{1}{N}\mathbf{G}_a^{(1)}(z)+\mathcal{O}\left(\frac{1}{N^2}\right)
\end{align}
with $\mathbf{G}_a^{(1)}(z)$ given in \eqref{eq:G1} and with:
\begin{equation}\label{eq:CBarra}
\overline{c}(z) = \sigma_W^2 \tonde{\frac{\sigma_H^2 \mathfrak{g}^3_\sigma(z) }{[1- \sigma_H^2 \mathfrak{g}^2_\sigma(z)]^2}- u \frac{\sigma_H^2 \mathfrak{g}^3_\sigma(z) }{1- \sigma_H^2 \mathfrak{g}^2_\sigma(z)}}
 \quadre{1-\sigma^2_W \mathfrak{g}'_\sigma(z)}.
\end{equation}

\vspace{1 cm}
\underline{\bf Consistency checks. } We perform some consistency checks on the expansion \eqref{app:one_resolvent2}. First, for ${\bf H}=0, \sigma_H=0$ one recovers \eqref{eq:SingleRes} with $\sigma \to \sigma_W$, as it follows from the identity:
\begin{equation}
\frac{z- \sqrt{z^2-4 \sigma^2}}{2[z^2- 4 \sigma^2]}=\frac{\sigma^2 \mathfrak{g}_\sigma^3(z) }{[1- \sigma^2 \mathfrak{g}_\sigma^2(z)]^2}.
\end{equation} 
Moreover, \eqref{eq:SingleRes} is recovered when taking the trace of  \eqref{app:one_resolvent2} and averaging over ${\bf H}$. From \eqref{eq:Mio1} together with the identity \eqref{eq:Bouch} it follows that:
\begin{equation}\label{eq:Rh_corrections}
    \frac{1}{N}\mathbb{E}_{\mathbf{H}}[\Tr\mathbf{R}_\mathbf{H}(z-\sigma_W^2\mathfrak{g}_\sigma(z))]=\mathfrak{g}_\sigma(z)+\frac{1}{N}\tonde{\frac{\sigma_H^2\mathfrak{g}^3_\sigma(z)}{[1-\sigma_H^2\mathfrak{g}^2_\sigma(z)]^2}- \frac{u \sigma_H^2\mathfrak{g}^3_\sigma(z)}{1-\sigma_H^2\mathfrak{g}^2_\sigma(z)}- u \mathfrak{g}_\sigma(z) }+ \mathcal{O}\tonde{\frac{1}{N^2}}.
\end{equation}
Making use of  \eqref{eq:FirstDer}, we see that the expectation of the second term in \eqref{app:one_resolvent2} reads:
\begin{align}
\frac{1}{N} \overline{c}(z)\mathbb{E}_{\mathbf{H}}\left[ \frac{1}{N} \text{Tr} \mathbf{R}^2_{\mathbf{H}}(z-\sigma_{W}^2\mathfrak{g}_{\sigma}(z))\right]=-\frac{1}{N}  \frac{\overline{c}(z) \, \mathfrak{g}'_\sigma(z)}{1-\sigma^2_W \mathfrak{g}'_\sigma(z)}+ \mathcal{O}\tonde{\frac{1}{N^2}}.
\end{align}
Finally, using that for $\zeta=z-\sigma^2_W \mathfrak{g}_\sigma(z)$ it holds:
\begin{equation}\label{eq:SecondDer}
\mathbb{E}_{\bf H} \quadre{\frac{1}{M} \text{Tr} \tonde{\frac{  1 }{z- {\bf H} - \sigma^2_W \mathfrak{g}_\sigma(z)}}^3}= \frac{1}{2} \partial^2_\zeta \mathbb{E}_{\bf H} \quadre{\frac{1}{M} \text{Tr} \tonde{\frac{  1 }{\zeta- {\bf H}}}} =\frac{1}{2}   \partial^2_\zeta \mathfrak{g}_{\sigma_H}(\zeta)= \frac{1}{2} \frac{ \mathfrak{g}''_{\sigma}(z)}{[1-\sigma^2_W   \mathfrak{g}'_{\sigma}(z)]^3}, 
\end{equation}
we find that the expectation value of the trace of \eqref{eq:G1} can be written as: 
\begin{equation}\label{eq:Taxi3}
\mathbb{E}_{\bf H} \quadre{\frac{1}{N} \text{Tr} {\bf G}_a^{(1)}(z)}=\frac{1}{2} \frac{ \sigma_W^2 \mathfrak{g}''_{\sigma}(z)}{1-\sigma^2_W   \mathfrak{g}'_{\sigma}(z)} + 
  u \sigma_W^2  \mathfrak{g}_\sigma(z) \mathfrak{g}'_{\sigma}(z)+\mathcal{O}\tonde{\frac{1}{N^2}}.
\end{equation}
Combining everything, one gets:
\begin{equation}\label{eq:CompAll}
\begin{split}
&\frac{1}{N}\mathbb{E}\left[ \text{Tr}\mathbf{G}_a\right]= \tonde{1-\frac{u}{N}}\mathfrak{g}_\sigma(z)+\\
&\frac{1}{N}\quadre{\frac{\sigma_H^2\mathfrak{g}^3_\sigma(z) (1-\sigma_W^2 \mathfrak{g}'_\sigma(z))}{[1-\sigma_H^2\mathfrak{g}^2_\sigma(z)]^2}
+\frac{1}{2} \frac{ \sigma_W^2 \mathfrak{g}''_{\sigma}(z)}{1-\sigma^2_W   \mathfrak{g}'_{\sigma}(z)} -\frac{u\sigma_H^2\mathfrak{g}^3_\sigma(z) (1-\sigma_W^2 \mathfrak{g}'_\sigma(z))}{1-\sigma_H^2\mathfrak{g}^2_\sigma(z)}+ u \mathfrak{g}_\sigma(z) \sigma^2_W  \mathfrak{g}'_\sigma(z)}+ \mathcal{O}\tonde{\frac{1}{N^2}}.
\end{split}
\end{equation}
It can be checked explicitly that:
\begin{equation}
\frac{\sigma_H^2 \mathfrak{g}_\sigma^3(z) }{[1- \sigma_H^2 \mathfrak{g}_\sigma^2(z)]^2} \tonde{1-\sigma^2_W  \mathfrak{g}'_\sigma(z)}  +\frac{1}{2}\frac{ \sigma_W^2 \mathfrak{g}''_{\sigma}(z)}{1- \sigma^2_W \mathfrak{g}'_{\sigma}(z) } = \frac{\sigma^2 \mathfrak{g}_\sigma^3(z) }{[1- \sigma^2 \mathfrak{g}_\sigma^2(z)]^2}
, \quad \quad \sigma^2= \sigma^2_H+ \sigma^2_W,
\end{equation}
as well as:
\begin{equation}
\frac{1-\sigma^2_W \mathfrak{g}'_\sigma(z)}{1- \sigma_H^2 \mathfrak{g}_\sigma^2(z)}  = \frac{1}{1- \sigma^2 \mathfrak{g}_\sigma^2(z)}
, \quad \quad \sigma^2= \sigma^2_H+ \sigma^2_W,
\end{equation}
which imply that \eqref{eq:CompAll} is also equal to:
\begin{equation}\label{eq:CompAll2}
\begin{split}
&\frac{1}{N}\mathbb{E}\left[ \text{Tr}\mathbf{G}_a\right]= \mathfrak{g}_\sigma(z)+\frac{1}{N}\quadre{\frac{\sigma^2 \mathfrak{g}_\sigma^3(z) }{[1- \sigma^2 \mathfrak{g}_\sigma^2(z)]^2}- \frac{u \sigma^2\mathfrak{g}^3_\sigma(z)}{1-\sigma^2\mathfrak{g}^2_\sigma(z)}+u \sigma_W^2\mathfrak{g}_\sigma(z)
\tonde{\frac{\mathfrak{g}^2_\sigma(z)}{1-\sigma^2\mathfrak{g}^2_\sigma(z)} 
+ \mathfrak{g}'_\sigma(z)}
-u \mathfrak{g}_\sigma(z)}+ \mathcal{O}\tonde{\frac{1}{N^2}}
\end{split}
\end{equation}
which coincides with \eqref{eq:SingleRes}, given that the sum in the round brackets vanishes.

\subsection{\bf Expectation over  ${\bf H}$: the leading order term. }\label{app:multi_resolv}

\noindent
We derive the leading order  contribution to $ {\bf \Pi}_{0,0}$.
From \eqref{app:one_resolvent2} it appears that the leading order contribution is given by the term $\mathbb{E}_{\bf H}\left[ \mathbf{R}_{\mathbf{H}}(z-\sigma_{W}^2\mathfrak{g}_{\sigma}(z))\, \mathbf{R}_{\mathbf{H}}(\xi-\sigma_{W}^2\mathfrak{g}_{\sigma}(\xi))\right]$. The resolvent identity implies:

\begin{align*}
\mathbb{E}_{\bf H}\left[ \mathbf{R}_{\mathbf{H}}(z-\sigma_{W}^2\mathfrak{g}_{\sigma}(z))\, \mathbf{R}_{\mathbf{H}}(\xi-\sigma_{W}^2\mathfrak{g}_{\sigma}(\xi))\right]=  \frac{\mathbb{E}_\mathbf{H}[(z-\sigma_{W}^2 \mathfrak{g}_\sigma(z)-\mathbf{H})^{-1}]-  \mathbb{E}_\mathbf{H}[(\xi-
   \sigma_{W}^2 \mathfrak{g}_\sigma(\xi)-\mathbf{H})^{-1} ]}{\xi -z-\sigma^2_W\tonde{\mathfrak{g}_\sigma(\xi)- \mathfrak{g}_\sigma(z)}}.
\end{align*}

To leading order in $N$, thanks to \eqref{eq:Bouch}, it holds:
\begin{align*}
    \lim_{N \to \infty} \mathbb{E}_\mathbf{H}[(z-\sigma_{W}^2 \mathfrak{g}_\sigma(z)-\mathbf{H})^{-1}]=\mathfrak{g}_{\sigma_H}(z-\sigma_{W}^2 \mathfrak{g}_\sigma(z))=\mathfrak{g}_{\sigma}(z),
\end{align*}
This implies
\begin{align*}
 \lim_{N \to \infty} {\bf \Pi}_{0,0} = \lim_{N \to \infty} \mathbb{E}_{\bf H}\left[ \mathbf{R}_{\mathbf{H}}(z-\sigma_{W}^2\mathfrak{g}_{\sigma}(z))\, \mathbf{R}_{\mathbf{H}}(\xi-\sigma_{W}^2\mathfrak{g}_{\sigma}(\xi))\right]=\frac{\mathfrak{g}_\sigma(z)-\mathfrak{g}_\sigma(\xi)}{\xi -z-\sigma_W^2(\mathfrak{g}_\sigma(\xi)-\mathfrak{g}_\sigma(z))}\mathbbm{1}=\Psi(z,\xi)\mathbbm{1}
\end{align*}
with $\Psi(z,\xi)$  defined in \eqref{BigPsi}. In the rest of Appendix we derive the $1/N$ corrections to this term, and thus Eq. \eqref{eq:multi_resolv_formula}.

\subsection{\bf Expectation over  ${\bf H}$: the $1/N$ corrections.}\label{app:finite_size}

We determine the finite size corrections to ${\bf \Pi}_{0,0}$.
We recall that the terms proportional to $u$ correspond to corrections that are due to the fact that our matrices have size $M=N-1$ and not $N$, while having variances rescaled with $N$. 
 In getting \eqref{app:one_resolvent2}, we have used the fact that traced quantities can be approximated with their leading order, deterministic contribution. Reasoning in an analogous way and assuming everywhere $\zeta=z-\sigma^2_W \mathfrak{g}_{\sigma}(z)$, we define:
\begin{equation}
\begin{split}\label{eq:ab}
    &a(z):=\lim_{N \to \infty} {\frac{\sigma_W^2}{1-\frac{\sigma_W^2}{M}\Tr\mathbf{R}^2_{\mathbf{H}}(z-\sigma_{W}^2\mathfrak{g}_{\sigma}(z))}}=\frac{\sigma_W^2}{1+\sigma_W^2 \partial_\zeta \mathfrak{g}_{\sigma_H}(\zeta)}= \sigma^2_W (1- \sigma^2_W \mathfrak{g}'_{\sigma}(z))\\
    &b(z):= \lim_{N \to \infty} {\frac{\frac{\sigma_W^4}{N}\Tr\mathbf{R}^3_{\mathbf{H}}(z-\sigma_{W}^2\mathfrak{g}_{\sigma}(z))}{\left(1-\frac{\sigma_W^2}{N}\Tr\mathbf{R}^2_{\mathbf{H}}(z-\sigma_{W}^2\mathfrak{g}_{\sigma}(z))\right)^2}}=
    \frac{\sigma_W^4}{2}\frac{\partial^2_\zeta \mathfrak{g}_{\sigma_H}(\zeta)}{[1+\sigma_W^2 \partial_\zeta \mathfrak{g}_{\sigma_H}(\zeta)]^2}=\frac{\sigma_W^4}{2}\frac{\mathfrak{g}''_{\sigma}(z)}{1- \sigma_W^2 \mathfrak{g}'_{\sigma}(z)}\\
      &d(z):=\lim_{N \to \infty} {\frac{\frac{\sigma_W^2}{N}\Tr\mathbf{R}_{\mathbf{H}}(z-\sigma_{W}^2\mathfrak{g}_{\sigma}(z))}{1-\frac{\sigma_W^2}{M}\Tr\mathbf{R}^2_{\mathbf{H}}(z-\sigma_{W}^2\mathfrak{g}_{\sigma}(z))}}=\frac{\sigma_W^2 \mathfrak{g}_{\sigma}(z)}{1+\sigma_W^2 \partial_\zeta \mathfrak{g}_{\sigma_H}(\zeta)}= \sigma^2_W \mathfrak{g}_{\sigma}(z) (1- \sigma^2_W \mathfrak{g}'_{\sigma}(z))\\
    \end{split}
\end{equation}
and set:
\begin{equation}\label{eq:bbarra}
  \overline{b}(z):= b(z)+ \overline{c}(z) - u d(z)=\frac{\sigma_W^4}{2}\frac{\mathfrak{g}''_{\sigma}(z)}{1- \sigma_W^2 \mathfrak{g}'_{\sigma}(z)}+\sigma_W^2 \tonde{\frac{\sigma_H^2 \mathfrak{g}^3_\sigma(z) }{[1- \sigma_H^2 \mathfrak{g}^2_\sigma(z)]^2}- u \frac{\sigma_H^2 \mathfrak{g}^3_\sigma(z) }{1- \sigma_H^2 \mathfrak{g}^2_\sigma(z)} - u  \mathfrak{g}_{\sigma}(z)}
 \quadre{1-\sigma^2_W \mathfrak{g}'_\sigma(z)}
\end{equation}
so that 
\begin{align}
\label{resolvent_corrections}
\mathbb{E}_{\mathbf{W}^{(a)}}\left[\mathbf{G}_a(z)\right]=\mathbf{R}_{\mathbf{H}}(z-\sigma_{W}^2\mathfrak{g}_{\sigma}(z))+\frac{\overline{b}(z)}{N}\mathbf{R}^2_{\mathbf{H}}(z-\sigma_{W}^2\mathfrak{g}_{\sigma}(z))+\frac{a(z)}{N}\mathbf{R}^3_{\mathbf{H}}(z-\sigma_{W}^2\mathfrak{g}_{\sigma}(z)).
\end{align}

We can then express
\begin{equation}
\begin{split}
   & \mathbb{E}[\mathbf{G}_0(z)\mathbf{G}_1(\xi)]=\mathbb{E}_{\mathbf{H}}\Bigg[\mathbf{R}_{\bf H}(z-\sigma_W^2\mathfrak{g}_\sigma(z))\mathbf{R}_{\bf H}(\xi-\sigma_W^2\mathfrak{g}_\sigma(\xi))\Bigg]+\\
    &\frac{\overline b(\xi) }{N}\mathbb{E}_{\mathbf{H}}\Bigg[    \mathbf{R}_{\bf H}(z-\sigma_W^2\mathfrak{g}_\sigma(z))\mathbf{R}^2_{\bf H}(\xi-\sigma_W^2\mathfrak{g}_\sigma(\xi)) \Bigg]+\frac{\overline b(z)}{N}\mathbb{E}_{\mathbf{H}}\Bigg[      \mathbf{R}^2_{\bf H}(z-\sigma_W^2\mathfrak{g}_\sigma(z))\mathbf{R}_{\bf H}(\xi-\sigma_W^2\mathfrak{g}_\sigma(\xi))\Bigg]+\\
    &+\frac{a(\xi)}{N} \mathbb{E}_{\mathbf{H}}\Bigg[ \mathbf{R}_{\bf H}(z-\sigma_W^2\mathfrak{g}_\sigma(z))\mathbf{R}^3_{\bf H}(\xi-\sigma_W^2\mathfrak{g}_\sigma(\xi))\Bigg]
    +\frac{a(z)}{N} \mathbb{E}_{\mathbf{H}}\Bigg[ \mathbf{R}^3_{\bf H}(z-\sigma_W^2\mathfrak{g}_\sigma(z))\mathbf{R}_{\bf H}(\xi-\sigma_W^2\mathfrak{g}_\sigma(\xi))
    \Bigg] + \mathcal{O}\tonde{\frac{1}{N^2}}.
    \end{split}
    \end{equation}

The expectations in the second and third line in this formula need to be computed to lowest order in $N$. Proceeding as above, we find: 
\begin{align*} \mathbb{E}_{\mathbf{H}}\left[\mathbf{R}_{\mathbf{H}}(z-\sigma_W^2\mathfrak{g}_\sigma(z))\mathbf{R}^3_{\mathbf{H}}(\xi-\sigma_W^2\mathfrak{g}_\sigma(\xi))\right]&=\frac{1}{2(1-\sigma_W^2\mathfrak{g}_\sigma'(\xi))^2}\left[\partial_{\xi}^2\Psi(z,\xi)+\frac{\sigma_W^2\mathfrak{g}_\sigma''(\xi)}{1-\sigma_W^2\mathfrak{g}_\sigma'(\xi)}\partial_{\xi}\Psi(z,\xi)
    \right] \mathbbm{1} + \mathcal{O}\tonde{\frac{1}{N}}
\end{align*}
as well as
\begin{align*}
\mathbb{E}_{\mathbf{H}}\left[\mathbf{R}_{\mathbf{H}}(z-\sigma_W^2\mathfrak{g}_\sigma(z))\mathbf{R}^2_{\mathbf{H}}(\xi-\sigma_W^2\mathfrak{g}_\sigma(\xi))\right]=-\frac{1}{1-\sigma_W^2\mathfrak{g}'_\sigma(\xi)}\partial_\xi\Psi(z,\xi)\mathbbm{1} + \mathcal{O}\tonde{\frac{1}{N}}
\end{align*}
and similarly for the terms with $\xi \to z$. It now remains to determine the $1/N$ expansion of the first expectation value in the above formula. Using the resolvent identity we obtain:
\begin{equation}
\begin{split}
   & \mathbb{E}_{\mathbf{H}}\Bigg[\mathbf{R}_{\bf H}(z-\sigma_W^2\mathfrak{g}_\sigma(z))\mathbf{R}_{\bf H}(\xi-\sigma_W^2\mathfrak{g}_\sigma(\xi))\Bigg]
   =\mathbb{E}_{\mathbf{H}} \Bigg[\frac{\mathbf{R}_{\bf H}(z-\sigma_W^2\mathfrak{g}_\sigma(z))-\mathbf{R}_{\bf H}(\xi-\sigma_W^2\mathfrak{g}_\sigma(\xi))}{\xi-z -\sigma^2_W \mathfrak{g}_\sigma(\xi)+ \sigma_W^2\mathfrak{g}_\sigma(z)}\Bigg].
       \end{split}
    \end{equation}
Making use of \eqref{eq:Rh_corrections} 
we finally get: 
\begin{equation}
\begin{split}
   & \frac{1}{N}\mathbb{E}_{\mathbf{H}}\Bigg[\text{Tr} \mathbf{R}_{\bf H}(z-\sigma_W^2\mathfrak{g}_\sigma(z))\mathbf{R}_{\bf H}(\xi-\sigma_W^2\mathfrak{g}_\sigma(\xi))\Bigg]=
   \Psi(z,\xi)+ \frac{1}{N}\Lambda(z, \xi) + \mathcal{O}\tonde{\frac{1}{N^2}}
       \end{split}
    \end{equation}
with the function $\Lambda$ reads
\begin{align}\label{eq:Lambd}
\begin{split}
&\Lambda(z, \xi)=\frac{1}{{\xi-z -\sigma^2_W \mathfrak{g}_\sigma(\xi)+ \sigma_W^2\mathfrak{g}_\sigma(z)}} \tonde{ 
\frac{-u\mathfrak{g}_\sigma(z)+\sigma_H^2\mathfrak{g}^3_\sigma(z)(u+1)}{[1-\sigma_H^2\mathfrak{g}^2_\sigma(z)]^2}-
\frac{-u\mathfrak{g}_\sigma(\xi)+\sigma_H^2\mathfrak{g}^3_\sigma(\xi)(u+1)}{[1-\sigma_H^2\mathfrak{g}^2_\sigma(\xi)]^2}
}\\
&=\frac{1}{{\xi-z -\sigma^2_W \mathfrak{g}_\sigma(\xi)+ \sigma_W^2\mathfrak{g}_\sigma(z)}} \tonde{\frac{\sigma_H^2\mathfrak{g}^3_\sigma(z)}{[1-\sigma_H^2\mathfrak{g}^2_\sigma(z)]^2}- \frac{u \sigma_H^2\mathfrak{g}^3_\sigma(z)}{1-\sigma_H^2\mathfrak{g}^2_\sigma(z)}
-\frac{\sigma_H^2\mathfrak{g}^3_\sigma(\xi)}{[1-\sigma_H^2\mathfrak{g}^2_\sigma(\xi)]^2}+ 
\frac{u \sigma_H^2\mathfrak{g}^3_\sigma(\xi)}{1-\sigma_H^2\mathfrak{g}^2_\sigma(\xi)}}- u \Psi(z, \xi).
\end{split}
\end{align}
Combining everything, we finally obtain 
\begin{align}
\begin{split}
    \frac{1}{N}\mathbb{E}\Tr[\mathbf{G}_0(z)\mathbf{G}_1(\xi)]=\Psi(z,\xi)+\frac{1}{N}\Psi^{(1)}(z,\xi)+\mathcal{O}\tonde{\frac{1}{N^2}}
\end{split}
\end{align}
with

\begin{align}
\Psi^{(1)}(z, \xi)=   \Lambda(z,\xi) + \alpha (z) \partial_z \Psi(z,\xi)+ \alpha(\xi) \partial_\xi \Psi(z,\xi)+ \beta(z) \partial^2_z \Psi(z,\xi) + \beta(\xi) \partial^2_\xi \Psi(z,\xi)
\end{align}
and 
$\alpha,\beta$ defined in equation \eqref{eq:CoefficientiLabda}. Eq.~\eqref{eq:corrections_prod} in the main text is derived analogously, without taking the trace. The only difference is in the $1/N$ term, and stems from the fact that one is taking the trace of $(N-1) \times (N-1)$ identity matrices, normalizing them by a factor $N$. For this reason, $\overline{\Psi}^{(1)}(z, \xi)$ in \eqref{eq:corrections_prod} differs from $\Psi^{(1)}(z, \xi)$ by a simple factor,
   $\overline{\Psi}^{(1)}(z, \xi)={\Psi}^{(1)}(z, \xi)+u \Psi(z, \xi),$ 
and similarly $\overline{\Lambda}(z,\xi)= \Lambda(z, \xi)+ u \Psi(z, \xi)$, as it follows comparing \eqref{eq:Lambd} with \eqref{eq:Lambda} in the main text. 
We recall that the contributions coming from the fact that the matrices considered here are of dimension $M=N-1$ are multiplied by the constant $u$, which must then be set equal to 1. Instead, if we set $u \to 0$ we get the finite size corrections to the case of unperturbed matrices with size $N$.

\vspace{1 cm}

\section{Computation of the auxiliary function $\psi(z, \xi)$}\label{app:computation_psi_00}
\subsection{Computation of $\psi_{00},\psi_{0N},\psi_{NN}$}
\label{app:comp_psi_pieces}
In this Appendix we report how to obtain the expressions \eqref{eq:psi00Implicit},  \eqref{eq:psi0NImplicit} and \eqref{eq:psiNNImplicit}. We begin by $\psi_{00}$, defined in \eqref{psi_decomposition}. By using \eqref{eq:Block1} and the Dyson expansion in \eqref{eq:dyson} we can easily rewrite
\begin{equation}\label{eq:Sum}
    \psi_{00}(z,\xi)= \sum_{k,m=0}^{+\infty}\frac{1}{N}\Tr_{N-1}\mathbb{E}\left[{\bf S}_{k,m} \right], \quad \quad {\bf S}_{k,m}={\bf G}_0(z)[{\bf A}^{(0)}(z){\bf G}_0(z)]^k{\bf G}_1(\xi)[{\bf A}^{(1)}(\xi){\bf G}_1(\xi)]^m \quad 
\end{equation}
where we used the subscript to stress that the trace is over a subspace of dimension $N-1$. The definition for $\mathbf{A}^{(a)}$ is expressed in \eqref{eq:DefA}. We compute the partial averages of the strings ${\bf S}_{k,m}$ over the entries $m^0_{iN}, m^1_{iN}$, to order $1/N$. Since the term with $k=0=m$ is independent of the entries $m_{iN}^a$, we focus on the remaining terms. \\

For either $k$ or $m$ different from 0, we need to evaluate
\begin{align*}
    &\frac{1}{N}\Tr\mathbb{E}\left[\mathbf{G}_0(z)[\mathbf{A}^{(0)}(z)\mathbf{G}_0(z)]^k\mathbf{G}_1(\xi)[\mathbf{A}^{(1)}(\xi)\mathbf{G}_1(\xi)]^m\right] =\frac{1}{N}\sum_{\substack{i_1,\ldots,i_{2k+1}=1\\j_1,\ldots j_{2m+1}=1}}^{N-1}\mathbb{E}\bigg[\mathbf{G}_0(z)_{i_1i_2}\frac{m^0_{i_2N}m^0_{i_3N}}{z-m^{0}_{NN}}\mathbf{G}_0(z)_{i_3i_4}\cdots\\
    &\times\frac{m^0_{i_{2k}N}m^0_{i_{2k+1}N}}{z-m^{0}_{NN}}\mathbf{G}_0(z)_{i_{2k+1}j_1} \times
    \mathbf{G}_1(\xi)_{j_1j_2}\frac{m^1_{j_2N}m^1_{j_3N}}{\xi-m^{1}_{NN}}\mathbf{G}_1(\xi)_{j_3j_4}\cdots \frac{m^1_{j_{2m}N}m^1_{j_{2m+1}N}}{\xi-m^{1}_{NN}}\mathbf{G}_1(\xi)_{j_{2m+1}i_1}
    \bigg].
\end{align*}
We first take the average over $m_{iN}^a$ for $a\in\{0,1\}$ (they do not appear in the resolvents). Let's start from $m_{NN}^a$. We have:
\begin{equation}
\mathbb{E} \quadre{\frac{1}{(z-m^{0}_{NN})^{k}} \frac{1}{(\xi-m^{1}_{NN})^{m}}}= \frac{(-1)^{k-1}}{(k-1)!} \, \frac{(-1)^{m-1}}{(m-1)!} \partial_z^{k-1} \partial_{\xi}^{m-1}\mathbb{E} \quadre{\frac{1}{z-m^{0}_{NN}}\frac{1}{\xi-m^{1}_{NN}}}.
\end{equation}
The expectation is over the joint Gaussian distribution of the $m_{NN}^a$,
\begin{equation}
\mathbb{E} \quadre{\frac{1}{z-m^{0}_{NN}}\frac{1}{\xi-m^{1}_{NN}}}= \int \frac{d^2{\bf u}}{\sqrt{4 \pi^2  \text{det} {\bf V}}}\frac{e^{-\frac{1}{2} {{\bf u}^T {\bf V}^{-1} {\bf u}}}}{\tonde{z- \mu_0 -\frac{u^0}{\sqrt N}}\tonde{\xi- \mu_1 -\frac{u^1}{\sqrt N}}}, \quad \quad {\bf V}=\begin{pmatrix}
    v^2_0 & v_h^2\\
    v_h^2 & v^2_1
\end{pmatrix}.
\end{equation}
A simple expansion shows that:
\begin{equation}
\begin{split}
&\mathbb{E} \quadre{\frac{1}{(z-m^{0}_{NN})^{k}} \frac{1}{(\xi-m^{1}_{NN})^{m}}}= \frac{1}{(z-\mu_0)^{k}}\frac{1}{(\xi-\mu_1)^{m}}+ \\
&\frac{1}{N} \quadre{\frac{v^2_0}{2} \frac{(k+1)k}{(z-\mu_0)^{k+2}}\frac{1}{(\xi-\mu_1)^{m}} +\frac{v^2_1}{2} \frac{(m+1)m}{(\xi-\mu_1)^{m+2}}\frac{1}{(z-\mu_0)^{k}}+ v_h^2 \frac{k}{(z-\mu_0)^{k+1}}  \frac{m}{(\xi-\mu_1)^{m+1}}}+ \mathcal{O}\tonde{\frac{1}{N^2}}.
\end{split}
\end{equation}

Therefore:
\begin{equation}\label{eq:Nina}
\begin{split}
&\Tr \mathbb{E}\left[{\bf S}_{k,m}\right] = \quadre{\frac{1}{(z-\mu_0)^k}\frac{1}{(\xi-\mu_1)^m} + \mathcal{O}\tonde{\frac{1}{N}}} \sum_{\substack{i_1,\ldots,i_{2k+1}=1\\j_1,\ldots j_{2m+1}=1}}^{N-1}\mathbb{E}\bigg[\mathbf{G}_0(z)_{i_1i_2}\cdots \mathbf{G}_0(z)_{i_{2k+1}j_1}
    \mathbf{G}_1(\xi)_{j_1j_2}\cdots \mathbf{G}_1(\xi)_{j_{2m+1}i_1}
    \bigg]\\
    &\times\mathbb{E}
    [m^0_{i_2N}m^0_{i_3N}\cdots m^0_{i_{2k}N}m^0_{i_{2k+1}N}m^1_{j_2N}m^1_{j_3N}\cdots m^1_{j_{2m}N}m^1_{j_{2m+1}N}].
    \end{split}
\end{equation}
The second average can be evaluated using Wick theorem, paying attention on whether the contractions involve matrix elements with the same or with different $a=0,1$. Below, we determine the subset of contractions that contribute to leading order in $N$. We begin by discussing some special cases.

\vspace{.3 cm}
Let us focus on the case $k=0$. By Wick theorem, the average $\mathbb{E}
    [m^1_{j_2N}m^1_{j_3N}\cdots m^1_{j_{2m}N}m^1_{j_{2m+1}N}]$ appearing in \eqref{eq:Nina} will be contributed by all possible pairwise contractions of the variables $m^1_{jN}$, each one contributing with a factor of $\Delta^2_1/N$. To each possible Wick contraction, there corresponds a contraction of the indices in the term $ \mathbf{G}_1(\xi)_{j_1j_2}\mathbf{G}_1(\xi)_{j_3j_4}\cdots \mathbf{G}_1(\xi)_{j_{2m+1}i_1}$ also appearing in \eqref{eq:Nina}. 
We now argue that there is a unique Wick contraction that contributes to \eqref{eq:Nina} to leading order, which is the contraction corresponding to $\delta_{j_3 j_4}\cdots \delta_{j_{2m+1} j_2}$. As a matter of fact, as we argue below the products of resolvent operators converge in the large-$N$ limit to a deterministic matrix proportional to the identity. Therefore, each trace of such products is of order $N$. For this reason, to get the largest contribution from the term  $\mathbb{E}\bigg[\mathbf{G}_0(z)_{i_1i_2}\mathbf{G}_0(z)_{i_3i_4}\cdots \mathbf{G}_0(z)_{i_{2k+1}j_1}
\mathbf{G}_1(\xi)_{j_1j_2}\mathbf{G}_1(\xi)_{j_3j_4}\cdots \mathbf{G}_1(\xi)_{j_{2m+1}i_1}
    \bigg]$ in \eqref{eq:Nina}, one has to select the contraction of indices that corresponds to maximizing the number of resulting traces, while recalling that some matrices have common indices and cannot therefore be decoupled into separate traces. For $k=0, m\geq 1$ we see that $\mathbf{G}_0(z)_{i_1j_1}\mathbf{G}_1(\xi)_{j_1j_2}\mathbf{G}_1(\xi)_{j_{2m+1}i_1}$ is the only block which cannot be decoupled. This term is of order $1/N$. Hence in this case the only leading term is given by 
\begin{equation}\label{eq:B8}
\frac{1}{N}\Tr \mathbb{E}\left[{\bf S}_{0,m}(z, \xi)\right]=    \frac{1}{N}\frac{\Delta^{2m}_1}{(\xi-\mu_1)^m}\left(\frac{1}{N}\Tr\mathbb{E}[\mathbf{G}_0(z)\mathbf{G}_1(\xi)^2]\right)\left(\frac{1}{N}\Tr\mathbb{E}\mathbf{G}_1(\xi)\right)^{m-1}+ \mathcal{O} \tonde{\frac{1}{N^2}},
\end{equation}
where we recall that $\Delta^2_a= \Delta^2_h + \Delta^2_{w,a}$. The case $k\geq 1, m=0$ is analogous, we get the leading contribution
\begin{equation}\label{eq:B9}
 \frac{1}{N}\Tr \mathbb{E}\left[{\bf S}_{k,0}(z, \xi)\right]=   \frac{1}{N}\frac{\Delta^{2k}_0}{(z-\mu_0)^k}\left(\frac{1}{N}\Tr\mathbb{E}[\mathbf{G}_0(z)^2\mathbf{G}_1(\xi)]\right)\left(\frac{1}{N}\Tr\mathbb{E}\mathbf{G}_0(z)\right)^{k-1}+ \mathcal{O} \tonde{\frac{1}{N^2}}.
\end{equation}

\vspace{.3 cm}

In the case $k\geq 1, m\geq 1$, the only coupled matrices are the two pairs $\mathbf{G}_0(z)_{i_{2k+1}j_1}\mathbf{G}_1(\xi)_{j_1j_2}$ and $\mathbf{G}_0(z)_{i_1i_2}\mathbf{G}_1(\xi)_{j_{2m+1}i_1}$. A reasoning analogous to the one above shows that the leading term in the $1/N$ expansion is given by:

\begin{align}\label{eq:B10}
\frac{1}{N}\Tr \mathbb{E}\left[{\bf S}_{k,m}\right]&=    \frac{\Delta_h^4}{N}\frac{\Delta^{2k-2}_0}{(z-\mu_0)^k}\frac{\Delta^{2m-2}_1}{(\xi-\mu_1)^m}\left(\frac{1}{N}\Tr\mathbb{E}[\mathbf{G}_0(z)\mathbf{G}_1(\xi)]\right)^2\left(\frac{1}{N}\Tr\mathbb{E}\mathbf{G}_0(z)\right)^{k-1}\left(\frac{1}{N}\Tr\mathbb{E}\mathbf{G}_1(\xi)\right)^{m-1}.
\end{align}
The dependence on  $\Delta_h$ appears due to the fact that the contractions corresponding to $\delta_{i_{2k+1} j_2}$ and $\delta_{i_2 j_{2m+1}}$ involve elements $m^a_{i N}$ corresponding to two different indices $a\in\{0,1\}$.\\
It is straightforward to check that the re-summation of the $1/N$  contributions for arbitrary $k, m$ leads to the expression \eqref{eq:psi00Implicit} in the main text. \\\\

Let us now discuss the partial average of term $\psi_{0N}(z, \xi)$. 
From equations \eqref{psi_decomposition} we see that we have to compute the following:
\begin{align*}
    \psi_{0N}(z,\xi)
    &=\mathbb{E}\left[\frac{2}{N}\frac{1}{(z-m^{0}_{NN})(\xi-m^{1}_{NN})}\sum_{i=1}^{N-1}\sum_{k,l=1}^{N-1}m^0_{kN}m^1_{lN}(z-\mathbf{H}-\mathbf{W}^{(0)}-\mathbf{A}^{(0)}(z))^{-1}_{ik}
    (\xi-\mathbf{H}-\mathbf{W}^{(1)}-\mathbf{A}^{(1)}(\xi))^{-1}_{il}
    \right]\\
\end{align*}
Notice that in this case we have to perform a similar analysis to the one of $\psi_{00}$, with the exception that there are no terms at zeroth-order in $1/N$. In order to obtain the correction at order $1/N$, when using the Wick theorem we need to select the contraction that decouples all the resolvents (i.e., which traces each of them separately) except for the first and last ones, which are clearly coupled. An analogous analysis as the one above then immediately gives us the result \eqref{eq:psi0NImplicit}.\\\\

Finally, let us consider the term $\psi_{NN}(z, \xi)$: from equation \eqref{psi_decomposition} we have to compute the following:
\begin{align*}
    &\psi_{NN}(z,\xi)=\mathbb{E}\left[\frac{1}{N}
    (z-\mathbf{M}^{(0)})^{-1}_{NN}(\xi-\mathbf{M}^{(1)})^{-1}_{NN}\right]\\
    &=\mathbb{E}\Bigg[\frac{1}{N(z-m^{0}_{NN})(\xi-m^{1}_{NN})}\left\{1+\sum_{i,j=1}^{N-1}\mathbf{A}^{(0)}(z)_{ij}(z-\mathbf{H}-\mathbf{W}^{(0)}-\mathbf{A}^{(0)}(z))^{-1}_{ij}\right\}\\
    &\times\left\{1+\sum_{k,l=1}^{N-1}\mathbf{A}^{(1)}(\xi)_{kl}(\xi-\mathbf{H}-\mathbf{W}^{(1)}-\mathbf{A}^{(1)}(\xi))^{-1}_{kl} \right\}
    \Bigg]\\
    &=\mathbb{E}\left[\frac{1}{N(z-m^{0}_{NN})(\xi-m^{1}_{NN})}\right]+\mathbb{E}\left[\frac{1}{N(z-m^{0}_{NN})(\xi-m^{1}_{NN})}\Tr[\mathbf{A}^{(0)}(z)(z-\mathbf{H}-\mathbf{W}^{(0)}-\mathbf{A}^{(0)}(z))^{-1}]\right]\\
    &+\mathbb{E}\left[\frac{1}{N(z-m^{0}_{NN})(\xi-m^{1}_{NN})}\Tr[\mathbf{A}^{(1)}(\xi)(\xi-\mathbf{H}-\mathbf{W}^{(1)}-\mathbf{A}^{(1)}(\xi))^{-1}]\right]\\
    &+\mathbb{E}\left[\frac{1}{N(z-m^{0}_{NN})(\xi-m^{1}_{NN})}\Tr[\mathbf{A}^{(0)}(z)(z-\mathbf{H}-\mathbf{W}^{(0)}-\mathbf{A}^{(0)}(z))^{-1}]\Tr[\mathbf{A}^{(1)}(\xi)(\xi-\mathbf{H}-\mathbf{W}^{(1)}-\mathbf{A}^{(1)}(\xi))^{-1}]\right].
\end{align*}
This expression boils down to computing
\begin{align*}
    &\frac{1}{N}\mathbb{E}\Tr[\mathbf{A}^{(a)}(z)(z-\mathbf{H}-\mathbf{W}^{(a)}-\mathbf{A}^{(a)}(z))^{-1}]=\frac{1}{N}\mathbb{E}\sum_{k=1}^{+\infty}\Tr([\mathbf{A}^{(a)}(z)\mathbf{G}_a(z)]^k)\\
    &=\frac{1}{N}\sum_{k=1}^{+\infty}\frac{\Delta^{2k}_a}{(z-\mu_a)^k}\mathbb{E}\left[\left(\frac{1}{N}\Tr \mathbf{G}_a(z)\right)^k\right]=\frac{1}{N}\frac{\Delta^2_a}{[z-\mu_a-\frac{\Delta^2_a}{N} \Tr \mathbb{E}\mathbf{G}_a(z)]}\left(\frac{1}{N}\Tr\mathbf{G}_a(z)\right)
\end{align*}
with which we finally obtain \eqref{eq:psiNNImplicit}. 

\subsection{Formula for $\psi$}
\label{app:comp_psi}
By combining equations \eqref{eq:psi00Implicit},  \eqref{eq:psi0NImplicit}, \eqref{eq:psiNNImplicit} and using the result in \eqref{eq:multi_resolv_formula}, we obtain the following expressions:
\begin{align}
\begin{split}
    \psi_{00}(z,\xi)&= \Psi(z,\xi)+\frac{1}{N}\Psi^{(1)}(z,\xi)-\frac{1}{N}\partial_z\Psi(z,\xi)\frac{\Delta_0^2}{z-\mu_0-\Delta_0^2\mathfrak{g}_{\sigma}(z)}
    -\frac{1}{N}\partial_\xi\Psi(z,\xi)\frac{\Delta_1^2}{\xi-\mu_1-\Delta_1^2\mathfrak{g}_{\sigma}(\xi)}\\&+ \frac{1}{N}\Delta_h^4 \Psi^2(z,\xi)   \frac{1}{[z-\mu_1-\Delta_0^2\mathfrak{g}_{\sigma}(z)]}\frac{1}{[\xi-\mu_1-\Delta_1^2\mathfrak{g}_{\sigma}(\xi)]}
    + \mathcal{O}\tonde{\frac{1}{N^2}}.
\end{split}
\end{align}
\begin{align}
    \psi_{0N}(z,\xi)=\frac{2}{N} \Delta^2_h\frac{1}{z-\mu_0-\Delta_0^2\mathfrak{g}_{\sigma}(z)}
    \frac{1}{\xi-\mu_1-\Delta_1^2\mathfrak{g}_{\sigma}(\xi)}\Psi(z,\xi) +\mathcal{O}\left(\frac{1}{N^2}\right)
\end{align}
\begin{align}
\begin{split}
    \psi_{NN}(z,\xi)&=\frac{1}{N}\frac{1}{(z-\mu_0)(\xi-\mu_1)}\Bigg[
     1 + \frac{\Delta_0^2}{z-\mu_0-\Delta_0^2\mathfrak{g}_{\sigma}(z)}\mathfrak{g}_\sigma(z)
    +\frac{\Delta_1^2}{\xi-\mu_1-\Delta_1^2\mathfrak{g}_{\sigma}(\xi)}\mathfrak{g}_\sigma(\xi)
    \\
    &+ \frac{\Delta_0^2}{z-\mu_0-\Delta_0^2\mathfrak{g}_{\sigma}(z)} \frac{\Delta_1^2}{\xi-\mu_1-\Delta_1^2\mathfrak{g}_{\sigma}(\xi)}
    \mathfrak{g}_\sigma(z) \mathfrak{g}_\sigma(\xi)
    \Bigg]
    +\mathcal{O}\left(\frac{1}{N^2}\right)
\end{split}
\end{align}

\noindent Summing up we finally get the explicit expression for $\psi$: 
\begin{align}
\label{eq:app:psi}
    \begin{split}
&\psi(z,\xi)=\Psi(z,\xi)+\frac{1}{N}\bigg\{\Psi^{(1)}(z,\xi)+\frac{1}{(z-\mu_0)(\xi-\mu_1)}+\frac{\Delta_0^2}{z-\mu_0-\Delta_0^2\mathfrak{g}_\sigma(z)}\left[-\partial_z\Psi(z,\xi)+\frac{\mathfrak{g}_\sigma(z)}{(z-\mu_0)(\xi-\mu_1)}\right]\\
    &+\frac{\Delta_1^2}{\xi-\mu_1-\Delta_1^2\mathfrak{g}_\sigma(\xi)}\left[-\partial_\xi\Psi(z,\xi)+\frac{\mathfrak{g}_\sigma(\xi)}{(z-\mu_0)(\xi-\mu_1)}\right]+\frac{1}{[z-\mu_0-\Delta_0^2\mathfrak{g}_\sigma(z)][\xi-\mu_1-\Delta_1^2\mathfrak{g}_\sigma(\xi)]}\bigg[\Delta_h^4\Psi^2(z,\xi)\\
&+2\Delta_h^2\Psi(z,\xi)+\frac{\Delta_0^2\Delta_1^2\mathfrak{g}_\sigma(z)\mathfrak{g}_\sigma(\xi)}{(z-\mu_0)(\xi-\mu_1)}\bigg]\bigg\}+ \mathcal{O}\tonde{\frac{1}{N^2}}.
    \end{split}
\end{align}

\section{Computation of the overlaps. }
\label{app:computation_overlaps}

\vspace{1 cm}

With the expression for $\psi$, we can now use equation \eqref{re_psi} to compute the three overlaps expressed in equations \eqref{eq:phi_bulk_bulk}, \eqref{eq:phi_iso_iso}, \eqref{eq:phi_iso_bulk} of the main text.

\subsection{Overlap between bulk eigenvectors}\label{app:computation_phi}
To recover the leading-order term in the overlap between bulk eigenvectors we have to neglect all $1/N$ corrections in Eq.~\eqref{eq:app:psi}. In general, we see that when applying Eq.~\eqref{re_psi}, the square roots in $\psi$ will give rise to branch cuts, contained in $\mathfrak{g}_\sigma$ and its derivatives. In order to face this issue we have to carefully take the limit of $\mathfrak{g}_\sigma(x\pm i\eta)$ when $\eta\to0^+$. Since the branch cuts come from all the terms of the form $\sqrt{x^2-4\sigma^2}$, we have to carefully analyse $\sqrt{(x\pm i\eta)^2-4\sigma^2}$ as $\eta\to 0^+$. As is known, the square root function in the complex plane presents a branch cut, which we fix here to be toward the negative real axis (i.e. we define angles between $[-\pi,\pi]$). With such convention, we simply have that the square root behaves as follows:
\begin{align}
    \lim_{\eta\to 0^+}\sqrt{(x\pm i\eta)^2-4\sigma^2}=
    \begin{cases}
        \sqrt{x^2-4\sigma^2}\quad |x|\geq 2\sigma\\
        \pm\text{sign}(x)i\sqrt{4\sigma^2-x^2}\quad |x|< 2\sigma
    \end{cases}
\end{align}
and by applying this to $\mathfrak{g}_\sigma$, defined in \eqref{eq.Stjlt}, we obtain
\begin{equation}
\label{br_cut}
\lim_{\eta \to 0}\mathfrak{g}_\sigma(x \mp i \eta)=\begin{cases}
  \frac{1}{2 \sigma^2} \tonde{x- \text{sign} (x) \sqrt{x^2-4 \sigma^2}} \quad  |x|> 2 \sigma\\
  \frac{1}{2 \sigma^2}\tonde{x \pm i \sqrt{4 \sigma^2-x^2}} \quad  |x|< 2 \sigma
  \end{cases} \equiv
\mathfrak{g}_R(x) \pm i \mathfrak{g}_I(x),
  \end{equation}
where $\text{Im}\,\mathfrak{g}(x) \neq 0$ only if $|x|< 2 \sigma$.
Similarly, we set:
\begin{equation}
   \lim_{\eta \to 0} \zeta(x \mp i \eta)= x - \sigma_W^2
\mathfrak{g}_R(x) \mp i \sigma^2_W \mathfrak{g}_I(x)= \zeta_R(x)\pm i \zeta_I(x).
\end{equation}

With this notation, the bulk-bulk overlap reads:
 \begin{align}\label{eq:OvBouch}
 \begin{split}
\Phi(x, y)&=  \frac{1}{2\pi^2\rho(x)\rho(y)}\lim_{\eta\to 0^+}\text{Re}\left[\Psi(x-i\eta,y+i\eta)-\Psi(x-i\eta,y-i\eta)\right]\\&=
\frac{ 2 \sigma_W^2 \quadre{\zeta_R(x)-\zeta_R(y)} (x-y) }{\quadre{\tonde{\zeta_R(x)-\zeta_R(y)}^2+ \tonde{\zeta_I(x)+\zeta_I(y)}^2}\quadre{\tonde{\zeta_R(x)-\zeta_R(y)}^2+ \tonde{\zeta_I(x)-\zeta_I(y)}^2}},
\end{split}
\end{align} 
which can be more explicitly rewritten in the form of Eq.~\eqref{eq:phi_bulk_bulk}.

\subsection{Overlap between isolated eigenvectors and bulk eigenvectors}
\label{app:computation_of_phi_y_iso}
In order to compute $\Phi(\lambda_{\rm iso}^0,y)$ one has to make use of  Eq.~\eqref{re_psi} in the main text, and consider only the part of $\psi$ in Eq.~\eqref{eq:app:psi} which presents a singularity when evaluated at $x=\lambda_{\rm iso}^0:=\lambda_{\rm iso, -}^0$, see \eqref{eq:IsoExplicit}. We are therefore focusing on  $|x|=|\lambda_{\rm iso}^0|>2\sigma$ and $|y|\leq 2\sigma$: the first argument of $\psi(x,y)$ does not belong to the bulk of the eigenvalue density, while the second does. It is simple to check that, given the two solutions \eqref{app:solutions},
for $x$ real one has:
\begin{align}\label{eq:res}
    \frac{1}{x-\mu_0-\Delta_0^2\mathfrak{g}_\sigma(x)}= \frac{\tonde{1- \frac{\Delta_0^2}{2 \sigma^2}}x-\mu_0-\text{sign}(x) \frac{\Delta_0^2}{2 \sigma^2} \sqrt{x^2-4\sigma^2}}{\tonde{1-\frac{\Delta_0^2}{\sigma^2}} [x-\lambda^0_{\text{iso},+}(\mu_0, \Delta_0, \sigma)][x-\lambda^0_{\text{iso},-}(\mu_0, \Delta_0, \sigma)]}.
\end{align}
This term is therefore singular for $x \to \lambda_{\rm iso}^0$.  In particular, 
\begin{align}\label{eq:Residuo}
   \lim_{\eta \to 0} \frac{1}{x-i \eta-\mu_0-\Delta_0^2\mathfrak{g}_\sigma(x-i \eta)}=  i \pi \delta(x-\lambda_{\rm iso}^0) \mathfrak{q}_{\sigma,\Delta_0}(\lambda_{\rm iso}^0,\mu_0) + \text{   regular terms},
   \end{align}
where $\mathfrak{q}_{\sigma,\Delta}(\lambda,\mu)$ is given in Eq. \eqref{eq:Defq} while the regular terms are not proportional to the delta. To select the relevant contributions to  $\Phi(\lambda_{\rm iso}^0,y)$, we single out the term in \eqref{eq:app:psi} which produces a delta function when $x \to \lambda_{\rm iso}^0$, which reads 
\begin{align}\label{eq:Aux}
    \hat{\psi}(z,\xi)=\frac{1}{z-\mu_0-\Delta_0^2\mathfrak{g}(z)}\Bigg[-\Delta_0^2\partial_z\Psi(z,\xi)+\frac{\Delta_0^2\mathfrak{g}(z)}{(z-\mu_0)}\frac{1}{(\xi-\mu_1-\Delta_1^2\mathfrak{g}(\xi))}+\frac{\Delta_h^4\Psi^2(z,\xi)+2\Delta_h^2\Psi(z,\xi)}{\xi-\mu_1-\Delta_1^2\mathfrak{g}(\xi)}
\Bigg].
\end{align}
Consider the first term in brackets. One has
\begin{equation}
    \partial_z \Psi(z,\xi)= \frac{\mathfrak{g}_\sigma(z)-\mathfrak{g}_\sigma(\xi)+ \mathfrak{g}'_\sigma(z)(\xi-z)}{[\xi-z -\sigma_W^2 (\mathfrak{g}_\sigma(\xi)-\mathfrak{g}_\sigma(z))]^2},
\end{equation}
and for real $|x|> 2 \sigma$ , given that $\mathfrak{g}_I(x)= \zeta_I(x)=0$, we find
\begin{equation}
\begin{split}
 A(x,y)&:=  \text{Im} \lim_{\eta \to 0} \; \lim_{z \to x-i \eta} \frac{\quadre{\partial_z \Psi(z, y + i \eta)- \partial_z \Psi(z, y - i \eta)}}{\mathfrak{g}_I(y)}\\
 &=\frac{2 [\zeta_R(y)- \zeta_R(x)]^2 -2\zeta^2_I(y) -4 \sigma_W^2 \, \quadre{ \mathfrak{g}_R(x)-\mathfrak{g}_R(y)-(x-y)\mathfrak{g}'_R(x)} \,[\zeta_R(y)-\zeta_R(x)]}{\tonde{\quadre{\zeta_R(y)-\zeta_R(x)}^2 + \zeta_I^2(y)}^2}.
    \end{split}
\end{equation}
The second term in brackets in \eqref{eq:Aux} can be neglected, as its imaginary part is proportional to $\delta(y-\lambda^1_{\rm iso})$ and thus it will only give contributions to the overlap between isolated eigenvectors discussed in the next subsection. The third term instead will contribute with an imaginary part that is not proportional to $\delta(y-\lambda^1_{\rm iso})$. It holds:
\begin{equation}
\begin{split}
& \text{Im} \lim_{\eta \to 0} \; \lim_{z \to x-i \eta}\quadre{\frac{\Psi(z, y + i \eta)}{y+ i \eta -\mu_1-\Delta^2 \mathfrak{g}_\sigma(y+i \eta)}- \frac{\Psi(z, y - i \eta)}{y- i \eta -\mu_1-\Delta^2 \mathfrak{g}_\sigma(y-i \eta)}}=\\
   &\frac{y -\mu_1-\Delta^2 \mathfrak{g}_R(y)}{[y -\mu_1-\Delta^2 \mathfrak{g}_R(y)]^2+ \Delta^4 \mathfrak{g}^2_I(y)} \, \text{Im} \Psi_D(x, y) -\frac{\Delta^2 \mathfrak{g}_I(y)}{[y -\mu_1-\Delta^2 \mathfrak{g}_R(y)]^2+ \Delta^4 \mathfrak{g}^2_I(y)} \, \text{Re} \Psi_S(x, y)
    \end{split}
    \end{equation}
where 
\begin{equation}
\begin{split}
& \Psi_D(x, y)=\lim_{\eta \to 0} \lim_{z \to x-i \eta}\quadre{\Psi(z, y + i \eta)- \Psi(z, y - i \eta)},\\
&\Psi_S(x, y)=\lim_{\eta \to 0} \lim_{z \to x-i \eta}\quadre{\Psi(z, y + i \eta)+ \Psi(z, y - i \eta)},
 \end{split}
\end{equation}
and for $|x|> 2 \sigma$:
\begin{equation}
\begin{split}
    \text{Im} \Psi_D(x, y)&= \frac{ 2 \mathfrak{g}_I(y) \, (y-x)}{[\zeta_R(x)-\zeta_R(y)]^2+ \zeta^2_I(y)},\\
    \text{Re} \Psi_S(x, y)&= \frac{  2\sigma_W^2 \mathfrak{g}^2_I(y) + 2 \, (\mathfrak{g}_R(x)-\mathfrak{g}_R(y))\, (\zeta_R(x)-\zeta_R(y))}{[\zeta_R(x)-\zeta_R(y)]^2+ \zeta^2_I(y)}
    \end{split}
\end{equation}
Therefore 
\begin{equation}
\begin{split}
B(x,y):=& \text{Im} \lim_{\eta \to 0} \lim_{z \to x-i \eta}\, \frac{1}{\mathfrak{g}_I(y)} \quadre{\frac{\Psi(z, y + i \eta)}{y+ i \eta -\mu_1-\Delta^2 \mathfrak{g}_\sigma(y+i \eta)}- \frac{\Psi(z, y - i \eta)}{y- i \eta -\mu_1-\Delta^2 \mathfrak{g}_\sigma(y-i \eta)}}\\
  =&
  \frac{
       2  (y-x)\,  [y -\mu_1-\Delta^2 \mathfrak{g}_R(y)]- 2 \Delta^2\quadre{\sigma_W^2 \mathfrak{g}^2_I(y) + (\mathfrak{g}_R(x)-\mathfrak{g}_R(y)) (\zeta_R(x)-\zeta_R(y))}}{\quadre{\tonde{\zeta_R(x)-\zeta_R(y)}^2+ \zeta^2_I(y)} \; \quadre{\tonde{y -\mu_1-\Delta^2 \mathfrak{g}_R(y)}^2+ \Delta^4 \mathfrak{g}^2_I(y)}}.
    \end{split}
    \end{equation}
Finally, 
\begin{equation}
\begin{split}
 C(x,y):= & \text{Im} \lim_{\eta \to 0} \lim_{z \to x-i \eta} \, \frac{1}{\mathfrak{g}_I(y)}\quadre{\frac{[\Psi(z, y + i \eta)]^2}{y+ i \eta -\mu_1-\Delta^2 \mathfrak{g}_\sigma(y+i \eta)}- \frac{[\Psi(z, y - i \eta)]^2}{y- i \eta -\mu_1-\Delta^2 \mathfrak{g}_\sigma(y-i \eta)}}\\
  = &\frac{y -\mu_1-\Delta^2 \mathfrak{g}_R(y)}{[y -\mu_1-\Delta^2 \mathfrak{g}_R(y)]^2+ \Delta^4 \mathfrak{g}^2_I(y)} \frac{\text{Im} \Psi_{2D}(x, y)}{\mathfrak{g}_I(y)} -\frac{\Delta^2 }{[y -\mu_1-\Delta^2 \mathfrak{g}_R(y)]^2+ \Delta^4 \mathfrak{g}_I^2(y)} \text{Re} \Psi_{2S}(x, y)
    \end{split}
    \end{equation}
where now
\begin{equation}
\begin{split}
& \Psi_{2D}(z, y)=\lim_{\eta \to 0}\quadre{\Psi^2(z, y + i \eta)- \Psi^2(z, y - i \eta)},\\
&\Psi_{2S}(z, y)=\lim_{\eta \to 0}\quadre{\Psi^2(z, y + i \eta)+ \Psi^2(z, y - i \eta)}.
 \end{split}
\end{equation}
Again, for $|x|>2 \sigma$, one finds:

\begin{equation}
\begin{split}
   & \frac{\text{Im} \Psi_{2D}(x, y)}{\mathfrak{g}_I(y) }=\frac{4 (x-y)
\quadre{(\mathfrak{g}_R(x)- \mathfrak{g}_R(y)) \, (\zeta_R(x)- \zeta_R(y))- \sigma_W^2 \mathfrak{g}_I^2(y)} }
     { \quadre{(\zeta_R(x)- \zeta_R(y))^2 +\zeta^2_I(y)}^2 }\\
     &\text{Re} \Psi_{2S}(x, y)=\frac{2  \quadre{(\mathfrak{g}_R(x)- \mathfrak{g}_R(y)) \, (\zeta_R(x)- \zeta_R(y))- \sigma_W^2 \mathfrak{g}_I^2(y)}^2-2 (x-y)^2 \mathfrak{g}^2_I(y)}
     {\quadre{(\zeta_R(x)- \zeta_R(y))^2 +\zeta^2_I(y)}^2 }. 
     \end{split}
    \end{equation}
Combining everything, we find:
\begin{equation}\label{eq:OverlapMisto}
    \boxed{\Phi(\lambda_{\rm iso}^0, y)= \frac{\mathfrak{q}_{\sigma,\Delta}(\lambda_{\rm iso}^0,\mu_0) }{2} \quadre{\Delta^2 A(\lambda_{\rm iso}^0,y) - 2 \Delta_h^2 B(\lambda_{\rm iso}^0,y)- \Delta_h^4 C(\lambda_{\rm iso}^0,y)}}.
\end{equation}
with:
\begin{equation}\label{eq:CostantiMisto1}
    \begin{split}
 A(x,y)=&\frac{2 [\zeta_R(y)- \zeta_R(x)]^2 -2\zeta^2_I(y) -4 \sigma_W^2 \, \quadre{ \mathfrak{g}_R(x)-\mathfrak{g}_R(y)-(x-y)\mathfrak{g}'_R(x)} \,[\zeta_R(y)-\zeta_R(x)]}{\tonde{\quadre{\zeta_R(y)-\zeta_R(x)}^2 + \zeta_I^2(y)}^2},\\
B(x,y)=&
  \frac{
       2  (y-x)\,  [y -\mu_1-\Delta^2 \mathfrak{g}_R(y)]- 2 \Delta^2\quadre{\sigma_W^2 \mathfrak{g}^2_I(y) + (\mathfrak{g}_R(x)-\mathfrak{g}_R(y)) (\zeta_R(x)-\zeta_R(y))}}{\quadre{\tonde{\zeta_R(x)-\zeta_R(y)}^2+ \zeta^2_I(y)} \; \quadre{\tonde{y -\mu_1-\Delta^2 \mathfrak{g}_R(y)}^2+ \Delta^4 \mathfrak{g}^2_I(y)}},\\
 C(x,y)=& \frac{1}{[y -\mu_1-\Delta^2 \mathfrak{g}_R(y)]^2+ \Delta^4 \mathfrak{g}^2_I(y)} \frac{1}{\quadre{(\zeta_R(x)- \zeta_R(y))^2 +\zeta^2_I(y)}^2} \, c(x,y)
 \end{split}
    \end{equation}
and 
\begin{equation}\label{eq:CostantiMisto2}
\begin{split}
 c(x,y)&= \quadre{y -\mu_1-\Delta^2 \mathfrak{g}_R(y)} 4 (x-y)
\quadre{(\mathfrak{g}_R(x)- \mathfrak{g}_R(y)) \, (\zeta_R(x)- \zeta_R(y))- \sigma_W^2 \mathfrak{g}_I^2(y)}\\
&- 2 \Delta^2 
\tonde{  \quadre{(\mathfrak{g}_R(x)- \mathfrak{g}_R(y)) \, (\zeta_R(x)- \zeta_R(y))- \sigma_W^2 \mathfrak{g}_I^2(y)}^2- (x-y)^2 \mathfrak{g}^2_I(y)}.
 \end{split}
    \end{equation}
More explicitly, we can also apply formula \eqref{re_psi} to each term in \eqref{eq:Aux}, and obtain the more explicit formula presented in the main text in Eq.\eqref{eq:phi_iso_bulk}. We recall that formula, and we write explicitly all of its parameters:
\begin{empheq}[box=\widefbox]{align*}
\begin{split}
\Phi(\lambda^0_{\rm iso},y)&=\frac{\mathfrak{q}_{\sigma, \Delta_0}(\lambda_{\rm iso}^0,\mu_0)}{2\pi\rho(y)}\Bigg[\frac{4\Delta_0^2\sigma^2}{\sqrt{[\lambda_{\rm iso}^0]^2-4\sigma^2}}\frac{bc-ad}{c^2+d^2}-4\sigma^2\Delta_h^4\frac{b_1c_1e_1-a_1d_1e_1-a_1c_1f_1-b_1d_1f_1}{(c_1^2+d_1^2)(e_1^2+f_1^2)}\\
&-8\sigma^2\Delta_h^2\frac{b_2c_2e_2-a_2d_2e_2-a_2c_2f_2-b_2d_2f_2}{(c_2^2+d_2^2)(e_2^2+f_2^2)}+\frac{\Delta_0^2\Delta_1^2\mathfrak{g}_\sigma(\lambda_{\rm iso}^0)}{\sigma^2(\lambda_{\rm iso}^0-\mu_0)(y-\mu_1)}\frac{b_3c_3-a_3d_3}{c_3^2+d_3^2}
\Bigg]
\end{split}
\end{empheq}
where the parameters explicitly read:
\begin{equation}
    \begin{split}
&a=(4\sigma^2-\lambda_{\rm iso}^0 y)\text{sign}(\lambda_{\rm iso}^0)\\
&b=\sqrt{[\lambda_{\rm iso}^0]^2-4\sigma^2}\; \sqrt{4\sigma^2-y^2}\\
&c=(2\sigma^2 - \sigma_W^2)^2 (\lambda_{\rm iso}^0 - y)^2 + 
 2\sigma_W^2 (2 \sigma^2 - \sigma_W^2) \sqrt{  [\lambda_{\rm iso}^0]^2-4 \sigma^2} \;(\lambda_{\rm iso}^0 - y)\text{sign}(\lambda_{\rm iso}^0) \\
& + 
 \sigma_W^4 ( [\lambda_{\rm iso}^0]^2-4 \sigma^2 ) + \sigma_W^4 (y^2 -4 \sigma^2 )\\
&d=-2\sigma_W^2\sqrt{4\sigma^2 - y^2}\;
   [\sigma_W^2\sqrt{ [\lambda_{\rm iso}^0]^2 -4 \sigma^2}
     \text{sign}(\lambda_{\rm iso}^0) + (2\sigma^2 - \sigma_W^2) (\lambda_{\rm iso}^0 - y)]
         \end{split}
\end{equation}

\begin{equation}
    \begin{split}
&a_1=(\lambda_{\rm iso}^0 - y)^2 + 2\sqrt{ [\lambda_{\rm iso}^0]^2-4\sigma^2 } \,(y-\lambda_{\rm iso}^0 )\text{sign}(\lambda_{\rm iso}^0) + [\lambda_{\rm iso}^0]^2-4 \sigma^2 +  y^2-4\sigma^2\\
&b_1=-2\sqrt{4 \sigma^2 - y^2} \; [y-\lambda_{\rm iso}^0  + \sqrt{[\lambda_{\rm iso}^0]^2 -4 \sigma^2 } \text{sign}(\lambda_{\rm iso}^0)]\\
&c_1=-2\mu_1\sigma^2-\Delta_1^2y+2\sigma^2y\\
&d_1=\Delta_1^2\sqrt{4\sigma^2-y^2}\\
&e_1=(\sigma_W^2-2\sigma^2 )^2 (\lambda_{\rm iso}^0 - y)^2 + 
 2\sigma_W^2 (2\sigma^2 - \sigma_W^2)\sqrt{[\lambda_{\rm iso}^0]^2-4 \sigma^2 }\, (\lambda_{\rm iso}^0- y)\text{sign}(\lambda_{\rm iso}^0) \\
& \quad + \sigma_W^4 ([\lambda_{\rm iso}^0]^2-4 \sigma^2 ) + \sigma_W^4 ( y^2-4 \sigma^2 )\\ 
&f_1=-2\sigma_W^2 \sqrt{4\sigma^2 - y^2} [(2\sigma^2 - \sigma_W^2) (\lambda_{\rm iso}^0 - y) + \sigma_W^2\sqrt{[\lambda_{\rm iso}^0]^2-4 \sigma^2}\, \text{sign}(\lambda_{\rm iso}^0)]
         \end{split}
\end{equation}

\begin{equation}
    \begin{split}
&a_2=\lambda_{\rm iso}^0-y-\text{sign}(\lambda_{\rm iso}^0)\sqrt{[\lambda_{\rm iso}^0]^2-4\sigma^2}\\
&b_2=\sqrt{4\sigma^2-y^2}\\
&c_2=c_1\\
&d_2 = d_1\\
&e_2=(\sigma_W^2-2 \sigma^2 ) (\lambda_{\rm iso}^0 - y) - \sigma_W^2 \sqrt{[\lambda_{\rm iso}^0]^2-4 \sigma^2}\text{sign}(\lambda_{\rm iso}^0)\\
&f_2=\sigma_W^2\sqrt{4\sigma^2 - y^2}\\
&a_3 = -y\\
&b_3=\sqrt{4 \sigma^2 - y^2}\\
&c_3=y-\mu_1 - \frac{\Delta_1^2 y}{2 \sigma^2}\\
&d_3=\frac{\Delta_1^2\sqrt{4\sigma^2 - y^2}}{2 \sigma^2}
         \end{split}
\end{equation}

\subsection{Overlap between isolated eigenvectors}
\label{app:computation_phi_iso_iso}

In order to compute $\Phi(\lambda_{\rm iso}^0,\lambda_{\rm iso}^1)$ we have to make use of  Eq.~\eqref{re_psi} in the main text, and consider only the part of $\psi$ in Eq.~\eqref{eq:app:psi} which presents a singularity when evaluated at both of the isolated eigenvalues $\lambda_{\rm iso}^a:=\lambda_{\rm iso, -}^a$ for $a\in\{0,1\}$. This term is the one proportional to the product  $(z-\mu_0-\Delta_0^2\mathfrak{g}_\sigma(z))^{-1} (\xi-\mu_1-\Delta_1^2\mathfrak{g}_\sigma(\xi))^{-1}$ in \eqref{eq:app:psi}. We single out such a term, defining:
\begin{align}\label{eq:PsiTilde}
    \begin{split}
&\tilde{\psi}(z,\xi):=
    \frac{1}{[z-\mu_0-\Delta_0^2\mathfrak{g}_\sigma(z)][\xi-\mu_1-\Delta_1^2\mathfrak{g}_\sigma(\xi)]}\bigg(\Delta_h^4\Psi^2(z,\xi)+2\Delta_h^2\Psi(z,\xi)+\frac{\Delta_0^2\Delta_1^2\mathfrak{g}_\sigma(z)\mathfrak{g}_\sigma(\xi)}{(z-\mu_0)(\xi-\mu_1)}\bigg)
    \end{split}
\end{align}

Similarly to \eqref{eq:Residuo}, we have: 
\begin{equation}
\begin{split}
     &\lim_{\eta \to 0} \;\frac{1}{x-i \eta-\mu_0-\Delta_0^2\mathfrak{g}_\sigma(x-i \eta)}\;\frac{1}{y \pm i \eta-\mu_1-\Delta_1^2\mathfrak{g}_\sigma(y \pm i\eta)}=\\
     &\mp \pi^2  \delta(x-\lambda_{\rm iso}^0)  \delta(y-\lambda_{\rm iso}^1) \mathfrak{q}_{\sigma,\Delta_0}(\lambda_{\rm iso}^0,\mu_0)\mathfrak{q}_{\sigma,\Delta_1}(\lambda_{\rm iso}^1,\mu_1)+ \text{  regular terms  },
    \end{split}
\end{equation}
where again we neglect all terms that are not proportional to the product of delta functions. 
The terms within brackets in \eqref{eq:PsiTilde} are real when computed at $x,y \to \lambda_{\rm iso}^a$ due to the fact that $|\lambda_{\rm iso}^a|> 2 \sigma$, for $a\in\{0,1\}$. Therefore, we find:
\begin{equation}\label{eq:uu}
\begin{split}
    & \text{Re}\lim_{\eta\to 0^+}\left[\tilde{\psi}(\lambda_{\rm iso}^0-i\eta,\lambda_{\rm iso}^1+i\eta)-\tilde{\psi}(\lambda_{\rm iso}^0-i\eta,\lambda_{\rm iso}^1-i\eta)\right]=2 \pi^2  \delta(x-\lambda_{\rm iso}^0)  \delta(y-\lambda_{\rm iso}^1) \times \\
    &\times \mathfrak{q}_{\sigma,\Delta_0}(\lambda_{\rm iso}^0,\mu_0)\mathfrak{q}_{\sigma,\Delta_1}(\lambda_{\rm iso}^1,\mu_1) \tonde{\Delta_h^4\Psi^2(\lambda_{\rm iso}^0,\lambda_{\rm iso}^1)
+2\Delta_h^2\Psi(\lambda_{\rm iso}^0,\lambda_{\rm iso}^1)
+\frac{\Delta_0^2\Delta_1^2\mathfrak{g}_\sigma(\lambda_{\rm iso}^0)\mathfrak{g}_\sigma(\lambda_{\rm iso}^1)}{(\lambda_{\rm iso}^0-\mu_0)(\lambda_{\rm iso}^1-\mu_1)}}.
     \end{split}
\end{equation}
Notice that the equation satisfied by the isolated eigenvalues, Eq. \eqref{egval_eqn}, implies that the last term within brackets in \eqref{eq:uu} is equal to $1$. From this, Eq.~\eqref{eq:phi_iso_iso} immediately follows.\\\\
\noindent When $\mu_0=\mu_1$ and $\Delta_0=\Delta_1$,  the typical values of the two isolated eigenvalues are the same and this expression converges to a finite limit, see Sec.~\ref{sec:applications}.

\vspace{1 cm}

\section{The curvature of random landscapes: a random matrix problem}\label{app:RotationLAndscape}

In this Appendix, we give a few additional details on the connection between the Hessians of random Gaussian landscapes and the random matrices discussed in this work. Consider the random function \eqref{eq:Enp} defined on the surface of the $D$- dimensional unit sphere.  The local curvature of $\mathcal{E}[{\bf s}]$ around a configuration ${\bf s}$ is given by the $(D-1) \times (D-1)$ Riemannian Hessian on the sphere, which can be compactly written as:
\begin{equation}\label{eq:Hessian}
\mathcal{H}[{\bf s}]= \Pi_{\tau[{\bf s}]} \, \tonde{\frac{\partial^2 \mathcal{E}[{\bf s}]}{\partial s_i \partial s_j }} \Pi_{\tau[{\bf s}]} -\tonde{\nabla \mathcal{E}[{\bf s}] \cdot {\bf s}} \mathbbm{1}.
\end{equation} 
In this expression, the first term is the matrix of second derivatives of the landscape projected on the tangent plane ${\tau[{\bf s}]}$ to the sphere at the point ${\bf s}$ (the $\Pi_{\tau[{\bf s}]}$ are the corresponding projection operators), while $ \mathbbm{1}$ is the $(D-1) \times (D-1)$ identity matrix. 
Both the projectors and the diagonal term arise from imposing that the landscape is restricted to the sphere. Given that the matrices  \eqref{eq:Hessian} are projected onto the corresponding tangent planes, their components have to be expressed in a basis $\mathcal{B}[{\bf s}]$ which depends on the point ${\bf s}$ itself, since it has to span the tangent plane $\tau[{\bf s}]$ -- defined as the space of unit vectors ${\bf v}$ satisfying  ${\bf v} \cdot {\bf s}=0$.

 The Hessians \eqref{eq:Hessian} are random matrices, and one can characterize explicitly their statistical distribution conditioned to the fact that ${\bf s}$ is a stationary point with a given energy density $\epsilon \sim D^{-1} \mathcal{E}[{\bf s}]$. Consider two such stationary points ${\bf s}_0, {\bf s}_1$ at energy density $\epsilon_0$ and $\epsilon_1$ respectively, conditioned to be at overlap $q$. To describe the statistics of the corresponding Hessians, it is convenient to choose the bases $\mathcal{B}[{\bf s}_a]$ with $a=0,1$  in each tangent plane in such a  way that $(D-2)$ vectors ${\bf e}_{i=1, \cdots, D-2}$ in each basis span the subspace orthogonal to both ${{\bf s}_a}$ (these vectors can be the same in both $\mathcal{B}[{\bf s}_a]$), while the last one ${\bf e}_{D-1}^a$ is the normalized linear combination of  ${\bf s}_0, {\bf s}_1$ that is orthogonal to  ${\bf s}_a$:
 \begin{equation}
 {\bf e}_{D-1}^0=\frac{ {\bf s}_{1}- q\, {\bf s}_{0}}{\sqrt{1-q^2}}, \quad \quad   {\bf e}_{D-1}^1=\frac{ { \bf s}_{0}- q\, {\bf s}_{1}}{\sqrt{1-q^2}}.
 \end{equation}
When expressed in these bases, the two matrices $ D^{-1/2} \, \Pi_{\tau [{\bf s}^a]} \, \tonde{\partial^2_{ij} \mathcal{E}[{\bf s}^a]} \Pi_{\tau[{\bf s}^a]} $ take the form \eqref{eq:MatrixForm}: the GOE blocks correspond to the subspace spanned by the vectors ${\bf e}_{i=1, \cdots, D-2}$, while the special line and column correspond to the direction identified by the basis vectors ${\bf e}_{D-1}^a$. The two GOE blocks are correlated with each others, and their statistics is described by equation \eqref{eq:VarHEssGOE}. The components of the last row and column have instead fluctuations whose strength depends on $\epsilon_0, \epsilon_1$ and $q$, as discussed in the main text.

We set $N=D-1$. The results discussed in this work are derived assuming that the two matrices are expressed in the same basis. In the Hessian case, the two matrices are defined on different spaces (the tangent planes) spanned by different basis vector. It can be checked that in this case, the quantity $\Phi(\lambda^0,\lambda^1)$ defined from \eqref{re_psi} is given by:
\begin{equation}\label{eq:NewOverlap}
  \Phi(\lambda^0,\lambda^1)= N \sum_{i,j=1}^{N}  \mathbb{E} \quadre{ \langle {\bf e}_i^0, \mathbf{u}_{\lambda^{0}} \rangle  \langle \mathbf{u}_{\lambda^{0}},
  {\bf e}_j^0\rangle \langle {\bf e}_j^1,\mathbf{u}_{\lambda^{1}}\rangle \langle \mathbf{u}_{\lambda^{1}},  {\bf e}_i^1\rangle  }.
\end{equation}
Using that ${\bf e}_i^0={\bf e}_i^1$
for $i \leq N-1$, we see that \eqref{eq:NewOverlap} is equivalent to:
\begin{equation}\label{eq:NewOverlap2}
  \Phi(\lambda^0,\lambda^1)= N  \mathbb{E} \quadre{ \tonde{
\langle \mathbf{u}_{\lambda^{0}},  \mathbf{u}_{\lambda^{1}} \rangle - \langle \mathbf{u}_{\lambda^{0}},  {\bf e}_{N}^0\rangle  \langle {\bf e}_{N}^0, \mathbf{u}_{\lambda^{1}}\rangle + 
\langle \mathbf{u}_{\lambda^{0}},  {\bf e}_{N}^0\rangle  \langle {\bf e}_{N}^1, \mathbf{u}_{\lambda^{1}}\rangle }^2}.
\end{equation}
We see that if both $\lambda^a$ are bulk eigenvalues, all the scalar products appearing in \eqref{eq:NewOverlap2} are of order $N^{-{1}/{2}}$ and thus $ \Phi(\lambda^0,\lambda^1)= N  \mathbb{E} [\langle \mathbf{u}_{\lambda^{0}},  \mathbf{u}_{\lambda^{1}} \rangle^2]$ to leading order in $N$. On the other hand, when one of the eigenvalues is isolated, $\lambda^a= \lambda^a_{\rm iso}$, the typical value of $\langle \mathbf{u}_{\lambda^{0}_{\rm iso}},  {\bf e}_{N}^0\rangle$ is of $\mathcal{O}(1)$, see \eqref{eq:ProjVectors}. Therefore, all terms in \eqref{eq:NewOverlap2} are of the same order of magnitude, and to get the eigenvectors overlaps one has to subtract from $\Phi$ the projections of the eigenvectors along the special direction. We remark that the projections  $\langle \mathbf{u}_{\lambda^{0}},  {\bf e}_{N}^0\rangle$  for bulk eigenvalues $\lambda^{0}$ of perturbed matrices are discussed in \cite{noiry2021spectral} for purely additive perturbations.

\end{document}